\documentclass[11pt]{amsart}
\usepackage[mathscr]{euscript}
\usepackage{amsthm}
\usepackage{amssymb}
\usepackage{hyperref}

\newcommand{\m}[1]{{\uppercase {\mathbf {#1}}}}
\newcommand{\sm}[1]{{\mbox{\scriptsize {\uppercase {\bf {#1}}}}}}
\newcommand{\vr}[1]{{\uppercase {\mathcal{#1}}}}

\newcommand{\clo}[1]{{\rm Clo\:\m #1}}
\newcommand{\clon}[2]{{\rm Clo_{#1}\m #2}}
\newcommand{\pol}[1]{{\rm Pol\:\m #1}}
\newcommand{\poln}[2]{{\rm Pol_{#1}\m #2}}
\newcommand{\con}{{\rm Con\:}}
\newcommand{\cn}[1]{{\con\m {#1}}}
\newcommand{\Cn}[1]{{{\bf\con}\m {#1}}}
\newcommand{\Cg}{{\rm Cg}}
\newcommand{\Sg}{{\rm Sg}}

\def\Dj{\mbox{\raise0.3ex\hbox{-}\kern-0.4em D}}

\newcommand{\lb}{\langle}
\newcommand{\rb}{\rangle}
\newcommand{\mt}{\wedge}
\newcommand{\usub}{\subseteq}
\newcommand{\subp}{{\sf {SP}}}
\newcommand{\sub}{{\sf {S}}}

\newtheorem{thm}{Theorem}[section]
\newtheorem{lm}[thm]{Lemma}
\newtheorem{prp}[thm]{Proposition}
\newtheorem{cor}[thm]{Corollary}
\theoremstyle{definition}

\newtheorem{prb}[thm]{Problem}

\newtheorem{exmpl}[thm]{Example}

\newtheorem{df}[thm]{Definition}
\newtheorem{claim}[thm]{\sc Claim}

\begin{document}

\title[SMB algebras II]{SMB algebras II: On the Constraint Satisfaction Problem over Semilattices of Mal'cev Blocks}
\author[P. Markovi\'c]{Petar Markovi\'c}
\address{Department of Mathematics and Informatics\\ 
University of Novi Sad\\ Serbia}
\email[Petar Markovi\'c]{pera@dmi.uns.ac.rs}
\author[M. Mar\'oti]{Mikl\'os Mar\'oti}
\address{Bolyai Institute\\
University of Szeged, Hungary}
\email[Mikl\'os Mar\'oti]{mmaroti@math.u-szeged.hu}
\author[R. McKenzie]{Ralph McKenzie}
\address{Department of Mathematics\\ 
Vanderbilt University\\ Nashville, TN 37240, USA}
\email[Ralph McKenzie]{ralph.n.mckenzie@vanderbilt.edu}
\author[A. Proki\'c]{Aleksandar Proki\'{c}}
\address{Faculty of Technical Sciences\\ 
University of Novi Sad\\ Serbia}
\email{aprokic@uns.ac.rs}
\thanks{
Petar Markovi\'c was supported by the Ministry of Education, Science and Technological Development of the Republic of Serbia (Grants No. 451-03-66/2024-03/200125 and 451-03-65/2024-03/200125) and by the Science Fund of the Republic of Serbia grant no. 6062228. Mikl\'os Mar\'oti's research was partially supported by Project no TKP2021-NVA-09 financed by the Ministry of Culture and Innovation of Hungary from the National Research, Development and Innovation Fund. Aleksandar Proki\'{c}'s research was supported by the Ministry of Science, Technological Development and Innovation (Contract No. 451-03-65/2024-03/200156)
}

\begin{abstract}
We define a class of algebras, the semilattices of Mal'cev blocks (for short, SMB algebras). In a nutshell, these algebras are semilattices in which each element gets blown up into a Mal'cev algebra. We publish for the first time our old proofs that some SMB algebras induce tractable templates of the reprove that the Constraint Satisfaction Problem. Next, we reprove that, in fact, all SMB algebras induce tractable templates of the Constraint Satisfaction Problem, a result already proved by A. Bulatov. Also, we compare the two general proofs of the CSP Dichotomy and prove they are more similar than initially thought when they are applied to SMB algebras. This paper is the second in the series of papers investigating the SMB algebras and it is a precursor to our further research on the similarities between the proofs of the Dichotomy Theorem.
\end{abstract}

\maketitle

\section{Introduction, history and motivation}

This paper concerns the computational complexity of the {\em Constraint Satisfaction Problem} (CSP). CSP is a decision problem, and it has several equivalent formulations. The versions we will consider in this paper (the variable-value version and the multisorted version) are defined in Section 2. We prove that certain algebraic restrictions guarantee that the problem is tractable.

\subsection{History} 

The interest in CSP has its roots in Descriptive Complexity. In the seminal paper by Feder and Vardi \cite{FV}, the complexity of the fixed-template CSP was proved to be computationally equivalent up to random-P reductions to the class of model-checking problems defined by monotone monadic SNP formulae without inequalities (MMSNP). Later the reduction was derandomized by Kun \cite{K}, so the two classes of problems are computationally equivalent (they traverse the same complexity classes). Feder and Vardi, also in \cite{FV}, conjectured that the CSP can have only two complexities, tractable and NP-complete. Since the reduction by Kun, this Dichotomy Conjecture also implies that model-checking the MMSNP formulae can have only these two complexities.

In \cite{J}, Jeavons proved that the complexity of the CSP depends on the algebra of compatible operations (polymorphisms) of the template model. Subsequently, this was sharpened by Bulatov, Jeavons and Krokhin \cite{BJK}, who proved that, in fact, the pseudovariety generated by the algebra of idempotent polymorphisms of the core controls the complexity. In the same paper, they also proved that if there is no polymorphism which satisfies a set of Taylor equations, then the problem is NP-complete and conjectured that in the converse situation, CSP is tractable. This second conjecture, which would not only imply the Dichotomy Conjecture, but would delineate the bound between P and NP-complete cases, was known as the Algebraic Dichotomy Conjecture. The Algebraic Dichotomy Conjecture was finally settled positively, independently by Bulatov (\cite{Buconf} and \cite{Budich}) and by Zhuk (\cite{Zhconf} and \cite{Zhdich}), thus also confirming the original Dichotomy Conjecture.

Though the Dichotomy Conjecture is resolved, both proofs are very complicated, and it is desirable to find simplifications. This is not just for its own sake, the CSP Dichotomy could lead to major generalizations (e.g. finding ``Logic for $P$'' in Descriptive Complexity), if only its proof was simpler and more versatile. For a more detailed list of other possible applications of CSP, see, e.g. the final section of \cite{Bulect}.

In this paper we will look at a special case of finite algebras with a Taylor operation, equivalently, a weak near-unanimity operation (see \cite{MM}). The first and the third author invented these algebras, which we call the SMB algebras, in an unpublished note \cite{RP}. Our hope at the time was to prove CSP over SMB algebras is tractable, and then to generalize from SMB algebras to Taylor algebras. We were only partially successful in \cite{RP}, proving the tractability in the case when the underlying semilattice is either linearly ordered, or a flat semilattice. The second author, in another unpublished note \cite{M2}, improved to the case when the underlying semilattice is ordered as a rooted tree. In the present paper we unite these two notes and publish them for the first time. Subsequently, Bulatov managed to prove in \cite{BuSMB} that CSP over any SMB algebra is tractable, and shortly later used these ideas to prove the Dichotomy Theorem. We fix a gap we discovered in his proof is Sections 5 and 6 in two different ways.

\subsection{Motivation} 

The reason for our focus on SMB algebras is that they exhibit several ``bad'' properties, thus being in the ``worst case'' (in the sense of Tame Congruence Theory) of algebras with a Taylor operation. On the other hand, Barto and Kozik's famous Absorption Theorem, a major structural property of binary relations compatible with Taylor algebras, is much easier to prove in SMB algebras than in the general Taylor case. Moreover, ``binary absorption'', one of the main cases considered by Zhuk in his proof of the Dichotomy, trivially appears all over SMB algebras. So SMB algebras seem to be a good special case to consider before trying to prove any result about Taylor algebras. In a companion paper \cite{SMB1} (authored by the first, third and fourth author of this paper, and P. \Dj api\'{c}) we prove some universal algebraic results about SMB algebras and some properties of compatible relations which may prove useful for our overall plan.

A few words about relation with other results. The algorithm of Theorem~\ref{lintract} is inspired by, and generalizes, a special case considered in \cite{3el}. On the other hand, our old results had applications later on, as follows: Of our two main results in Section 3, Theorem~\ref{lintract} was applied in \cite{M2} to solve the CSP over SMB algebras whose $\sim$-classes are tree-ordered, see Corollary~\ref{treetract}. More importantly, the proof of Theorem~\ref{flattract} contains the idea of strands which later evolved into A. Bulatov's coherent sets (see Definition~\ref{coherentdef}), and also a primitive application of link partitions to the strands was developed by A. Bulatov in a much deeper way in his proofs of tractability of SMB algebras and of the Dichotomy Conjecture. As for Section 4 which contains the results from the previously unpublished note \cite{M2}, this section culminates in Corollary~\ref{SMBeliminated} and its application, Theorem~\ref{M-irrweaker} (and its stronger variant, Theorem~\ref{M-irrdesired}, which we prove in Section 5). All of those are key parts of the proof of tractability of SMB algebras, and the consistent maps are also used by A. Bulatov for the proof of the Dichotomy Conjecture.

We draw inspiration for defining the SMB algebras from Tame Congruence Theory as developed in the monograph \cite{tct} and also from a series of results of Bulatov, in papers \cite{Bulatovcons}, \cite{BulatovGraph1}, \cite{BulatovGraph2} and \cite{BulatovGraph3}. 

Namely, if we assume that every two element subset is a subalgebra, the paper \cite{Bulatovcons} has distinguished between three types of behavior of the polymorphisms on the two-element subsets of $A$, the semilattice, the near-unanimity and the minority (Mal'cev). In subsequent papers  \cite{BulatovGraph1}, \cite{BulatovGraph2} and \cite{BulatovGraph3}, Bulatov developed a theory of two-element subsets even without the assumption that each such subset is a subalgebra, where he considers two-element subsets of factor algebras (``thick edges") which separate points and which help him classify local behavior of finite algebras in a similar way as Tame Congruence Theory. In \cite{BulatovGraph3}, he proved that, if there are no Mal'cev edges in the resulting graph of the algebra, then the corresponding CSP has {\em bounded width} and therefore it is tractable, according to \cite{bk} and/or \cite{barto-collapsewidth}. Also in \cite{BulatovGraph3}, A. Bulatov proved that, if there are no semilattice edges, then the corresponding CSP has {\em few subpowers} (see \cite{BIMMVW}) and is therefore tractable by \cite{IMMVW}.

As for Tame Congruence Theory, we know that if an idempotent algebra with a weak near-unanimity operation has no Mal'cev factors of subalgebras, then the CSP has bounded width, while the absence of the semilattice type, together with a technical condition (the absence of tails in the minimal sets) means that the algebra of polymorphisms generates a congruence modular variety, and this implies that the corresponding CSP has few subpowers, according to \cite{Barto-Valerioteconj}. 

Hence, in both Bulatov's and Tame Congruence Theoretic approaches, absence of either the semilattice or Mal'cev case implies the tractability of the corresponding CSP is easy, either always, or in a major subcase. So, we need to consider algebras which exhibit both the semilattice and Mal'cev local behavior. In that direction, the second author proved in \cite{MMontop} that if there is a congruence such that its quotient algebra has a Mal'cev operation, while each congruence class, viewed as a subalgebra, contains no Mal'cev factor of a subalgebra, then the CSP is tractable. So the natural case to consider are algebras which have a congruence modulo which the algebra is a semilattice, while each block of that congruence is a Mal'cev algebra. Add a condition that ensures the operations don't interfere with each other, and we get SMB algebras.

\subsection{Overview}

The paper is organized as follows: Section 2 deals with necessary background and defines the notation. Section 3 defines the class of algebras we consider, semilattices of Mal'cev blocks (SMB algebras), recalls the results on SMB algebras we obtained in the companion paper \cite{SMB1} and solves the CSP in two subclasses, when the underlying semilattice is linearly ordered, or flat. The results of Section 3 were initially written up in the unpublished note \cite{RP}. Section 4 solves the CSP over SMB algebras which are tree-ordered. The results of Section 4 were initially written up in the unpublished note \cite{M2}. After writing up these old results, Section 4 ends with the statement of Theorem~\ref{M-irrdesired}, whose proof in A. Bulatov's paper \cite{BuSMB} had a gap and Theorem~\ref{M-irrweaker}, a weaker form of Theorem~\ref{M-irrdesired}, which is what the arguments up to that point prove. Section 5 recalls some details from the proofs of the CSP Dichotomy by Zhuk and uses these to plug the gap in Theorem~\ref{M-irrdesired}. Our proof has the drawback that it uses the full power of the proof of a more general result to fix a gap in a less general result. In Section 6 we fix Bulatov's proof of tractability of CSP over SMB algebras in a way that is fully independent of other results. Namely, Theorem~\ref{M-irrweaker} is shown to be sufficient for making Bulatov's original proof work, once one has restricted the scope of subinstances to just the maximal-sized non-Mal'cev domains of variables. The non-Malcev smaller domains and the Mal'cev domains of variables can be dealt with by a small restatement of some of Bulatov's definitions. Section 7 contains concluding remarks and open problems.

In Sections 3 and 4 we left our original unpublished notes \cite{RP} and \cite{MMontop} virtually unchanged. The only mathematical changes we made in light of the new results was to change the definition of the class of algebras we look at in Theorem~\ref{lintract}, from a more special class we initially considered to all SMB algebras. The proof is unchanged by this, it works exactly the same way. We also added some corollaries near the end of Section 4 to better explain how those old results can be used.

\section{Background and Notation}

\subsection{Universal Algebra}

We assume that the reader is familiar with the basics of Universal Algebra. The readers who need these facts and definitions are referred to classic textbooks \cite{burris-sank} and \cite{alvin}. Moreover, we need some basic results of Tame Congruence Theory, an advanced theory in Universal Algebra which was developed in \cite{tct}.

Following \cite{alvin}, we use the notation $\clo a$ and $\clon{n}{a}$ for the clone of term operations and the set of $n$-ary term operations of the algebra $\m a$, respectively. Also, $\pol a$ and $\poln{n}{a}$ denote the clone of polynomial operations and the set of $n$-ary polynomial operations of the algebra $\m a$, respectively.

An operation $f$ of an algebra $\m a$ is said to be idempotent if the identity $f(x,x,\ldots,x)\approx x$ holds in $\m a$. An algebra is idempotent if all of its fundamental operations (equivalently, term operations) are idempotent. We will assume all algebras are idempotent and finite in this paper. Furthermore, we will introduce one more restriction on the algebras we consider and one construction on algebras, which we will use at will. The reason is that we are interested in application of algebras to the Constraint Satisfaction Problem, so any constructions which make the Constraint Satisfaction Problem no easier will be allowed.

The third restriction on the algebras under consideration is that we are interested in algebras which generate varieties that omit type {\bf 1} covers (in the language of Tame Congruence Theory). Equivalently, according to \cite{MM}, we may assume for any algebra $\m a$ that there exists a term $w$ such that the identities
$$w(x,x,\ldots,x,y)\approx w(x,x,\ldots,x,y,x)\approx \ldots \approx w(y,x,x,\ldots,x)$$
hold in $\m a$. Those identities, plus idempotence (which we need not assume again) make $w$ a {\em weak near-unanimity} term of $\m a$, or a wnu term for short. We introduce the notation $x\circ_w y$ for the binary term operation $w(x,x,\ldots,x,y)$. A wnu term $w$ of $\m a$ is {\em special} if $\m a$ satisfies the additional identity $x\circ_w(x\circ_w y)\approx x\circ_w y$. Any finite algebra which has a wnu term must also have a special wnu term obtained by iterated composition of the original wnu term (cf. \cite[Lemma 4.7]{MM}).

When we have a finite idempotent algebra $\m a$ with a wnu term $w$ and $\m a'$ is another algebra on the same set with a wnu term such that each fundamental operation of $\m a'$ is a term operation of the algebra $\m a$, we may start considering $\m a'$ instead of $\m a$ when this suits our purpose. We will justify this in the next subsection.

Finally, if $\m a$ is an algebra, $p\in\poln{1}{a}$ and $p(p(x))=p(x)$ for all $x\in A$, then the polynomial $p$ is called a retraction. The retraction $p$ induces an algebra $p(\m a)$ on the set $p(A)$ with the same similarity type as $\m a$ in the following way: If $f$ is an $n$-ary operation symbol then $f^{p(\sm a)}(a_1,\dots,a_n) = p(f^{\sm a}(a_1,\dots,a_n))$. It is easy to verify that, if $\m a$ is idempotent, finite and/or has a wnu term, then $p(\m a)$ has the same properties. The algebra $p(\m a)$ is called a {\em retract} of $\m a$. We will use retraction in Section 4, but it is not a construction which reduces to an equivalent, or a more difficult, problem, like all the previous ones. We will have to prove under which conditions we can work with retracts and what can we conclude about the original instance when we solve the instance over its retracts.

We proceed with an introduction to the Constraint Satisfaction Problem to justify these three restrictions and the constructions.

\subsection{CSP} The Constraint Satisfaction Problem (CSP for short) has several equivalent definitions, each offering slightly different language to express the same thing. We find it most convenient to define it similarly as in \cite{barto-collapsewidth}.

\begin{df}\label{constraintlanguage}
A relation on the set $A$ is a subset of $A^n$ for some positive integer $n$. Here $n$ is the {\em arity} of that relation. By a {\em constraint language} we mean a set $\Gamma$ of relations (of any arities) on the same nonvoid base set $A$.
\end{df}

\begin{df}\label{csp}
Given a constraint language $\Gamma$ on the set $A$, we define an instance of $CSP(\Gamma)$ as any ordered triple $(V,A,\vr c)$ where $V$ is called the set of variables, while $\vr c$ is a set whose elements are called constraints. Each constraint is an ordered pair $(S,R)$, where $S\subseteq V$ is the constraint scope, while $R\subseteq A^S$ is such that there exist an integer $n$ and a surjective mapping $\varphi:n\rightarrow S$ such that $R\circ \varphi\in\Gamma$. Here $R\circ \varphi =\{g\circ \varphi:g\in R\}$. $R$ is the constraint relation of the constraint $(S,R)$. A mapping $f:V\rightarrow A$ is a solution to the instance $(V,A,\vr c)$ of $CSP(\Gamma)$ if for every constraint $(R,S)\in \vr c$, $f|_S\in R$.
\end{df}

In Section 3, we will assume that the variable set $V$ is $n=\{0,1,\dots,n-1\}$, to have a linear order on the variables. We will sometimes write $[n]$ instead of $n$ to remind the reader that $n$ is a set, with elements, subsets etc.

We may assume, without loss of generality, that different constraints have different scopes. This is because we may intersect all constraint relations with the same scope and replace all of those constraints with the single constraint. The way we defined $CSP(\Gamma)$ allows us to also assume, without loss of generality, that $\Gamma$ is closed under intersection, and also under permutation and identification of coordinates. Otherwise, our definition is the same as in \cite{barto-collapsewidth}.

\begin{df}\label{klminimaldef}
Let $P$ be an instance of CSP$(\Gamma)$. We say that $P$ is $(k,l)$-minimal if
\begin{itemize}
\item for any $S\usub V$, $|S|\leq l$, there is precisely one $i$ such that $S_i=S$ ($l$-density) and
\item whenever $(S,R)$ and $(S',R')$ are constraints such that $S'\usub S$ and $|S'|=k$, then $R'=R|_{S'}$ ($k$-consistency).
\end{itemize}
\end{df}

By using the $(k,l)$-minimality algorithm (see e.g. \cite{barto-collapsewidth}), for any given fixed numbers $k\leq l$ we can transform, in polynomial time, a given instance of $CSP(\Gamma)$ into an equivalent $(k,l)$-minimal instance. Note that the second condition implies that each constraint relation $R$ is a subdirect product of $R_j$, $j\in S$, where $(\{j\},R_j)$ are constraints in $\vr c$. Whenever an instance is at least $(1,1)$-minimal, $R_j$ will be fixed and we will denote it by $A_j$.

For the last three decades or so, the main focus in the investigation of computational complexity of the Constraint Satisfaction Problem was the Dichotomy Conjecture of Feder and Vardi, stated in \cite{FV}, which said that the complexity of $CSP(\Gamma)$ can be either tractable, or NP-complete. Our paper consists of old never-published results, which were partial results aimed at proving the Dichotomy Conjecture and influenced Bulatov's proof \cite{Budich}, and some new results which compare the two proofs of the Dichotomy Conjecture and suggest a direction for their simplification. Note that Feder and Vardi postulated the conjecture only when $\Gamma$ is finite, but we do not make that requirement here. On the other hand, we do assume that the domain $A$ is finite throughout this paper. According to \cite{J}, the complexity of the $CSP(\Gamma)$ depends on the compatible operations (polymorphisms) of the relational structure $(A,\Gamma)$. Therefore, it suffices to verify the conjecture for $\Gamma=\subp_{fin}(\m a)$ for some finite algebra $\m a$. More precisely (to conform to our definitions) we may assume that $\Gamma=\bigcup\limits_{n=1}^\infty\sub(\m a^n)$ and we will denote such $\Gamma$ by $\Gamma(\m a)$. The impact of this assumption on $\Gamma$ is that now $\Gamma$ contains all full finite powers of $A$ and that $\Gamma$ is closed under intersections and products (along with permutations and identifications of variables assumed earlier), thus $\Gamma$ is a {\em relational clone}.

This justifies our construction where we move from an algebra $\m a$ to one of its term reducts $\m a'$. Namely, any relation compatible with all operations of $\m a$ is compatible with all term operations of $\m a'$, as well. Therefore, any instance of $CSP(\m a)$ is an instance of $CSP(\m a')$, so if we can solve any instance of $CSP(\m a')$ in polynomial time, we can do the same with instances of $CSP(\m a)$.

We may assume that the relational structure $(A,\Gamma)$ has no endomorphisms except for automorphisms. Namely, if $\varphi$ is an endomorphism which maps $A$ onto a proper subset, then the set of instances of $CSP(\Gamma)$ which have a solution is precisely the same as the set of instances of $CSP(\Gamma|_{\varphi(A)})$ on the domain $\varphi(A)$ which has a solution (in the nontrivial direction, just compose the solution of $CSP(\Gamma)$ with the endomorphism). Such relational structures for which all endomorphisms are automorphisms are called {\em cores}.

For $(A,\Gamma)$ a core, the complexity of $CSP(\Gamma)$ equals the complexity of $CSP(\Gamma^c)$, which is $\Gamma$ augmented with all one-element unary relations. If $\Gamma=\Gamma(\m a)$, then $\Gamma^c=\Gamma(\m a^{id})$, where $\m a^{id}$ is the idempotent reduct of $\m a$, i.e. $\m a^{id}$ is an algebra whose operations are all idempotent term operations in $\clo{a}$. So, we justified the focus on idempotent finite algebras and assume from now on that all algebras are such. For more details about the reductions in this and the previous paragraph, cf. Theorems 4.4 and 4.7 of \cite{BJK}.

As we mentioned in the Introduction, if $\m a$ has no wnu term then, according to \cite{BJK} and \cite{MM}, $CSP(\Gamma(\m a))$ is NP-complete. In \cite{BJK} it is conjectured that in the converse case the $CSP(\Gamma(\m a))$ is in P. A proof of this, the Algebraic Dichotomy Conjecture, in the papers \cite{Zhdich} and \cite{Budich} thus also confirmed the original Dichotomy Conjecture by Feder and Vardi.

In our paper we will use, as a black box, that $CSP(\Gamma(\m a))$ is tractable when $\m a$ is an algebra with Mal'cev term $d(x,y,z)$, i.e. when the identities $d(x,x,y)\approx y\approx d(y,x,x)$ hold in $\m a$. The proof of this was published in \cite{BulatovMalcev} and a shorter one in \cite{BulatovDalmau}.

\subsection{Multisorted CSP and templates} Already when we reduce an instance of $CSP(\m a)$ to a $(k,l)$-minimal one, we introduced subuniverses $A_i$ of $\m a$ for $1\leq i\leq n$. Instead of $R\leq A^S$, we can consider $R$ as a subdirect product of $\{\m a_i:i\in S\}$. For some reductions we will also need $A_i$ to be homomorphic images of $\m a$ (or of its subalgebra), or a retract of $\m a$ (or of its subalgebra).

Now we define a {\em template} of CSP.

\begin{df}\label{CSPtemplate}
A class $\vr t$ of isomorphism types of similar finite algebras is called a CSP template if $\vr t$ is closed under homomorphic images, subalgebras and unary polynomial retracts. 
\end{df}

Note that we do not allow isomorphic algebras to appear more than once in a template, which is why we speak of isomorphism types, rather than algebras themselves. The reason is, if the language of a finite algebra $\m a$ is finite, then there are only finitely many algebras, up to isomorphism, which can be obtained from $\m a$ by taking subalgebras, homomorphic images and retracts. This will allow us, in Section 4, to reduce a multisorted CSP instance by making the template a smaller set of isomorphism types. Now we define the multisorted CSP instance.

\begin{df}\label{multisortedCSP}
Given a template $\vr t$, a (multisorted) instance $P$ of $CSP(\vr t)$ is the triple $(V, D, \vr c)$. Here $D=\{\m a_i:i\in V\}$, where each $\m a_i$ is isomorphic to some algebra in $\vr t$, is the tuple of domains, while $\vr c$ is the set of constraints. Each constraint is an ordered pair $(S,R)$, where $S\subseteq V$, while $R\leq \prod\limits_{i\in S}\m a_i$ is a subdirect product. A mapping $f\in\prod\limits_{i\in V}A_i$ is a solution to the instance $(V,D,\vr c)$ of $CSP(\vr t)$ if for every constraint $(S,R)\in \vr c$, $f|_S\in R$.
\end{df}

Of course, we can take as the template $\vr t(\m a)$, the set of isomorphism types of all finite algebras which can be obtained from a finite algebra $\m a$ by taking homomorphic images, subalgebras and retracts. Then each $(1,1)$-minimal instance of $CSP(\m a)$ is an instance of $CSP(\vr t(\m a))$. Moreover, if $\m a$ has a finite language, then $\vr t(\m a)$ is finite, as we said above.

A construction used repeatedly in the $(k,n)$-minimality algorithm, and also in many procedures in our paper, is the {\em tightening} of an instance $P=(V,D,\vr c)$ of a template $\vr t$. If $D=\{\m A_i:i\in V\}$ and $B_i\subseteq A_i$, we say that the instance $P'=(V,D',\vr c')$ is the tightening of $P$ to $D'=\{B_i:i\in V\}$ if $\vr c'=\{(S,R'):(S,R)\in \vr c\}$ and for each $(S,R)\in \vr c$, $R'=R\cap\prod\limits_{i\in S}B_i$. To ensure that $P'$ is also an instance of $\vr t$, it suffices for all $i\in V$ to have $B_i$ as a subuniverse of $\m a_i$, and indeed, this is what we will always ensure whenever we tighten an instance. Moreover, for any tightening $P'$ of an instance $P$ which will be used in this paper, we will prove that $P'$ has a solution iff $P$ does and in such a case we will say that $P$ {\em can be tightened}.

Another construction which we will repeatedly use is be the {\em restriction} $P|_W$ of an instance $P=(V,D,\vr c)$ to a subset $W\subseteq V$. $P|_W$ is the instance $(W,D',\vr c')$, where $D'=\{\m a_i:i\in W\}$, while $\vr c'=\{(S',R'):(S,R)\in \vr c\}$, where $S'=S\cap W$, while $R'=R\cap\prod\limits_{i\in S'}A_i$. Bear in mind that for several $(S_i,R_i)\in\vr c$, the intersection of $S_i$ with $W$ can be the same, so we will actually silently take the intersection of all corresponding $R_i$ to make $R'$, but this detail will never cause problems so we simply gloss over it.

Another way of encoding a CSP instance was used by Zhuk in his proof of the Dichotomy Conjecture. Take a multisorted instance $(V,D,\vr c)$ over the template $\vr t$. We construct two hypergraphs, one over the set of vertices $V$ (the constraints graph), and the other over the union of the domains $\bigcup D=\bigcup\{A_i:i\in V\}$ (the microstructure graph). For each constraint $(S,R)$ in $\vr c$, we add one hyperedge $S$ in the constraints graph, and for each tuple $f:S\rightarrow\bigcup D$ in $R$ we add a hyperedge $\{f(s):s\in S\}$ into the microstructure graph. Each hyperedge of the microstructure graph intersects each $A_i$ in one point or not at all (depending on whether $i\in S$).

\subsection{Tame Congruence Theory} We will work on $CSP(\m a)$ for certain $\m a$ which have a wnu term. The tame congruence theory analyzes finite algebras according to local behavior of the polynomial operations of the algebra. The polynomial operations are term operations in which some of the variables may have been substituted by fixed constants. The tame congruence theory classifies the covers in the congruence lattice of a finite algebra by first finding a minimal image of an unary polynomial which distinguishes the two congruences (this is an $(\alpha,\beta)$-minimal set). The unary polynomial whose image is the minimal set can be assumed to be idempotent, and if $\alpha\prec \beta$ in the congruence lattice $\cn a$, the set of all $(\alpha,\beta)$-minimal sets is denoted by $M_{\m a}(\alpha,\beta)$. The structure of the minimal set together with the polynomial operations of the algebra which are compatible with that minimal set depend only on the congruences which constitute the cover. It may be of five types, unary (type {\bf 1}), affine (type {\bf 2}), Boolean (type {\bf 3}), lattice (type {\bf 4}) and semilattice (type {\bf 5}). We will write $\alpha\prec_i\beta$ when the congruences $\alpha$ and $\beta$ of a finite algebra form a covering pair in the congruence lattice, and the type of that cover is $i$. The absence of the unary type not only in the finite algebra $\m a$, but also in any finite algebra in the variety $\m a$ generates, is equivalent to the existence of the wnu term in $\m a$. Thus we may assume that 
o type {\bf 1} cover occurs.

If $\alpha\prec\beta$ in $\Cn a$, and the type of that cover is {\bf 5}, then in any minimal set $U$, one $\beta$-class restricts to $U$ as $B$, which intersects exactly two $\alpha$-classes, while $\alpha$ and $\beta$ restrict the same way to $U\setminus B$. We call $B$ the body of $U$. The polynomials of $\m a$ restricted to $B/(\alpha|_B)$ are those of a two-element semilattice. Moreover, for any $\beta$-class, the transitive closure of the semilattice orders coming from all bodies of minimal sets is a connected partial order $\leq$ of $\alpha$-classes inside that $\beta$-class, and $\leq$ is compatible with $\m a/\alpha$ (i.e. a subuniverse of $(\m a/\alpha)^2$).

We proved in the companion paper \cite{SMB1} that when $\alpha\prec_2\beta\prec_5\gamma$ in the congruence lattice of a finite algebra with a Taylor (equivalently, wnu) term, then we can find a subalgebra of a term reduct of $\m a/\alpha$ which has a Taylor term and an interesting structure. The subalgebra is contained in one $\gamma$-class and consists of more than one $\beta$-class. We called such algebras semilattices of Mal'cev blocks, SMB algebras for short, and we will define them in the next section.

\section{Semilattices of Mal'cev Blocks}

We begin by defining the main object of study of this paper and recalling several results about it from \cite{SMB1}.

\begin{df}\label{SMBdef}
Let $\m A=(A; \mt ,d)$ be an idempotent algebra and ${\sim}\in\cn a$. We say that $\m a$ is a semilattice of Mal'cev blocks with respect to ${\sim}$, SMB algebra over ${\sim}$ for short, if
\begin{enumerate}
\item $(A/{\sim};\mt^{\sm a/{\sim}})$ is a semilattice and
\item on each ${\sim}$-class $D$, the operation $\mt|_{D}$ acts as the first projection, while $d|_{D}$ acts as a Mal'cev operation.
\end{enumerate}
We say that $\m A=(A; \mt ,d)$ is a semilattice of Mal'cev blocks, SMB algebra for short, if $\m A$ is an idempotent algebra such that there exists a congruence ${\sim} \in \cn A$ so that $\m a$ is an SMB algebra over ${\sim}$. We denote the class of all SMB algebras by $\vr S$.
\end{df}

We note in passing that we changed the condition $(2)$ from paper \cite{SMB1}, in there $\mt$ acted as the second projection on all $\sim$-classes. We did it to be more in line with A. Bulatov's paper \cite{BuSMB}, which will be very important to us in Section 6. On the other hand, we keep having $\m a/{\sim}$ as meet-semilattices, like in our paper \cite{SMB1}, while A. Bulatov uses join-semilattices in \cite{BuSMB}.

\begin{thm}[Theorem 15 and Proposition 17 of \cite{SMB1}]\label{SMBvar}
The class $\vr s$ of all SMB algebras is a variety with a Taylor term.
\end{thm}

\begin{df}\label{regSMBdef}
We say that an SMB algebra $\m a=\lb A;\mt,d\rb$ over ${\sim}\in\cn a$ is {\em regular} if 
\begin{enumerate}
\item\label{property1} for all $a,b,c\in A$, $[d(a,b,c)]_{\sim}=[(a\mt b)\mt c]_{\sim}$,
\item for all $a,b\in A$ such that $[b]_{\sim}\geq[a]_{\sim}$, $a\mt b=a$,
\item $\m a\models d(x,y,z)\approx d(x\mt (z\mt y),y\mt (z\mt x),z\mt (y\mt x))$, and
\item $\m a\models x\mt (x\mt y)\approx x\mt y$.
\end{enumerate}
\end{df}

It can be proved (see \cite{SMB1}, Lemma 19) that any regular SMB algebra satisfies also the following identity:
\[\tag{5}\m a\models x\mt y\approx d(y,y,x)\approx d(x,y,y).\]
Therefore, in regular SMB algebras, the clone of all terms is generated by $d$. Moreover, the class $\vr r$ of all regular SMB algebras also forms a variety. Also, we will make use of the following obvious corollary of the property \eqref{property1} of Definition~\ref{regSMBdef} (see also \cite{SMB1}, Lemma 22): 

\begin{cor}\label{regSMBsubalg}
If $S$ is any subuniverse of the semilattice $\m a/{\sim}$ induced by the regular SMB algebra $\m a$, then the union of all $\sim$-classes in $S$ is a subuniverse of $\m a$. In particular, the conclusion holds if $S$ is a down-set (order ideal) or an interval in the partially ordered set $(A/{\sim};\leq)$.
\end{cor}

Another special property of SMB algebras we will use is the following

\begin{df}\label{unitSMBdef}
An SMB algebra $\m a$ over $\sim$ is {\em unital} if there exists an element $1\in A$ such that for all $x\in A$, $1\mt x=x\mt 1=x$.
\end{df}

It is clear that in an unital SMB algebra $[1]_\sim=\{1\}$ and $[1]_{\sim}$ is the greatest element in the semilattice $\m a/{\sim}$.

\begin{prp}[Proposition 21 of \cite{SMB1}]\label{SMBtospec}
Let $\m a$ be a finite SMB algebra. Then there are $\m a$-terms $d'(x,y,z)$ and $x\mt'y$ such that $\lb A;\mt',d'\rb$ is a regular SMB algebra. Moreover, the congruence ${\sim}$ remains unchanged and whenever $a,b,c$ are in the same ${\sim}$-class, then $d'(a,b,c)=d(a,b,c)$.
\end{prp}

In view of Proposition~\ref{SMBtospec} and the discussion in Section 2, if $\m a$ is a finite SMB algebra, an instance of $CSP(\m a)$ can be viewed as an instance of $CSP(\m a')$, where $\m a'$ is the term reduct of $\m a$ which is a regular SMB algebra.

Now we begin our investigation of $CSP(\m a)$ over an SMB algebra $\m a$:

\begin{thm}\label{lintract}
Let $\m a=\lb A;\mt,d\rb$ be an SMB algebra with respect to $\sim$. When the order of $\sim$-classes is linear, then $CSP(\m a)$ is tractable.
\end{thm}

\proof
We assume, as noted above, that $\m a$ is a regular SMB algebra. Let $P=([n],A,\vr c)$ be a $(1,1)$-minimal instance of $CSP(\m a)$. We replace $P$ with a multisorted instance of $CSP(\vr t(\m a))$ as in Definition~\ref{multisortedCSP}. Each domain $A_i$ is the projection to the coordinate $i$ of $R$, where $(S,R)\in \vr c$ and $i\in S$. Since $P$ is (1,1)-minimal, the domain $A_i$ (which is a subuniverse of $\m a$) does not depend on the choice of the constraint $(S,R)$. We call the new, multisorted instance also $P$, since it cannot be distinguished computationally from the original one.

We will call a solution $f\in A^n$ to a multisorted instance of $CSP(\vr t(\m a))$ a least-block solution when for all $i\in [n]$, $f(i)$ is in the least $\sim$-class of the SMB algebra $\m a_i$. If we restrict ourselves only to least-block solutions, then $d$ is a Mal'cev operation compatible with all restrictions of $R_j$ to least blocks, so it is well-known that we can check whether $P$ has a least-block solution in polynomial time. 

The algorithm $SOLVE(P)$ works like this:

\begin{enumerate}
\item[{\bf Step 1.}] Set the working instance $Q:=P$.
\item[{\bf Step 2.}] Replace $Q$ with the $(1,1)$-minimal equivalent instance. If for any $(S,R)\in\vr c(Q)$, $R$ is now empty, output that $P$ has no solutions and stop.
\item[{\bf Step 3.}] Set $i:=n$.
\item[{\bf Step 4.}] Check whether the restriction $Q|_{[i]}$ has a least-block solution using the Mal'cev algorithm.
\begin{enumerate}
\item[{\bf Step 4.a}] If 'YES' and $i=n$, output that $P$ has a solution and stop.
\item[{\bf Step 4.b}] If 'YES' and $i<n$, tighten $Q$ by replacing $A_i$ with $A_i\setminus D$, where $D$ is the least $\sim$-class of $A_{i}$ and go to Step 2.
\item[{\bf Step 4.c}] If 'NO', set $i:=i-1$ and go to step 4.
\end{enumerate}
\end{enumerate}

It is clear that, for a nonempty $(1,1)$-minimal $Q$, $Q|_{[1]}$ always has a least-block solution, so the decrease of $i$ in Step 4.c must stop at some point (we will not get the answer 'NO' at $i=1$). It is also clear that at each application of Step 4.b we tighten the instance and that it can be applied at most $n|A|$ times. By Corollary~\ref{regSMBsubalg}, the tightened instance is still an instance over a regular SMB algebra with a linear order of $\sim$-classes. The running time cost of the algorithm is dominated by at most $n|A|$ applications of $(1,1)$-minimality algorithm in Step 2 and at most $n^2|A|$ applications of the Mal'cev algorithm in Step 4. So, this algorithm works in polynomial time.

To prove that it faithfully computes whether $P$ has a solution, it suffices to prove the following claim

\begin{claim}
Let $P$ be a $(1,1)$-minimal multisorted instance of $CSP(\vr t(\m a))$. If $P|_{[i]}$ has a least-block solution and $P|_{[i+1]}$ does not, then for every solution $f$ of $P$, $f(i)\notin\mathrm{min}(\m A_{i})$.
\end{claim}

(Recall that $\{i\}=[i+1]\setminus[i]$.)

{\em Proof of Claim.} Assume not. Let $f$ be a solution of $P$ such that $f(i)\in\mathrm{min}(\m A_{i})$. Let $g\in A^i$ be the least-block solution of $P|_{[i]}$. We will prove that $\overline{f}\in A^{i+1}$, defined by $\overline{f}(i)=f(i)$ and for all $0\leq j<i$, $\overline{f}(j)=f(j)\mt g(j)$ is a solution to $P|_{[i+1]}$. This would be a contradiction, since $P|_{[i+1]}$ is assumed not to have a least block solution.

Indeed, let $(S,R)\in\vr c(P|_{[i+1]})$. If $i\notin S$, then $(S,R)\in\vr c(P|_{[i]})$, and since both $f|_{[i]}$ and $g$ are solutions to $P|_{[i]}$, then so is $f|_{[i]}\mt g$. Therefore, $\overline{f}|_S=f|_S\mt g|_S\in R$.

On the other hand, let $i\in S$. Let $S'=S\setminus\{i\}$ and by the definition of $P|_{[i]}$ we know that there exists a constraint $(S',R')\in\vr c(P|_{[i]})$. Therefore, $g|_{S'}\in R'\subseteq R|_{S'}$. Therefore, there must exist some $\overline{g}\in R$ such that $\overline{g}|_{S'}=g|_{S'}$. Since $f|_{S}\in R$ as well, $f|_{S}\mt \overline{g}\in R$. But, since $f(i)\in\mathrm{min}(\m A_i)$, $(f|_{S}\mt \overline{g})(i)=f(i)\mt\overline{g}(i) =f(i)$ by Definition~\ref{regSMBdef} (2). Hence, $(f|_{S}\mt \overline{g})(i)=\overline{f}(i)$. On the other hand, for $0\leq j<i$ and $j\in S$, $(f|_{S}\mt \overline{g})(j)=f(j)\mt g(j)=\overline{f}(j)$. We just proved that $\overline{f}$ is a solution of $P|_{[i+1]}$.

Finally, note that for all $0\leq j\leq i$, $\overline{f}(j)\in\mathrm{min}(\m A_j)$. Therefore, $\overline{f}$ is a least-block solution of $P|_{[i+1]}$, a contradiction.
\qed

We note in passing that this case is a generalization of A. Bulatov's rectangular case from \cite{3el} and that the algorithm displayed here would work there, as well. In fact, Bulatov's algorithm would work in our case, too, though it would have a more complicated proof.

We turn to another situation where we can prove tractability. We say that a semilattice is {\em flat} when it has the least element and all other elements are maximal in the semilattice order.

\begin{thm}\label{flattract}
Let $\m a=\lb A;w\rb$ be an SMB algebra wrt. $\sim$. $CSP(\m a)$ is tractable when the order of $\sim$-classes is that of a flat semilattice.
\end{thm}

\begin{proof}
We assume, as before, that $\m a$ is a regular SMB algebra and ensure that $P$ is $(2,3)$-minimal so we may call $P=([n],D,\vr c)$ a multisorted instance of $CSP(\vr t(\m a))$ whose domains are subuniverses of $\m a$. Let $A=O\cup I_1,\ldots,I_m$, where $O$ is the least $\sim$-class and $I_1,\ldots,I_m$ are the maximal ones. 

Because of $(2,3)$-minimality, not only the domains $A_1$ are fixed by the instance, but also $P$ contains all constraints of the form $(\{i,i'\},R_{i,i'})$, and $R_{i,i'}$ is the projection to $\{i,i'\}$ of any constraint relation $R$ such that $(S,R)\in \vr c$ and $\{i,i'\}\subseteq S$. We think of $R_{i,i'}$ as a bipartite graph between disjoint partitions $A_i$ and $A_{i'}$, and the minimality conditions imply that there are no isolated vertices.

$P$ defines a preorder $\preceq$ on the following set: 
$$S(P):=\{(i,j):0\leq i<n \;\&\; 1\leq j\leq m \;\&\; A_i\cap I_j\neq\emptyset\neq O\cap A_i\}.$$ 
We say $(i,j)\preceq(i',j')$ if, whenever $f\in R_{i,i'}$ and $f(i)\in I_j$, then $f(i')\in I_{j'}$ (here $(\{i,i'\},R_{i,i'})\in\vr c$). We define $(i,j)\cong(i',j')$ if $(i,j)\preceq(i',j')$ and $(i',j')\preceq(i,j)$. We call $\cong$-classes {\em strands}.

It is not hard to see that $\preceq$ is indeed a preorder: reflexivity is a trivial consequence of the definition of $\preceq$. For transitivity, assume that $(i,j)\preceq(i',j')$ and $(i',j')\preceq(i'',j'')$, and let $f\in R_{i,i''}$ be such that $f(i)\in I_j$. By $(2,3)$-minimality, there exists $\overline{f}\in R_{i,i',i''}$ such that $\overline{f}|_{\{i,i''\}}=f$. Denote $g=\overline{f}|_{\{i,i'\}}$ and $h=\overline{f}|_{\{i',i''\}}$. From $g(i)=\overline{f}(i)=f(i)\in I_j$ and $(i,j)\preceq(i',j')$ follows $h(i')=\overline{f}(i')=g(i')\in I_{j'}$, similarly, from $(i',j')\preceq(i'',j'')$ we obtain $f(i'')=h(i'')\in I_{j''}$.

The properties to remember about strands are (all of these follow from $(2,3)$-minimality of $P$): Firstly, for any strand and $i\leq n$, there is at most one $(i,j)$ in this strand (the argument uses that the meet of any two elements of $A_i$ which are in different $\sim$-classes must be in $O$). Secondly, for any pair $(i,j)\in S(P)$, it is in exactly one strand since $\cong$ is an equivalence relation. Finally, whenever $f(i)\in I_j$ for some solution $f\in A^n$ of $P$ (or of a restriction of $P$ to some subset of variables), then for all $i'\in[n]$ such that $(i',j')$ is in the same strand as $(i,j)$, then $f(i')\in I_{j'}$ (in case of a restriction of $P$, we need to assume here that $i'$ is in the scope of this restriction, as well). So, we may refer to a strand as $\{(i,t(i)):i\in F\}$ where $F\usub [n]$ and $t:F\rightarrow \{1,\ldots,m\}$.

Now we describe a procedure $TESTASTRAND(P,\{(i,t(i)):i\in F\})$.

\begin{enumerate}
\item[{\bf Step 1.}] Compute the restriction $Q:=P|_F$.
\item[{\bf Step 2.}] Using the Mal'cev algorithm decide if $Q$ has an $O$-valued solution.
\begin{enumerate}
\item[{\bf Step 2.1.}] If 'YES' output 'Exists an $O$-valued solution' and stop.
\item[{\bf Step 2.2.}] If 'NO', test if $Q$ has a solution $f$ which for all $i\in F$ has $f(i)\in t(i)$. This is again a Mal'cev algorithm. If 'YES', output 'Exists a strand-valued solution' and stop, if 'NO', output 'Neither' and stop.
\end{enumerate}
\end{enumerate}

We will repeatedly use this procedure. The main algorithm $SOLVE(P)$ now looks like this:

\begin{enumerate}
\item[{\bf Step 1.}] Set up the working instance $Q:=P$.
\item[{\bf Step 2.}] Applying consistency algorithm reduce $Q$ to a $(2,3)$-minimal instance. If any $(S,R)\in\vr c(Q)$ is such that $R=\emptyset$, output '$P$ has no solution' and stop.
\item[{\bf Step 3.}] Compute $S(Q)$ and $J=\{i\in n:(\exists j)(i,j)\in S(Q)\}$. Compute $\preceq$ and $\cong$ on $S(Q)$.
\item[{\bf Step 4.}] For each strand $\{(i,t(i)):i\in F\}$ do
\begin{enumerate}
\item[{\bf Step 4.1.}] If $TESTASTRAND(Q,\{(i,t(i)):i\in F\})$ outputs 'Exists an $O$-valued solution', move on to next strand.
\item[{\bf Step 4.2.}] If  $TESTASTRAND(Q,\{(i,t(i)):i\in F\})$ outputs 'Exists a strand-valued solution', tighten $Q$ by replacing each $A_i$, $i\in F$, with $I_{t(i)}$ and go to Step 2.
\item[{\bf Step 4.3.}] If $TESTASTRAND(Q,\{(i,t(i)):i\in F\})$ outputs 'Neither', tighten $Q$ by replacing each $A_i$, $i\in F$, with  $A_i\setminus I_{t(i)}$ and go to Step 2.
\end{enumerate}
\item[{\bf Step 5.}] Tighten $Q$ by replacing each $A_i$, $i\in J$, with  $A_i\cap O$.
\item[{\bf Step 6.}] $Q$ is now an instance over Mal'cev algebras, as each $A_i$ is now within a single $\sim$-class, so solve it. Output that $P$ has a solution if $Q$ does and that it has no solution if $Q$ does not and stop.
\end{enumerate}

We see that the instance always tightens to a smaller one whenever either Step 4.2 or Step 4.3 is applied. Therefore, they can be applied at most $n|A|$ many times. Moreover, the number of strands is bounded from above by $nm\leq n|A|$, so at most $n^2|A|^2$ many applications of the TestAStrand procedure can occur during running time of the algorithm. Each application requires at most two applications of the Mal'cev algorithm. Moreover, we need to apply the consistency algorithm in Step 2 at most $n|A|+1$ times. Finally, there is one more application of the Mal'cev algorithm in Step 6. So the time-cost of the algorithm is dominated by the number of steps needed to apply $(2,3)$-minimality algorithm $n|A|+1$ times, plus the number of steps needed to apply the Mal'cev algorithm $2n^2|A|^2+1$ times. As both of these elementary algorithms run in  polynomial time, so does this one.

To study the correctness of our algorithm, we note first that when it outputs that there exists a solution to $P$, then the algorithm really found a solution to a tightening of $P$ to certain specified path through $\sim$-classes (since when the Mal'cev algorithm produces 'YES' it actually finds a generating set for the space of all such solutions). Therefore, this is a solution to $P$, as well.

If the algorithm rejected the instance, on the other hand, it may have done so in the Step 2., or in Step 6. If it rejected $P$ in Step 2., it either did it in the initial $2$-consistency check, which is fine as the instance which fails the initial consistency check indeed has no solution, or in application of Step 2. which occurred after a tightening via either Step 4.2. or 4.3.

Both of these tightenings lead to equivalent instances. To see this, fix a strand $\{(i,t(i)):i\in F\}$. If the instance $P|_F$ has a solution which at some (equivalently, all) $i\in F$ takes values in $I_{t(i)}$ and another one which at some (equivalently, all) $i\in F$ takes values outside the $I_{t(i)}$, then the meet of these two solutions is a solution of $P|_F$ which is $O$-valued at each variable $i\in F$ (recall that by definition of $S(P)$, for all $i\in F$, $A_i\cap O\neq\emptyset$). Therefore, if there are no $O$-valued solutions to $P|_F$, then either $P|_F$ has no strand-valued solutions, or it has no solutions which are outside the strand at all $i\in F$. Again, by the definition of a strand, these are the only two options. Note that the relation $\cong$ is invariant under restriction since $P$ is $(2,3)$-minimal, so $\{(i,t(i)):i\in F\}$ is a strand for $P|_F$, as well as for $P$. Also note, by Corollary~\ref{regSMBsubalg}, that both $I_{t(i)}$ and $A_i\setminus I_{t(i)}$ are subuniverses of $\m a$, so the tightened instances are still within $CSP(\vr t(\m a))$.

Finally, the algorithm may have rejected the instance at Step 6. To see that this is a correct rejection, we need to prove the following

\begin{lm}\label{Malcev}
Assume that $f$ is a solution to a $(2,3)$-minimal multisorted instance $P=([n],D,\vr c)$ of $CSP(\vr t(\m a))$. If for all strands $\{(i,t(i)):i\in F\}$ such that $f(i)\in I_{t(i)}$ there exists an $O$-valued solution of $P|_F$, then $P$ has a solution $\overline{f}\in A^n$ such that $\overline{f}\in O$ whenever $A_i\cap O\neq\emptyset$.
\end{lm}

\begin{proof}
We first note that, since $f$ is a function, for any two different strands $\{(i,t_1(i)):i\in F_1\}$ and $\{(i,t_2(i)):i\in F_2\}$ visited by $f$ (i. e. such that $f(i)\in I_{t_1(i)}$ when $i\in F_1$ and $f(i)\in I_{t_2(i)}$ when $i\in F_2$), we must have that $F_1$ and $F_2$ are disjoint. Therefore, we know that the set $J\usub n$ defined by $J=\{i\in n:f(i)\notin O\;\&\;A_i\cap O\neq\emptyset\}$ is partitioned into $F_1\cup\ldots\cup F_l$, which are the sets of variables of all strands visited by $f$. We may define now uniformly $t:J\rightarrow\{1,\ldots,m\}$ such that for $i\in F_j$, $t(i):=t_j(i)$. As the strands are partially ordered by $\preceq$, we assume without loss of generality that when $i\in F_j$, $i'\in F_{j'}$ and $1\leq j'<j\leq l$, then $\neg(i,t_j(i))\preceq(i',t_{j'}(i'))$. Our assumptions give us $O$-valued solutions $h_j$ to $P|_{F_j}$, for all $1\leq j\leq l$. We prove inductively on $s$ that there exists a solution $g_s$ of $Q_s=P|_{(n\setminus J)\cup F_1\cup\ldots\cup F_s}$ which is equal to $f$ on $n\setminus J$ and to $f\mt h_j$ (which is $O$-valued) on each $F_j$, $1\leq j\leq s$. Clearly for $s=l$ this will be the desired $\overline{f}$, as $Q_l=P$.

The base case is for $s=1$, where we claim that $g_1(i)=f(i)$ for $i\in n\setminus J$ and $g_1(i)=f(i)\mt h_1(i)$ is a solution to $Q_1$. Indeed, let $(S,R)\in\vr c(Q_1)$. If $S$ is disjoint from $F_1$, then $g_1|_S=f|_S$, so $g_1|_S\in R$. If, on the other hand, $S$ is not disjoint from $F_1$, then from the fact that $h_1$ is a solution to $P|_{F_1}$, we know that there exists some $h\in R$ such that $h|_{S\cap F_1}=h_1|_{S\cap F_1}$. Therefore, where $f'=f|_S$, $f'\mt h\in R$. Now for $i\in S\setminus F_1$, $f'(i)\in\mathrm{min}(\m A_i)$, so $f'(i)\mt h(i)=f'(i)=g_1(i)$. On the other hand, where $i\in S\cap F_1$, $f'(i)\mt h(i)=f(i)\mt h_1(i)=g_1(i)$. Therefore, $g_1$ is a solution to $Q_1$.

Assume that $g_s$ is a solution to $Q_s$, where $g_s$ is
defined as above. Let $(S,R)\in \vr c(Q_{s+1})$. If $S\cap F_{s+1}=\emptyset$, there is nothing to prove, and if $S\cap (F_1\cup\ldots\cup F_s)=\emptyset$, this is the same as the base case $s=1$. So we assume that $S$ intersects both $F_1\cup\ldots\cup F_s$ and $F_{s+1}$.

We need to construct some auxiliary elements of $R$. Firstly, we are guaranteed to have $f',g',h'\in R$ such that $f'=f|_S$, $g'|_{S\setminus F_{s+1}}=g_s|_{S\setminus F_{s+1}}$, while $h'|_{S\cap F_{s+1}}=h_{s+1}|_{S\cap F_{s+1}}$, where the existence of $g'$ and $h'$ follows from the definition of the restriction of an instance to a set of coordinates. Secondly, we define $S_0=S\setminus(F_1\cup\ldots\cup F_{s+1})$ and $S_j=S\cap F_j$ for $1\leq j\leq s+1$. Our assumptions ensure that $S_{s+1}\neq\emptyset$ and for at least one $j$, $1\leq j\leq s$, $S_j\neq\emptyset$. Note also that $S_0$ consists precisely of those variables $i\in S$ where $f(i)\in\mathrm{min}(A_i)$, so $f(i)$ is a left absorbing element for $\mt$ on $A_i$ (here we use the fact that $(S,R)\in \vr c(Q_{s+1})$, so $S$ has no variables in $F_j$, $j>s+1$).

Now select some $i_{s+1}\in S_{s+1}$ and for any $j$ such that $1\leq j\leq s$ and $S_j\neq \emptyset$ a variable $i_j\in S_j$. Let $(\{i_j,i_{s+1}\},R_j)\in \vr c(Q_{s+1})$. As we know that $\neg(i_{s+1},t(i_{s+1}))\preceq(i_j,t(i_j))$ and that $(f(i_j),f(i))\in R_j\cap (I_{t(i_j)}\times I_{t(i_{s+1})})$, then there must exist also some $q_j\in R_j$ such that $q_j(i_j)\notin I_{t(i_j)}$, while $q_j(i_{s+1})\in I_{t(i_{s+1})}$. Let $q_j'\in R$ be such that $q_j'|_{\{i_j,i_{s+1}\}}=q_j$ and $p_j'=f'\mt q_j'$. We know that $q_j'|_{S_j}(i)\notin I_{t(i)}$ for all $i\in S_j$, as it is so at $i_j$ and $\{(i,t(i)):i\in S_j\}$ is the restriction of a strand. Therefore, for all $i\in S_j$, $p_j'(i)\in O$. Moreover, for all $i\in S_0\cup S_{s+1}$, $p_j'(i)=f(i)$. Now define $p'=\bigwedge\limits_{S_j\neq \emptyset}p_j'$ to be the left-associated meet. We know that $p'\in R$, that $p'|_{S_0\cup S_{s+1}}=f|_{S_0\cup S_{s+1}}$ and that $p'(i)\in O$ for all $i\in S_j$, $1\leq j\leq s$.

Now we define three tuples $f_1$, $f_2$ and $f_3$ in $R$ in the following way: $$f_1=f'\mt (g'\mt h')\text{, }f_2=p'\mt f_1\text{ and }f_3=p'\mt h'.$$ We claim that $\overline{f}|_S=d(f_1,f_2,f_3)$, thus proving $\overline{f}|_S\in R$. We break down the proof into three cases: 

If $i\in S_0$, then $\overline{f}(i)=f(i)=f'(i)=p'(i)\in\mathrm{min}(\m A_i)$, so it is a left absorbing for the meet. Therefore, $f_1(i)=f_2(i)=f_3(i)=\overline{f}(i)$, and the claim follows by the idempotence of $d$.

If $i\in S_j$ for some $1\leq j\leq s$, then $f_2(i)=p'(i)=f_3(i)\in O$, while 
$$
\begin{gathered}
f_1(i)=f(i)\mt (g'(i)\mt h'(i))=f(i)\mt g'(i)=\\
f(i)\mt(f(i)\mt h_j(i)) =f(i)\mt h_j(i)=\overline{f}(i)\in O
\end{gathered}
$$ (the second equality follows from $g'(i)\in O$, while the fourth equality follows from Definition~\ref{regSMBdef} (4)). Therefore, $$d(f_1,f_2,f_3)(i)=d(f_1(i),p'(i),p'(i))=f_1(i)=\overline{f}(i),$$ since $d$ is a Mal'cev operation on $O$.

Finally, if $i\in S_{s+1}$, then $g'(i)\mt h_{s+1}(i)\in O$ and $$f_1(i)=f(i)\mt(g'(i)\mt h_{s+1}(i)),$$ while $p'(i)=f(i)$ and therefore 
$$
\begin{gathered}
f_2(i)=f(i)\mt f_1(i)=f(i)\mt (f(i)\mt(g'(i)\mt h_{s+1}(i)))=\\
f(i)\mt(g'(i)\mt h_{s+1}(i))=f_1(i).
\end{gathered}
$$ 
On the other hand, $$f_3(i)=p'(i)\mt h_{s+1}(i)=f(i)\mt h_{s+1}(i)=\overline{f}(i)\in O.$$ Now $$d(f_1,f_2,f_3)(i)=f_3(i)=\overline{f}(i),$$ since $d$ is a Mal'cev operation on $O$. This finishes the inductive proof of the lemma.
\end{proof}
We have proved that, if there exists a solution of an instance and if for each strand there exists an $O$-valued solution to the restriction of the instance to the strand variables, then there exists an $O$-valued solution of the whole instance. If the instance $Q$ reached Step 6, there is an $O$-valued solution to the restriction of $Q$ to the strand variables, for each strand, hence testing $Q$ just for $O$-valued solutions determines whether $Q$ has a solution. 
\end{proof}

\section{M-irreducibility and tree-ordered SMB algebras}

\subsection{Consistent maps}

\begin{df}\label{consmaps}
Let $\vr t$ be a template and $(V,D,\vr c)$ be a multisorted instance of $CSP(\vr t)$. A set $p = \{\, p_i \mid i\in V \,\}$ of maps is {\em consistent} with $P$ if
for all $i\in V$ the map $p_i$ is a unary polynomial of $\m a_i$, and 
for every constraint $(S,R)$ and tuple $r\in R$ the tuple
$p|_S(r) = \lb p_i(r_i) : i\in S\rb$ is also in $R$.
We say that $p$ is {\em permutational}, if each $p_i$ is a permutation,
and it is {\em retractive}, if $p_i(x)$ is a retraction for all $i\in V$.
\end{df}

Clearly, every consistent set $p = \{\, p_i : i\in V \,\}$ of maps can be iterated to obtain an retractive one $p' = \{\, p^k_i : i\in V \,\}$ where $k = (\max_{i\in V} |A_i|)!$, for example. 

\begin{df}\label{retrinst}
Let $P$ be an instance of $CSP(\vr t)$ and $p = \{\, p_i \mid i\in V \,\}$ be a retractive consistent set of maps. The {\em retraction} of $P$ via $p$ is the new instance $p(P)$ of $CSP(\vr t)$ defined as
\[ p(P) = \{\,V, \{p_i(\m a_i) \mid i\in V\,\},\{\ (S,p|_S (R)\mid (S,R)\in \vr c \,\}. \]
\end{df}

It easily follows from the definitions that for each constraint $(S,R)$, the relation
\[ p|_S(R) = \{\, p|_S(r) \mid r\in R \,\} = R\cap\prod_{i\in S}p_i(A_i) \] 
is a subuniverse of $\prod_{i\in S} p_i(\m a_i)$. Also, if $P$ is $(k,l)$-minimal, then so is its retraction $p(P)$.

\begin{lm}\label{l:retraction}
Let $P$ be a multisorted instance of $CSP(\vr t)$ and $p$ be a consistent set of retractions. Then $P$ has a solution if and only if $p(P)$ does.
\end{lm}

\begin{proof}
Since $p_i(A_i)\subseteq A_i$ and for any constraint $(S,R)\in \vr c$, $p|_S(R)\subseteq R$, any solution of $p(P)$ is a solution of~$P$. Conversely, if $f$ is a solution of $P$, then the function $p\circ f = \lb p_i(f(i)) : i\in V\rb$ is a solution of~$p(\mathcal A)$.
\end{proof}

\begin{df}\label{decompdef}
Let $P=(V,D,\vr c)$ be a multisorted instance of $CSP(\vr t)$ and $t$ be a binary term
such that $t(x,t(x,y)) = t(x,y)$ for all algebras in $\vr t$. For an element $a\in A_i$ we put $t_a(x) = t(a,x)$, which is an idempotent polynomial of $\m a_i$.
The {\em decomposition} of $P$ via $t$ is the new instance $t(P)$ of
$CSP(\vr t)$ defined as $t(P) = (V',D',\vr c')$. The variables are
\[ 
V' = \{\,(i,a) \mid i\in V,\ a\in A_i\,\}.
\]
The domains are
\[ \begin{gathered}
D'=\{\,A_{i,a} \mid i\in V,\ a\in A_i\,\},\text{ where}\\
A_{i,a}=t_a(A_i)=\{\, t(a,x) \mid x\in A_i \,\}.
\end{gathered}
\]
Finally, the constraints are
\[ \begin{gathered}
\vr c'=\{\,(S_r,R_r)\mid(S,R)\in \vr c,\ r\in R\,\}\,\cup\,\{\,(S_i,T_i)\mid i\in V\,\},\text{ where}\\
S_r=\{\,(i,r(i)) \mid i\in S\,\}\text{ and }R_r=\{\,t(r,x)\mid x\in R\,\},\text{ while}\\
S_i=\{\,(i,a)\mid a\in A_i\,\}\text{ and }\\
T_i=\Sg^{\prod_{a\in A_i}t_a(\sm a_i)}(\{\,\lb t(a,b):a\in A_i\rb\mid b\in A_i\,\}).
\end{gathered}
\]
\end{df}

\begin{lm}
\label{l:tA-solution}
Let $P$ be a multisorted instance of $CSP(\vr t)$ and $t$ be a binary term
such that $\vr t\models t(x,t(x,y)) = t(x,y)$.
If $P$ has a solution, then so does $t(P)$.
\end{lm}

\begin{proof}
Let $f$ be a solution of the instance $P$. We define a solution $g$ of $t(P)$ as 
\[ g((i,a)) = t(a,f(i)) \]
for all $(i,a)\in V'$. Clearly, $g((i,a))\in A_{i,a}$. Take a constraint of the form $(S_r,R_r)$. By definition, 
\[ g|_{S_r} = \lb t(r(i),f(i)) : i\in S \rb = t(r,f|_S). \]
However, $f$ is a solution, so both $r$ and $f|_S$ are in $R$ and therefore $t(r,f|_I)\in R_r=\{\,t(r,x)\mid x\in R\,\}$, as well.

Now take a constraint $(S_i,T_i)$ of the second kind. Here
\[ g|_{S_i} = \lb t(a,f(i)) : a\in A_i \rb, \]
that is, $g|_{S_i}$ is one of the generating elements of the algebra $\m t_i$.
\end{proof}

In the next lemma we will try to understand the structure of the relations $\m T_i$ in $t(P)$, so we focus on a single algebra $\m b=\m a_i$ for the moment.

\begin{lm}
\label{l:t-constraint}
Let $\m B$ be an algebra, and $t$ be a binary term such that $\m b\models t(x,t(x,y))=t(x,y)$. For $b\in B$ let $\m B_b = t_b(\m B)$, and put
$\m B^* = \prod_{b\in B}\m B_b$. Let 
\[ \m T = \Sg^{\m B^*} \{\, \lb t(b,c) : b\in B \rb \mid c\in B \,\} \]
and take a tuple $r\in T$. Then the following hold.
\begin{enumerate}
\item $r$ viewed as the map $b\mapsto r(b)$ is a unary polynomial of $\m B$.
\item Let $b_1,b_2\in B$ and $\vartheta$ be a congruence of $\m B$.
If $t(b_1,x) \equiv_\vartheta t(b_2,x)$ for all $x\in B$, then $r(b_1) \equiv_\vartheta r(b_2)$.
\end{enumerate}
\end{lm}

\begin{proof}
Each generator tuple $\lb t(b,c) : b\in B \rb$ of $\m T$ is actually a map from $B$ to $B$ and it is a unary polynomial $\m B$ in the variable $b$ where $c$ is a constant. When we generate the subalgebra by these vectors, then we take a basic operation $f$ of $\m B^*$, some unary polynomials $r_1(b),\dots,r_k(b)$ already generated and generate the tuple $\lb r(b):b\in B\rb$ in $\m b^*$ where
\[r(b) = t(b,f(r_1(b),\dots,r_k(b))),\] 
so $r(b)$ is again a unary polynomial of $\m B$ in the variable $b$.

To prove the second claim it is enough to see that $s(b_1)\equiv_\vartheta s(b_2)$ for each generator tuple $s$ and verify that this property is preserved. For $s=\lb t(b,c):b\in B\rb$ a generator of $\m t$, the claim becomes trivial (what is assumed is the same as the desired conclusion). So assume that the unary polynomials $r_1,\dots,r_k$ are already generated, 
$$
\begin{gathered}
r=f^{\sm b^*}(r_1,\dots,r_k)=\lb t(b,f(r_1(b),\dots,r_k(b))):b\in B\rb
\text{ and}\\ 
r_1(b_1)\equiv_\vartheta r_1(b_2), \dots, r_k(b_1) \equiv_\vartheta  r_k(b_1).
\end{gathered}
$$ 
Thus $c_1 = f(p_1(b_1),\dots,f_k(b_1)) \equiv_\vartheta  f(p_1(b_2),\dots,f_k(b_2)) = c_2$, and using again our assumption that $t(b_1,x)\equiv_\vartheta t(b_2,x)$, we get that 
$$r(b_1)=t(b_1,c_1) \equiv_\vartheta t(b_1,c_2) \equiv_\vartheta t(b_2,c_2) = r(b_2)$$
for the newly generated polynomial $r$.
\end{proof}

\begin{lm}
\label{l:A-retraction}
Let $P$ be a multisorted instance of $CSP(\vr t)$ and $t$ be a binary term
such that $\vr t\models t(x,t(x,y)) = t(x,y)$.
If $t(P)$ has a solution, then there exists a consistent set $\{\, p_i : i\in V \,\}$ of unary polynomials for the instance $P$
such that each polynomial $p_i$ of $\m A_i$ satisfies the conclusion of Lemma~\ref{l:t-constraint}.
\end{lm}

\begin{proof}
Let $g$ be a solution of $t(P)$. We define a consistent set $p = \{\, p_i \mid i\in V \,\}$ of unary maps for $P$ as
\[ p_i(a) = g((i,a)) \]
for $i\in V$ and $a\in A_i$.
By Lemma~\ref{l:t-constraint}, each map $p_i : A_i\to A_i$ is a unary polynomial of $\m a_i$ which also satisfies Lemma~\ref{l:t-constraint} (2). To see that these polynomial maps are consistent with $P$, take a constraint $(S,R)\in \vr c$ and a tuple $r\in R$. Since $g$ was a solution to $t(P)$ it respects
the constraint $R_r$, that is the tuple $\lb g((i,r_i)) : i\in S\rb$ is in $R_r \subseteq R$. But this tuple is exactly $p|_S$, which shows that $p$ is consistent.
\end{proof}

\begin{df}\label{defeliminated}
We say that an idempotent algebra $\m b$ can be {\em eliminated}, if whenever 
$\vr t$ is a template such that $\m B\in\vr t$, and $\vr t\setminus\{\m B\}$ is also a template, and $CSP(\vr t\setminus\{\m B\})$ is tractable, then $CSP(\vr t)$ is also tractable.
\end{df}

\begin{lm}
\label{l:withC}
Let $\m B$ be an algebra and $t$ be a binary term 
of $\m B$ such that for each $b\in B$ the map $t_b(x) = t(b,x)$ is a retraction which is not surjective.
Let $C$ be the set of elements $c\in B$ such that $x\mapsto t(x,c)$ is a permutation. If $C$ generates a proper subuniverse of~$\m B$, then
$\m B$ can be eliminated.
\end{lm}

\begin{proof}
Let $\vr t$ be a template containing $\m B$, assume that $\vr t\setminus \{\m b\}$ is also a template, that $CSP(\vr t\setminus \{\m b\})$ is tractable, and let $P$ be an instance of $CSP(\vr t)$ containing at least one copy of $\m B$. Replace all occurences of $\m B$ in $P$ with the subalgebra generated by the set $C$ and restrict the constraints accordingly. Clearly, this new instance is an instance of $CSP(\vr t\setminus\{\m B\})$ so it can be solved in polynomial time. If it has a solution, then we are done, so we can assume that it does not.

Since the maps $t_b$ are not surjective, $|t_b(\m b)|<|\m B|$ and therefore the decomposition $t(P)$ is an instance of $CSP(\vr t\setminus\{\m B\})$. Thus it can be solved in polynomial time. If $t(P)$ has no solution, then $P$ has no solution, either, by Lemma~\ref{l:tA-solution}. On the other hand, if $t(P)$ has a solution, then by Lemma~\ref{l:A-retraction} we have a consistent set $p = \{\, p_i : i\in V \,\}$ of unary polynomials for $P$. Let us assume for a moment that $p$ is not permutational. Now $p$ can be iterated to obtain a retractive non-permutational consistent set $p'$ of unary polynomials for $P$. By Lemma~\ref{l:retraction} we know that $P$ has a solution if and only if $p'(P)$ does. Also, since $p'$ is non-permutational, at least one of the domains of $p'(P)$ is smaller than that of $P$. So by iterating this procedure we will either find out that $P$ has no solution, or get to a point when the algebra $\m B$ no longer occurs in the instance $P$.

Now we go back to the problem of making sure that $p$ becomes non-permutational. From the first paragraph of the proof, we know that if $P$ has a solution $f$, then for at least one $i\in V$, $\m a_i\cong\m B$ and $f(i)\not\in C$. Let us iterate through all variables $i\in V$ such that $\m a_i\cong \m B$ (for simplicity, we assume $\m A_i=\m B$) and all elements $d\in B\setminus C$. For each choice of $i$ and $d$ we create a new instance from $t(P)$ by adding the new unary constraint stating that the solution $g|_{S_i} = \lb t(b,d) : b\in B_i \rb$. This ensures that $p_i(b) = t(b,d)$, that is, it is not permutational. If for any of these choices we find a solution to $t(P)$, then we can reduce the instance $P$ as shown in the paragraph above. Otherwise we conculde that the instance has no solution.
\end{proof}

\subsection{Applications}

\begin{cor}
\label{c:single-maximal-block}
Let $\m a$ be a finite idempotent algebra, $t(x,y)$ a term of $\m a$ and ${\sim}\in\cn a$ such that $(A/{\sim};t)$ is a semilattice with more than one maximal element. Then $\m a$ can be eliminated.
\end{cor}

\begin{proof}
We can iterate $t$ to obtain $t(x,t(x,y)) = t(x,y)$ on $\m a$, while $(A/{\sim};t)$ is the same semilattice. Since $\m a/{\sim}$ has more than one maximal element, for all $a\in A$ the maps $x\mapsto t(a,x)$ and $x\mapsto t(x,a)$ are not permutations. Thus we can apply Lemma~\ref{l:withC} with $C=\emptyset$.
\end{proof}

\begin{cor}
\label{c:single-maximal-element}
Let $\m a$ be a finite idempotent algebra, ${\sim}\in\Cn a\setminus\{1_A\}$ and $t$ be a binary term such that $(A/{\sim};t)$ is a semilattice with the largest element $B$ ($B$ is the neutral element in $(A/{\sim};t)$). If $B$ contains more than one element and satisfies $t(x,y)=x$, then $\m a$ can be eliminated.
\end{cor}

\begin{proof}
We can assume that $t(x,t(x,y)) = t(x,y)$ on~$\m a$, since we can iterate $t$ in the second variable without destroying the required properties stated in the lemma. Suppose that the ${\sim}$-block $B$ has more than one element.
Then the maps $t_a(x) = t(a,x)$ are not permutations for any $a\in A$. Moreover, for any $c\in A$ for which $x\mapsto t(x,c)$ is a permutation we must have $c\in B$. However, $B$ is a proper subuniverse of $\m a$, thus we can apply Lemma~\ref{l:withC} to finish the proof.
\end{proof}

\begin{cor}\label{SMBeliminated}
Let $\m a$ be a finite SMB algebra over $\sim$. $\m a$ can be eliminated, unless $\m a/{\sim}$ has the largest block $B$ with respect to the semilattice order, that largest block satisfies $|B|=1$. 

Moreover, in the case when there is the largest $\sim$-block $B$, $B=\{b\}$, and $p(x)=b\mt x$ is not a permutation, then $\m a$ can be eliminated.
\end{cor}

\begin{proof}
The first paragraph is a special case of Corollaries~\ref{c:single-maximal-block} and \ref{c:single-maximal-element}.

As for the final sentence, the iteration which produces $\mt'$ from $\mt$ (used in \cite{SMB1} in the proof of Proposition~\ref{SMBtospec}) will give us a term $t(x,y)=x\mt' y$ such that $t(a,x)$ is a retraction which is not surjective for any $a\in A$. This is obvious for $a\notin B$, while if $a=b$, then $$b\mt' x=b\mt (b\mt \dots(b\mt x)\dots)=p^{|A|!}(x),$$
and since $p(x)$ is not surjective, neither is $b\mt' x$. The only $c$ such that $t(x,c)$ is a permutation is $b$, so the conditions of Lemma~\ref{l:withC} are fulfilled (unless $|A|=1$) and $\m a$ can be eliminated.
\end{proof}

\begin{cor}\label{treetract}
Let $\m a$ be a finite SMB algebra over $\sim$. Assume that $\m a/{\sim}$ is a tree-ordered meet semilattice (i.e. for any $\sim$-class $B$, the set of all $\sim$-classes below $B$ or equal to it is linearly ordered). Then $CSP(\m a)$ is tractable.
\end{cor}

\begin{proof}
According to Proposition~\ref{SMBtospec}, we replace $\m a$ with its term reduct to assume that $\m a$ is a regular SMB algebra with tree-ordered $\sim$-classes. Let $\vr t:=\vr t(\m a)$. Note that any $\m b\in \vr t$ in which the order of ${\sim}$-classes is not linear can not be a homomorphic image, a subalgebra, or a retract of any $\m c\in \vr t$ in which the order of ${\sim}$-classes is linear. Hence for any template $\vr t'\subseteq \vr t$ and any $\m b\in \vr t'$ which has maximal size among algebras in $\vr t'$ which do not have a linear order of $\sim$-classes, we have that $\vr t'\setminus\{\m b\}$ is also a template. Finally, note that in any $\m b\in\vr t$ in which the order of ${\sim}$-classes is not linear, the semilattice $\m b/{\sim}$ has more than one maximal element, and thus can be eliminated by Corollary~\ref{c:single-maximal-block}. 

So we can make a sequence of templates $\vr t=\vr t_0,\vr t_1,\dots,\vr t_k$ so that for all $i\leq k$, $\vr t_i=\vr t_{i-1}\setminus \{\m b_i\}$, $\m b_i\in\vr t_{i-1}$ has maximal size among algebras in $\vr t$ with non-linear order of ${\sim}$-classes, and $\vr t_k$ is a template of regular SMB algebras with linear order of $\sim$-classes. By a multisorted variant of Theorem~\ref{lintract}, $CSP(\vr t_k)$ is tractable, and thus, by Definition~\ref{defeliminated} and an inductive argument, $CSP(\vr t_0)=CSP(\vr t(\m a))$ is also tractable.
\end{proof}

\begin{df}\label{defM-irred}
Let $P=(V,D,\vr c)$ be a multisorted instance of $CSP(\vr t)$, where $\vr t$ consists of finite SMB algebras. We say that $P$ is strongly M-irreducible if for all $\m a_i\in D$, either $\m a_i$ is a Mal'cev algebra (i.e. ${\sim}$ is the full relation on $A_i$), or $\m a_i/{\sim}$ has the largest element $B$ (with respect to the semilattice order) and $|B|=1$.

We say that $P$ is weakly M-irreducible if for any $\m a_i\in D$ which is maximal-sized among non-Mal'cev sorts in $D$, $\m a_i$ is unital.
\end{df}

The following result is Corollary 1 of A.Bulatov's paper \cite{BuSMB}, but the argument given there has a gap:

\begin{thm}\label{M-irrdesired}
Let $\m a$ be a finite SMB algebra. For any instance $P$ of $CSP(\m a)$ we can apply a polynomial time algorithm which either solves $CSP(P)$ or reduces it to an equivalent strongly M-irreducible multisorted instance of $CSP(\vr t(\m a))$.
\end{thm}

One is tempted to just apply Corollary~\ref{SMBeliminated} and say we are done. But this doesn't quite work. If an algebra $\m a_i\in\vr t$ is such that $\vr t\setminus\{\m a_i\}$ is not a template (e.g. when $\m a_i$ is a subuniverse, or a factor, or a retract of some larger algebra in $\vr t$), then we cannot just remove $\m a_i$. Indeed, it is not hard to imagine such an ``algorithm'' would end up in an infinite loop, so it might not even be an algorithm. Namely, one would go from an instance $P$ to its decomposition $t(P)$, but then the proof of Lemma~\ref{l:withC} would require one to recursively solve $t(P)$. In the course of this recursion we may run into a strictly bigger instance than the original instance $P$, forcing our algorithm into an infinite loop.

We plug this gap in Bulatov's paper by proving the above theorem, but in the next Section, after we review some parts of D. Zhuk's proof of the Dichotomy Conjecture from \cite{Zhdich}.

Another, better proof of tractability of CSP over SMB algebras is just a small variation of the original proof by A. Bulatov. Namely, the original ideas and arguments by Bulatov were sufficient to the task, but the definitions and statements of lemmas need to be amended. We are going to use just the following weaker result:

\begin{thm}\label{M-irrweaker}
Let $\m a$ be a finite SMB algebra. For any instance $P$ of $CSP(\m a)$ we can apply a polynomial time algorithm which either solves $CSP(P)$ or reduces it to an equivalent weakly M-irreducible multisorted instance of $CSP(\vr t(\m a))$.
\end{thm}

\begin{proof}
This is what applying Corollary~\ref{SMBeliminated} gives us.
\end{proof}

So, why do we give the proof in Section 5? As we said, Section 6 demonstrates that Bulatov's original proof was essentially correct, needing only a minor intervention, while Section 5 uses as a black box a powerful result by Zhuk whose only known proof is inextricable from the whole Dichotomy proof. However, generalizing Bulatov's proof from SMB algebras to Taylor algebras is the route Bulatov took to the Dichotomy, and it proved a very difficult and complicated task. We hope our ideas from Section 5 can evolve into a simpler proof, merging Zhuk's and Bulatov's ideas into an argument simpler than either.

\section{Hypergraph connectivity and Z-irreducibility}

Lest we forget, and to avoid complications, we will assume henceforth that for any multisorted CSP instance we will consider, all domains of variables $A_i$ of the instance are pairwise disjoint. 

We remind the reader of some basic facts about hypergraphs. 

\begin{df}\label{hyperdef}
A {\em hypergraph} is an ordered pair $(V,E)$, where $V$ is the vertex set and $E\subseteq (P(V)\setminus\{\emptyset\})$ is the set of hyperedges.
\end{df}

The usual undirected graph is, therefore, a hypergraph where all hyperedges have two elements. The notions like paths and connectivity can also be generalized to hypergraphs.

\begin{df}\label{hyperpath}
Let $\Gamma=(V,E)$ be a hypergraph. A {\em path} in $\Gamma$ is a finite sequence of the form $$p=a_0, S_1, a_1,S_2,\dots,a_{n-1},S_n,a_n,$$ where each $a_i\in V$, $S_j\in E$ and for all $0<i\leq n$, $a_{i-1}, a_i\in S_i$. If $p$ is the path given above, we say that $p$ connects the vertices $a_0$ and $a_n$, and we say that vertices $a$ and $b$ are connected if there exists a path that connects them. A hypergraph $\Gamma=(V,E)$ is connected if any pair of vertices in $V$ is connected.

A {\em hypergraph homomorphism} is a map from the vertex set of one hypergraph to the vertex set of another such that each edge maps onto an edge.
\end{df}

\begin{df}\label{probhypergraphs} Let $\vr t$ be a template and let $P=(V,D,\vr c)$ be a multisorted instance of $CSP(\vr t)$. We define two hypergraphs associated with $P$, the {\em scope graph} of $P$, written $\Gamma_V$, and the {\em microstructure graph} of $P$, written $\Gamma_P$, by
$$
\begin{gathered}
\Gamma_V=(V,\{S:(S,R)\in \vr c\})\text{ and }\\
\Gamma_P=\left(\bigcup\limits_{i\in V}A_i,E_P\right).
\end{gathered}
$$
Recall that the domains $A_i$ are pairwise disjoint. $E_P$ is defined as the union
$$E_P=\bigcup\limits_{(S,R)\in \vr c}\{\{a_1,\dots,a_{|S|}\}:(a_1,\dots,a_{|S|})\in R\}.$$
\end{df}

It is clear that the instance $P$ has a solution iff there exists a hypergraph homomorphism $f$ from $\Gamma_V$ to $\Gamma_P$ such that for each $i\in V$, $f(i)\in A_i$. The hypergraph variant is just a rephrasing, but it allows us to speak of an instance in terms of hypergraph connectivity. For example, we notice the following easy fact.

\begin{prp}\label{hyperconnPimpliesV}
Let $\vr t$ be a template and let $P=(V,D,\vr c)$ be a multisorted instance of $CSP(\vr t)$. If the microstructure graph $\Gamma_P$ of $P$ is connected, then the scope graph $\Gamma_V$ of $P$ is connected.
\end{prp}

\begin{proof}
Let $i,j\in V$. Select arbitrary $a\in A_i$ and $b\in A_j$. Select a path
\[a=a_0,E_1,a_1,\dots,E_k,a_k=b\]
in $\Gamma_P$ connecting $a$ with $b$. For each hyperedge $E_i$ there exists a tuple $\overline{c}_i$ and a constraint $(S_i,R_i)\in\vr c$ such that $\overline{c}_i\in R_i$ and $E_i=\{\overline{c}_i(1),\dots,\overline{c}_i(|S_i|)\}$. If $a_j\in A_{i_j}$, for $j=0,1,\dots,k$, then
\[i=i_0,S_1,i_1,\dots,S_k,i_k=b\]
is a path in $\Gamma_V$ connecting $i$ with $j$.
\end{proof}

Next we define the notion of cycle consistency.

\begin{df}\label{cyclecons}
Let $\vr t$ be a template and let $P=(V,D,\vr c)$ be a multisorted instance of $CSP(\vr t)$. We say that $P$ is {\em cycle consistent} if for every $i\in V$, every $a\in A_i$ and every closed path 
\[p=i_0,S_1,i_1,S_2,\dots,S_k,i_k\]
in $\Gamma_V$ such that $i_0=i_k=i$, there exists a closed path
\[p'=a_0,E_1,a_1,E_2,\dots,E_k,a_k\]
in $\Gamma_P$ such that $a_0=a_k=a$ and for all $0<j\leq k$, $(S_j,R_j)\in\vr c$ and $E_j$ is the set of all coordinates of some tuple in $R_j$.
\end{df}

Cycle consistency is a consequence of $(2,3)$-minimality, but may be a weaker property. 

In order to express the Z-irreducibility property, we need another notion of an induced smaller instance, which generalizes the restriction of the instance which we introduced in the proof of Theorem~\ref{lintract}.

\begin{df}\label{inducedsubinstance} Let $\vr t$ be a template, let $P=(V,D,\vr c)$ be a multisorted instance of $CSP(\vr t)$, and let $V'\subseteq V$ and $\vr c_1\subseteq \vr c$. We say that $P'=(V',D',\vr c')$ is the subinstance of $P$ induced by $(V',\vr c_1)$ if $D'=\{A_i:i\in V'\}$, while $\vr c'=\vr c_1|_{V'}=\{(S',R'): (S,R)\in \vr c_1,$ $S'=S\cap V'$ and $R'=R|_{S'}\}$.
\end{df}

Now we are ready to define Z-irreducible instances.

\begin{df}\label{Zirred}
Let $\vr t$ be a template and let $P=(V,D,\vr c)$ be a multisorted instance of $CSP(\vr t)$. We say that $P$ is {\em Z-irreducible} if, for every $V'\subseteq V$ and $\vr c_1\subseteq\vr c$ such that the subinstance $P'=(V',D',\vr c')$ induced by $(V',\vr c_1)$ has a connected scope graph, but disconnected microstructure graph, there exists a solution of $P'$ through any point. More precisely, for every $i\in V'$ and any $a\in A_i$, there exists a solution $f$ of $P'$ such that $f(i)=a$.
\end{df}

Next we need the notion of a link partition, which we will also use in the next section.

\begin{df}\label{linkdef}
Let $\vr t$ be a template and $P=(V,D,\vr c)$ a (1,1)-minimal multisorted instance of $CSP(\vr t)$. We say that $P$ has a link partition if for each $i\in V$ there exists an equivalence relation $\varepsilon_i$ on $A_i$ such that 
\begin{enumerate}
\item There exists $k\in\mathbb{N}$, $k\geq 2$, such that for all $i\in V$, $|A_i/\varepsilon_i|=k$
\item There exists an ordering of the $\varepsilon_i$-classes $A_i=A_{i,1}\dot{\cup}\dots\dot{\cup}A_{i,k}$ so that for all $(S,R)\in \vr c$, $R=R_1\dot{\cup}\dots\dot{\cup}R_k$ and for all $j\leq k$, $R_j\subseteq \prod\limits_{i\in S}A_{i,j}$. 
\end{enumerate}
\end{df}

In the above definition, $\dot{\cup}$ stands for the disjoint union, of course.

\begin{prp}
Let $\vr t$ be a template and $P=(V,D,\vr c)$ a (1,1)-minimal multisorted instance of $CSP(\vr t)$ such that $\Gamma_V$ is connected. Then $P$ has a link partition iff $\Gamma_P$ is disconnected.
\end{prp}

\begin{proof}
Let $\Gamma_P$ be disconnected and let $Q_1,\dots,Q_k$ be the connected components of $\Gamma_P$. For all $i\in V$, we define $A_{i,j}:=A_i\cap Q_j$. If $(S,R)\in V$, then for any $r\in R$ and $i,j\in S$, $r(i)$ and $r(j)$ must be connected, so all $r(i)$ are in the same connected component, say $Q_j$. If we define $R_j$, $j=1,\dots,k$ to be $\{r\in R:(\forall i\in S)r(i)\in Q_j\}$, from the previous sentence it follows that $R=R_1\cup\dots\cup R_k$. As $Q_j$ are pairwise disjoint, thus $R=R_1\dot{\cup}\dots\dot{\cup}R_k$. 

Finally, for each $Q_j$ and $A_i$, we need to show $A_i\cap Q_j\neq\emptyset$ to prove property (1) of link partitions. Choose some $i'\in V$ so that $A_{i'}\cap Q_j\neq\emptyset$. Since $\Gamma_V$ is connected, there exists some $i'=i_0, S_1,i_1,S_2,\dots,S_t,i_t=i$ which is a path from $i'$ to $i$ in $\Gamma_V$. Let $(S_1,R_1), \dots ,(S_t,R_t)\in \vr c$. We select some $a=a_0\in A_{i'}\cap Q_j$. Next, by (1,1)-minimality, there exist some $r_1\in R_1,\dots,r_t\in R_t$ and $a_1\in A_{i_1},\dots,a_{i_t}\in A_{i_t}=A_i $ so that for all $1\leq s\leq t$, $r_s(i_{s-1})=a_{s-1}$ and $r_s(i_s)=a_s$. Inductively it follows that $a_1,a_2,\dots,a_s\in Q_j$, so $a_t\in Q_j\cap A_i$, as desired.

For the other direction, if $P$ has a link partition and $j\neq j'$ for some $j,j'\leq k$, then for any $i,i'\in V$ no edge in $\Gamma_P$ connects an element of $A_{i,j}$ and an element of $A_{i',j'}$. Thus $\Gamma_P$ must be disconnected.
\end{proof}

Next we prove that the connected components of $\Gamma_P$ give rise to congruences.

\begin{prp}\label{linksubuniv}
Let $\vr t$ be a template and $P=(V,D,\vr c)$ a (1,1)-minimal and cycle consistent multisorted instance of $CSP(\vr t)$ such that $\Gamma_V$ is connected. Then, for each $i\in V$, $\{(a,b)\in A_i:a$ and $b$ are connected in $\Gamma_P\}$ is a congruence relation of $\m a_i$. Consequently, for any connected component $Q_j$ of $\Gamma_P$, $Q_j\cap A_i$ is a subuniverse of $\m a_i$
\end{prp}

{\em Sketch of a proof.}
For each $a,b\in A_i$ such that they are connected in $\Gamma_P$, let $p_{a,b}$ be the path from $a$ to $b$ in $\Gamma_P$ and $q_{a,b}$ the corresponding closed path from $i$ to $i$ in $\Gamma_V$. By cycle consistency, for any $c\in A_i$ there is a closed path from $c$ to $c$ in $\Gamma_P$ which also corresponds to $q_{a,b}$. Now we define $q$ to be the closed path from $i$ to $i$ in $\Gamma_V$ obtained by concatenating all paths $q_{a,b}$ for any $(a,b)\in A_i\times A_i$ such that $a$ is connected to $b$ in $\Gamma_P$. 

We claim that for any $c$ and $d$ in $A_i$ which are connected in $\Gamma_P$, $c$ is connected to $d$ by a path in $\Gamma_P$ corresponding to $q$. To see this just circle around from $c$ to $c$ by paths corresponding to various $q_{a,b}$ until $q_{c,d}$ comes along. Then move from $c$ to $d$ and afterwards keep circling from $d$ to $d$. Note that
$$\{(a,b)\in A_i\times A_i: a\text{ is connected to }b\text{ by a path in }\Gamma_P\text{ corresponding to }q\}$$
is a pp-definable relation, so a compatible relation of the algebra $\m a_i$. Moreover, it is obviously an equivalence relation, being the restriction to $A_i$ of the hypergraph connectedness relation, so it is a congruence. The final sentence follows since congruence classes are subuniverses.
\qed

For each induced subinstance with a connected scope graph, but disconnected microstructure graph which is $(1,1)$-minimal and cycle consistent, we just proved that the subinstance splits into disjoint smaller instances. Each of these smaller instances is given by the connected components of the microstructure graph, and the domain of each variable $i$ in each of the smaller instances is a nonvoid proper subuniverse of $A_i$. Therefore, one can inductively solve each of these smaller instances to ensure there exists a solution through any point. However, there may be exponentially many induced subinstances to consider, so it is not yet obvious we can enforce Z-irreducibility this way.

This issue was resolved by the procedure CHECKIRREDUCIBILITY in \cite{Zhdich}. In a nutshell, D. Zhuk considers one domain of a variable and a maximal congruence on it, assumes the elements inside the same congruence classes are ``connected", and keeps adding constraints and variables for as long as the microstructure graph remains disconnected and the connected components restrict to the fixed domain of variable as a congruence contained in the selected maximal congruence. The maximal such induced subinstance is then solved through any point in the fixed domain of the variable. This is done in polynomial time, since there are only polynomially many maximal subinstances to check. Zhuk proves that any induced subinstance which is cycle consistent must be contained in a maximal subinstance checked by the procedure, hence the procedure correctly forces Z-irreducibility of the instance. Thus one can either resolve an instance or reduce it to an equivalent Z-irreducible instance.

The main tool from Zhuk's Dichotomy proof we will use in this paper is the following theorem (we will use it as a black box, without proof):

\begin{thm}[follows from Theorem 5.5 of \cite{Zhdich}]\label{Zhlemma}
Let $\vr t$ be a template of SMB algebras and let $P=(V,D,\vr c)$ be a multisorted instance of $CSP(\vr t)$. If $P$ is Z-irreducible, (1,1)-minimal, cycle consistent and has a solution, then $P$ has a solution $f$ such that for each $i\in V$, $f(i)$ is in the least $\sim$-class in $\m a_i$.
\end{thm}

Now we know all prerequisites needed to prove Theorem~\ref{M-irrdesired}. First we prove some easy, technical lemmas.

\begin{lm}\label{unitsconnected}
Let $R\leq_{sd}\m a\times\m b$ be a subdirect product of unital SMB algebras $\m a$ and $\m b$ (recall Definition~\ref{unitSMBdef}). Let $1_A$ and $1_B$ be the unit elements of $\m a$ and $\m b$, respectively. Then $1_A$ and $1_B$ are connected by a path in $R$ (viewed as a bipartite graph between $A$ and $B$). 
\end{lm}

\begin{proof}
If $(1_A,1_B)\in R$ then there is nothing to prove. Otherwise, by subdirectness, there exist $a\in A$ and $b\in B$ such that $(1_A,b),(a,1_B)\in R$. Since $\m a$ and $\m b$ are both unital, then \[(a,b)=(1_A,b)\mt (a,1_B)\in R.\]
Hence, the path $1_A - b - a - 1_B$ connects $1_A$ to $1_B$ in $R$.
\end{proof}

\begin{cor}\label{unitsconnhyper}
Let $\vr t$ be a template and let $P=(V,D,\vr c)$ be a (1,1)-minimal, multisorted instance of $CSP(\vr t)$. Assume, moreover, that for some $i,j\in V$ $i$ and $j$ are connected in the hypergraph $\Gamma_V$ and that $\m a_i$ and $\m a_j$ are unital, with units $1_i$ and $1_j$, respectively. Then $1_i$ and $1_j$ are connected in $\Gamma_P$.
\end{cor}

\begin{proof}
Let $i=i_0,S_1,i_1,S_2,\dots,S_k,i_k=j$ be a path in $\Gamma_V$ connecting $i$ and $j$. Let $(S_1,R_1),\dots,(S_k,R_k)\in \vr c$ be the constraints involving the scopes $S_1,\dots,S_k$. We define the relation $R\subseteq A_i\times A_j$ by
\[
\begin{gathered}
(x,y)\in R\text{ iff }(\exists\overline{z}_1,\overline{z}_2,\dots,\overline{z}_k)(\exists u_1,\dots,u_{k-1})\\
\left[
(x,\overline{z}_1,u_1)\in R_1,
(u_1,\overline{z}_2,u_2)\in R_2,\dots, (u_{k-1},\overline{z}_k,y)\in R_k
\right]
\end{gathered}
\]
In the above formula, the coordinates of $R_i$ were permuted for better clarity, so that the first coordinate of a tuple in $R_i$ should be in $A_{i-1}$ and the last coordinate of $R_i$ should be in $A_i$. The above formula is a primitive positive formula and all relations $R_\ell$ it uses are subuniverses of the products of the domains of variables, and hence (as is well known) thus defined relation $R$ must be a subuniverse of $\m a_i\times\m a_j$. 

Moreover, from (1,1)-minimality and by an induction on the length $k$ of the hypergraph path, it follows that, for any $a\in A_i$ there exists $b\in A_j$ such that $(a,b)\in R$ and similarly, for any $b\in A_j$ there exists $a\in A_i$ such that $(a,b)\in R$. Therefore, $R\leq_{sd}\m a_i\times\m a_j$ and the result follows by Lemma~\ref{unitsconnected}.
\end{proof}

\begin{lm}\label{addunittuple}
Let $R\leq_{sd}\prod\limits_{i=1}^n\m a_i$ be a subdirect product of regular and unital SMB algebras $\m a_i$ and let $1_i$ be the unit element of $\m a_i$. Let $\overline{1}\in \prod\limits_{i=1}^n A_i$ be the tuple of all units, i.e. $\overline{1}(i)=1_i$. Then $R\cup\{\overline{1}\}$ is also a subalgebra of $\prod\limits_{i=1}^n\m a_i$.
\end{lm}

\begin{proof}
Since $\overline{1}$ is a two-sided neutral element with respect to the operation $\mt$ in $\prod\limits_{i=1}^n\m a_i$, $R\cup\{\overline{1}\}$ is compatible with $\mt$. As for the compatibility with $d$, for any $\overline{a},\overline{b}\in R\cup\{\overline{1}\}$, using equation (3) from Definition~\ref{regSMBdef}, we get
\[
\begin{gathered}
d(\overline{a},\overline{b},\overline{1})=d((\overline{b}\mt\overline{1})\mt\overline{a},(\overline{a}\mt\overline{1})\mt\overline{b},(\overline{a}\mt\overline{b})\mt\overline{1})=\\
d(\overline{b}\mt\overline{a},\overline{a}\mt\overline{b},\overline{a}\mt\overline{b})=\overline{b}\mt\overline{a};\\
d(\overline{a},\overline{1},\overline{b})=d((\overline{1}\mt\overline{b})\mt\overline{a},(\overline{a}\mt\overline{b})\mt\overline{1},(\overline{a}\mt\overline{1})\mt\overline{b})=\\
d(\overline{b}\mt\overline{a},\overline{a}\mt\overline{b},\overline{a}\mt\overline{b})=\overline{b}\mt\overline{a};\\
d(\overline{1},\overline{a},\overline{b})=d((\overline{a}\mt\overline{b})\mt\overline{1},(\overline{1}\mt\overline{b})\mt\overline{a},(\overline{1}\mt\overline{a})\mt\overline{b})=\\
d(\overline{a}\mt\overline{b},\overline{b}\mt\overline{a},\overline{a}\mt\overline{b}).
\end{gathered}
\]
If both $\overline{a},\overline{b}\in R$, then $d(\overline{a}\mt\overline{b},\overline{b}\mt\overline{a},\overline{a}\mt\overline{b})\in R$. If, on the other hand, one of $\overline{a}$ and $\overline{b}$ is equal to $\overline{1}$, then $d(\overline{a}\mt\overline{b},\overline{b}\mt\overline{a},\overline{a}\mt\overline{b})\in\{\overline{a},\overline{b}\}$. In all cases we obtain that $R\cup\{\overline{1}\}$ is closed under $d$.
\end{proof}

In the case of unital SMB algebras, we will modify Theorem~\ref{Zhlemma} to the following stronger statement:

\begin{thm}\label{Aljareduce}
Let $\vr t$ be a template of SMB algebras and let $P=(V,D,\vr c)$ be a multisorted instance of $CSP(\vr t)$ such that each domain of a variable in $P$ is a unital SMB algebra. If $P$ is Z-irreducible, (1,1)-minimal and cycle consistent, then $P$ has a solution $f$ such that for each $i\in V$, $f(i)$ is in the least $\sim$-class in $\m a_i$.
\end{thm}

\begin{proof}
We invoke the proof of Proposition~\ref{SMBtospec} (Proposition 21 from \cite{SMB1}) which involved iterating $\mt$ and substituting $d$ in order to obtain term operations $\mt'$ and $d'$ of $\m a_i$ to obtain $\m a_i'=(A_i;\mt',d')$ which is a regular SMB algebra. Since the terms can be iterated the same way for all $\m a_i\in D$, we define the template $\vr t'$ to consist of the isomorphism types of $\{\m a_i':\m a_i\in D\}$ and their closure under homomorphic images, subalgebras and unary polynomial retracts.

Let us denote the unit element of $\m a_i$ by $1_i$. One can check that $1_i$ is still a two-sided neutral element with respect to $\mt '$ since $x\mt'y$ has the form $x\mt(\dots(x\mt (x\mt y))\dots)$. So $\m a_i'$ are unital and regular SMB algebras, with respect to the same congruences $\sim$ that were used in $\m a_i$.

We construct a new instance $P'=(V,D',\vr c')$, where $D'=\{A_i':i\in V\}$, $\vr c'=\{(S,R'):(S,R)\in \vr c\}$ and for each $(S,R)\in \vr c$, $R'=R\cup\{\lb 1_i:i\in S\rb\}$. In other words, we only changed the constraint relations by adding the tuple of all units to the constraint relations which didn't already have one. By Lemma~\ref{addunittuple} and (1,1)-minimality of $R$s, for each $(S,R')\in \vr c'$ we know $R'\leq_{sd}\prod\limits_{i\in S}\m a_i'$.

We claim that $P'$ is Z-irreducible, (1,1)-minimal and cycle consistent. For Z-irreducibility, note that the scope graph $\Gamma_V$ of $P'$ is the same as the scope graph of $P$. On the other hand, let $\Gamma_1$ be the subgraph of the microstructure graph $\Gamma_P$ induced by $(V_1,\vr c_1)$ and $\Gamma_1'$ be the subgraph of the microstructure graph $\Gamma_{P'}$ induced by $(V_1,\vr c_1')$, where $\vr c_1'=\{(S,R')\in\vr c':(S,R)\in \vr c_1\}$. Since the vertex sets of $\Gamma_1$ and $\Gamma_1'$ are the same, we will prove that the connectivity relations in $\Gamma_1$ and $\Gamma_1'$ are equal. Let $a$ and $b$ be two vertices of $\Gamma_1$. If $a$ and $b$ are connected by a path in $\Gamma_1$, then the same path connects them in $\Gamma_1'$, since the constraint relations of $\Gamma_1'$ contain the appropriate ones of $\Gamma_1$. If, on the other hand, $a$ and $b$ are connected by a path in $\Gamma_1'$ and at any time the newly added tuple of all units is used as an edge $c_{j-1}, E_j, c_j$, this means that $c_{j-1}$ and $c_j$ are both units in their respective domains of variables $\m a_{i_{j-1}}$ and $a_{i_j}$, and that there is a constraint $(S,R')\in\vr c_1'$ such that both $i_{j-1},i_j\in S$. By Corollary~\ref{unitsconnhyper} and using (1,1)-minimality of $\Gamma_1$ (which follows from $(1,1)$-minimality of $\Gamma_P$), $1_{i_{j-1}}$ and $1_{i_j}$ are connected in $\Gamma_1$ and the edge $1_{i_{j-1}}, E_j, 1_{i_j}$ can be replaced by the path that connects $1_{i_{j-1}}$ and $1_{i_j}$. In such a way, we replace all edges of $\Gamma_1'$ which are not in $\Gamma_1$ with paths in $\Gamma_1$ to prove that $a$ and $b$ are connected in $\Gamma_1$. We have proved that either both $\Gamma_1$ and $\Gamma_1'$ are connected, or neither is. In the interesting case for Z-irreducibility, the subinstance $P_1'$ of $P'$ induced by $(V_1,\vr c_1')$ has a connected scope graph, but disconnected microstructure graph. By the above arguments, the same holds for the subinstance $P_1$ of $P$ induced by $(V_1,\vr c_1)$. By the Z-irreducibility of $P$, for any point in the vertex set of $\Gamma_1$ there is a solution $f$ of $P_1$ through that point. But $f$ is a solution of $P_1'$ through the same arbitrarily chosen point, since all constraint relations of $P$ are subsets of the corresponding constraint relations of $P'$. Hence, $P'$ is Z-irreducible.

It remains to prove that $P'$ is (1,1)-minimal and cycle consistent. But this follows from the same properties of $P$, using again the fact that all constraint relations of $P$ are subsets of the corresponding constraint relations of $P'$.

Now we can apply Theorem~\ref{Zhlemma} to finish the proof. By construction, $P'$ has a solution which maps each $i$ to $1_i$. From Theorem~\ref{Zhlemma} follows that there exists a solution $f$ of $P'$ such that for each $i\in V$, $f(i)$ is in the least $\sim$-class in $\m a_i'$. The only thing to check is whether it is possible for some $(S,R')\in\vr c'$ that $f|_S\in R'\setminus R$. This would mean that, for each $i\in S$, $f(i)=1_i$. Since $f(i)$ is in the least $\sim$-class in $\m a_i'$ (which is the least $\sim$-class in $\m a_i$), it implies that, for each $i\in S$, $[1_i]_\sim$ is both the least and the greatest, so it is the only $\sim$-class. Moreover, since $1_i$ is the unit element for the operation $\mt$, $A_i=[1_i]_\sim=\{1_i\}$ holds for all $i\in S$. By $(1,1)$-minimality of $P$, $R=\{\lb 1_i:i\in S\rb\}$, so $f|_S\in R$.
\end{proof}

\subsection*{Proof of Theorem~\ref{M-irrdesired}}

Let $P$ be an instance of $CSP(\m a)$. First, we assume that $\m a$ is regular, convert the instance to a multisorted one, and apply the (2,3)-minimality algorithm and then Zhuk's CHECKIRREDUCIBILITY procedure to either solve $P$ or reduce it to an equivalent (2,3)-minimal and Z-irreducible instance of $CSP(\vr t(\m a))$. 

Now we consider the decomposition $t(P)$ of the instance $P$ via the term $t(x,y)=x\mt y$. Of course, $t(x,t(x,y))=x\mt (x\mt y)=x\mt y=t(x,y)$, so the decomposition exists according to Definition~\ref{decompdef}. We prove the following

\begin{lm}\label{Z-irrtodecomp}
If $P$ is (1,1)-minimal, cycle consistent and Z-irreducible, then the decomposition $t(P)=(V',D',\vr c')$ is (1,1)-minimal, cycle consistent and Z-irreducible, and each domain of a variable in $t(P)$ is a unital SMB algebra.
\end{lm}

\begin{proof}
First we prove that each domain of a variable $\m a_{i,a}\in D'$ is a unital SMB algebra. For each $b\in A_{i,a}$ we have $b=a\mt b'$, and by the regularity of $\m a_i$ it follows that 
\begin{equation}\label{eq1}\tag{*}
\text{for each }b\in A_{i,a}\text{, }a\mt b=a\mt (a\mt b') =a\mt b'=b. 
\end{equation}
This implies that $[a]_\sim$ is the greatest $\sim$-class which intersects $A_{i,a}$. As $a\mt b\sim b$ whenever $[b]_\sim\leq [a]_\sim$, the set of $\sim$-classes which intersects $A_{i,a}$ is exactly those that are below $[a]_\sim$ in the semilattice order of $\sim$-classes of $\m a_i$.

Moreover, the operations of $\m a_{i,a}$ are defined as 
$$b_1\mt^{\sm a_{i,a}} b_2=a\mt (b_1\mt b_2)\text{ and }d^{\sm a_{i,a}}(b_1,b_2,b_3)=a\mt d(b_1,b_2,b_3).$$

Of course, $\sim$ is compatible with all polynomials of $\m a_i$, and hence the restriction of $\sim$ is a congruence of $\m a_{i,a}$. For those $\sim$-classes below $[a]_\sim$ in the semilattice order of $\m a_i/{\sim}$, \eqref{eq1} and Definition~\ref{regSMBdef} (2) imply that $[a]_\sim$ acts as the two-sided neutral element, and hence
\[[x\mt^{\sm a_{i,a}} y]_\sim=[x]_\sim\mt[y]_\sim,\]
so $(A_{i,a}/{\sim};\mt^{\sm a_{i,a}})$ is a semilattice. The same argument implies that $$[d^{\sm a_{i,a}}(x,y,z)]_\sim=[x]_\sim\mt[y]_\sim\mt[z]_\sim,$$
in particular, each $[b]_\sim\cap \m a_{i,a}$ is closed under $d^{\sm a_{i,a}}$. We can compute that for $b_1,b_2\in\m a_{i,a}$ such that $b_1\sim b_2$, 
\[d^{\sm a_{i,a}}(b_1,b_2,b_2)=a\mt d(b_1,b_2,b_2)=a\mt b_1=b_1=d^{\sm a_{i,a}}(b_2,b_2,b_1),\]
i.e. $d^{\sm a_{i,a}}$ is a Mal'cev operation on each $[b]_\sim\cap \m a_{i,a}$.

For any $b_1,b_2\in\m a_{i,a}$ such that $[b_1]_\sim\leq [b_2]_\sim$,
\[b_1\mt^{\sm a_{i,a}} b_2=a\mt (b_1\mt b_2)=a\mt b_1=b_1,\]
so $\m a_{i,a}$ is an SMB algebra with respect to the restriction of $\sim$ to $A_{i,a}$ which satisfies (1) and (2) of Definition~\ref{regSMBdef}.

Now for each $x\in A_{i,a}$, 
\[
\begin{gathered}
a\mt^{\sm a_{i,a}} x=a\mt (a\mt x)=a\mt x=x \text{ and}\\
x\mt^{\sm a_{i,a}} a=a\mt (x\mt a)=a\mt x=x.
\end{gathered}
\]
In the above computations we used Definition~\ref{regSMBdef} (2) and (4), since $\m a_i$ is a regular SMB algebra, and also \eqref{eq1}. Thus, we know that each $\m a_{i,a}$ is a unital SMB algebra.

{\bf $\mathbf{t(P)}$ is (1,1)-minimal.} Let $b\in A_{i,a}$ and let $Q$ be a constraint relation having $(i,a)$ in its scope. $Q$ can take one of two forms: if $Q=R_r$, then there is a constraint $(S,R)\in\vr c$ and $r\in R$ such that $i\in S$ and $r(i)=a$. By (1,1)-minimality of $P$, there exists a tuple $q\in R$ such that $q(i)=b$. Then $r\mt q\in R_r$ and $r(i)\mt q(i)=a\mt b=b$, the last equality by \eqref{eq1}. On the other hand, if $Q=T_i$, then there exists a tuple in $T_i$ of the form $\lb c\mt b:c\in A_i\rb$. At the coordinate $(i,a)$, this tuple equals $a\mt b=b$, again by \eqref{eq1}. Either way, we obtain that $t(P)$ is (1,1)-minimal.

{\bf $\mathbf{t(P)}$ is cycle consistent.} Take $(i,a)\in V'$, $b\in A_{i,a}$ and a path 
\[(i,a)=(i_0,a_0),Q_1,(i_1,a_1),Q_2,\dots,Q_k,(i_k,a_k) = (i,a)\]
in $\Gamma_{V'}$. Let $Q_{j_1}=(S_1)_{r_1},\dots,Q_{j_v}=(S_v)_{r_v}$ be the subsequence of $Q$s consisting of those scopes that are scopes of constraint of the form $(S_r,R_r)$, while all other $Q_j$ are scopes of constraints of the form $(S_j,T_j)$ for some $j\in V$. Then 
\[i=i_{j_0},S_1,i_{j_1},\dots,S_v,i_{j_v}=i\]
is a path in $\Gamma_V$. By the cycle consistency of $P$, there exists a path
\[
b=b_0,E_1,b_1,\dots,E_v,b_v=b
\]
in $\Gamma_P$, where for each $0<u<v$, $b_u\in A_{j_u}$, while for each $0<u\leq v$, $E_u$ is the set of all coordinates of some tuple $t_u\in R_u$ and $(S_u,R_u)\in \vr c$.

Now we construct the path in $\Gamma_{t(P)}$ which verifies the cycle consistency of $t(P)$. The path $p$ we will use is
\[c_0,E_1',c_1,\dots,E_k',c_k.\]
Here 
\begin{itemize}
\item for all $u<j_1$, $c_u=a_u\mt b_0$,
\item for all $j_\ell\leq u<j_{\ell+1}$, $c_u=a_u\mt b_\ell$, while
\item for all $u\geq j_v$, $c_u=a_u\mt b_v$.
\end{itemize}
As for the edges,
\begin{itemize}
\item when $u=j_\ell$, $1\leq \ell\leq v$, then $E_u'$ is the set of all coordinates of the tuple $r_\ell\mt t_\ell\in (R_\ell)_{r_\ell}$,
\item if $u<j_1$, then $E_u'$ is the set of all coordinates of the tuple $\lb c\mt b_0:c\in A_{j_0}\rb\in T_{i_{j_0}}=T_i$, 
\item if $j_{\ell}<u<j_{\ell+1}$ (where $0<\ell< v$), then $E_u'$ is the set of all coordinates of the tuple $\lb c\mt b_\ell:c\in A_{j_\ell}\rb\in T_{j_\ell}$, while
\item if $j_v<u$, then $E_u'$ is the set of all coordinates of the tuple $\lb c\mt b_v:c\in A_{j_v}\rb\in T_{j_v}=T_i$.
\end{itemize}

We need to prove that $p$ is indeed a path from $b$ to $b$ traversing the desired domains of variables and constraint relations of $t(P)$.

First of all, \[c_0=a_0\mt b_0=a\mt b=b\in A_{j_0,a_0}=A_{i,a}\] 
and 
\[c_k=a_k\mt b_v=a\mt b=b\in A_{j_v,a_k}=A_{i,a}.\]
Next, note that when $Q_\ell$ is the scope of some constraint of the form $(S_j,T_j)$ for $j\in V$, then $i_{\ell-1}=i_\ell$ since all domains of variables in the scope of $T_j$ are of the form $(j,c)=(i_\ell,c)$ for various $c\in A_{i_\ell}$. Therefore, 
\begin{itemize}
\item for all $u<j_1$, $i_u=i_{j_0}=i$,
\item for all $j_{\ell}\leq u<j_{\ell+1}$, $i_u=i_{j_{\ell}}$, while
\item for all $u\geq j_v$, $i_u=i_{j_v}=i$.
\end{itemize}
Using the above, we can check that
\begin{itemize}
\item for all $u<j_1$, $c_u=a_u\mt b_0\in A_{i_{j_0},a_u}=A_{i_u,a_u}$,
\item for all $j_\ell\leq u<j_{\ell+1}$, $c_u=a_u\mt b_\ell\in A_{i_{j_\ell},a_u}=A_{i_u,a_u}$, while
\item for all $u\geq j_v$, $c_u=a_u\mt b_v\in A_{i_{j_v},a_u}=A_{i_u,a_u}$.
\end{itemize}
So the vertices $c_u$ along the path are indeed in the desired domains $A_{i_u,a_u}$.

Each edge $E'_{j_\ell}$ is the set of all coordinates of the tuple $t_\ell\mt r_\ell$ in $(R_\ell)_{r_\ell}$. Moreover, 
\[
\begin{gathered}
t_\ell(i_{j_\ell}-1)=t_\ell(i_{j_{\ell-1}})=b_{\ell-1}\text{, }t_\ell(i_{j_\ell})=b_\ell,\\
r_\ell(i_{j_\ell}-1)=a_{i_{j_\ell}-1}\text{ and }r_\ell(i_{j_\ell})=a_{i_{j_\ell}}.
\end{gathered}
\]
Therefore, the edge $E'_{j_\ell}$ connects 
\[c_{i_{j_\ell}-1}=a_{i_{j_\ell}-1}\mt b_{\ell-1}=r_{\ell}(i_{j_\ell}-1)\mt t_\ell(j_{\ell-1})\]
with 
\[c_{i_{j_\ell}}=a_{i_{j_\ell}}\mt b_{\ell}=r_{\ell}(i_{j_{\ell}})\mt t_\ell(i_{j_{\ell}}),\]
as desired. Here we use that $i_{j_{\ell-1}}=i_{j_{\ell-1}}+1=\dots=i_{j_{\ell}}-1$ which we proved above.

Next, we analyze the edges $E_u'$ when $u$ is not equal to any $j_\ell$. If $u<j_1$, then $E_j'$ is the set of all coordinates of the tuple
\[\lb c\mt b_0:c\in A_{j_0}\rb=\lb c\mt b:c\in A_i\rb.\]
(Remark: the tuple doesn't change for all coordinates $u<i_1$.) At coordinates $(i_{u-1},a_{u-1})=(i_0,a_{u-1})$ and $(i_u,a_u)=(i_0,a_u)$ the above tuple equals $a_{u-1}\mt b_0=c_{u-1}$ and $a_u\mt b_0=c_u$, respectively. An analogous argument proves the desired connections in the cases $j_{\ell-1}<u<j_\ell$ and $j_v<u$, completing the proof of cycle consistency.

{\bf $\mathbf{t(P)}$ is Z-irreducible.} Assume that $P_1'$ is a subinstance of $t(P)$ induced by $(V_1',\vr c_1')$ and that $P_1'$ has a connected scope graph. We ``project" the variables in $V_1'$ to their first coordinates, to obtain 
\[V_1=\{i\in V:(\exists a\in A_i)(i,a)\in V_1\}.\] 
The same ``projection" can be applied to constraints in $\vr c_1'$ to obtain
\[\vr c_1=\{(S,R)\in \vr c:(\exists r\in R)(S_r,R_r)\in \vr c_1'\}.\]

{\bf Claim 1.} If the subinstance $P_1'$ of $t(P)$ induced by $(V_1',\vr c_1')$ has a connected scope graph, then so does the subinstance $P_1$ of $P$ induced by $(V_1,\vr c_1)$. 

{\em Proof of Claim 1}. If we consider the scopes of the constraints of the form $(S_i,t_i)$, where $i\in V$, as $S_i=\{(i,a):a\in A_i\}$, we see that they are pairwise disjoint as the ``projection" of each such scope to the first coordinate is $\{i\}$. Let us assume that the scope graph of $P_1'$ is connected and that $i,j\in V_1$. Then there exist some variables $(i,a)$ and $(j,b)$ in $V_1'$. By the connectedness of the scope graph of $P_1'$, there must exist a path in that graph connecting $(i,a)$ and $(j,b)$. We can project the whole path to the first coordinates and conclude that there is a path in the scope graph of $P_1$ from $i$ to $j$ using only the projections of the scopes of the form $S_r$ (as the hyperedges obtained by projecting the other type of scopes become singletons, which can be omitted).

Next we want to prove that

{\bf Claim 2.} Let $p'$ be a path in the scope graph of $P_1'$ of the form
\[(i_0,a_0),Q_1,(i_1,a_1),\dots,Q_k,(i_k,a_k),\]
such that $i_0=i_k$ and let $b\in A_{i_0}$. Then there exists a corresponding path $q'$ in the microstructure graph of $P_1'$,
\[a_0\mt b=c_0,E_1,c_1,\dots,E_k,c_k=a_k\mt b\]
such that each hyperedge $E_i$ is the set of all coordinates of some tuple in the constraint relation $R_i'$ corresponding to the scope $Q_i$.

{\em Proof of Claim 2.} Let $P_1''$ be the extension of $P_1'$ obtained by adding the hyperedge $S_{i_0}=S_{i_k}$ to the scope graph. If we extend the path $p'$ by just one edge and vertex we get the path $p'':=p',S_{i_0},(i_k,a_0)=(i_0,a_0)$, which may not be a path in the scope graph of $P_1'$, but is a path in the scope graph of $P_1''$. We want to prove that there exists a path $q''$ in the microstructure graph of $P_1''$ corresponding to the path $p''$ whose first vertex is $a_0\mt b$ and last edge is $a_k\mt b, R_{i_0}, a_0\mt b$. By deleting the last edge from the path $q''$ we would obtain the desired path $q'$ in the microstructure graph of $P_1'$.

We have already proved a very similar claim in the course of the proof of cycle consistency of $t(P)$. The only difference is that we assumed that $a_0\mt b=b$, which was more of a convenience than a real requirement. Using an analogous argument as the one we made in the cycle consistency, together with a cyclic path $p$ in the microstructure graph of $P_1$ from $b$ to $b$, we prove the existence of the corresponding cyclic path $q''$ in the microstructure graph of $P_1''$ from $a_0\mt b$ to $a_0\mt b$. Since the last vertex in the path $p$ is $b$ and the last edge in the path $p''$ is $(i_0,a_k), S_{i_0},(i_0,a_0)$, the last edge in the path $q''$ must be $a_k\mt b, R_{i_0}, a_0\mt b$. As we already said, deleting this last edge from $q''$ proves the existence of the desired path $q'$ and thus, proves Claim 2.

{\bf Claim 3.} If the microstructure graph of $P_1$ is connected and the scope graph of $P_1'$ is connected, then the microstructure graph of $P_1'$ is connected.

{\em Remark}. The contrapositive of the statement of Claim 3 is what we need for Z-irreducibility. We can not prove Claim 3 in general, but we can in our more restrictive setting, recalling that $P$ and $t(P)$ are both (1,1)-minimal and cycle consistent.

{\em Proof of Claim 3}. Assume that $b_1\in A_{i,a_1}$ and $b_2\in A_{j,a_2}$. Then  there is a path $p$ in the microstructure graph of $P_1$ from $b_1\in A_i$ to $b_2\in A_j$. Let this path $p$ be
\[b_1=c_0,E_1,c_1,\dots, E_k,c_k=b_2,\]
where $c_i\in A_{j_i}$, for some $i=j_0,j_1,\dots,j_k=j\in V_1$, $E_u$ is the set of all coordinates of a tuple $t_u\in R_u$, while $(S_u,R_u)\in \vr c_1$ are constraints such that $j_{u-1},j_u\in S_u$.

By the definition of $\vr c_1$, there must exist tuples $r_u\in R_u$ and constraints $(S_{r_u},R_{r_u})\in\vr c_1'$. Hence there are hyperedges $E_1',E_2',\dots,E_k'$ in the microstructure graph of $P_1'$ such that $E_u'$ connects 
\[r_u(j_{u-1})\mt c_{u-1}\in A_{j_{u-1},r_u(j_{u-1})}\text{ to }r_u(j_u)\mt c_u\in A_{j_u,r_u(j_u)}.\]
It remains to find paths $q_0,q_1,\dots,q_k$ in the microstructure graph of $P_1'$ such that $q_0$ connects
\[b_1=c_0=a_1\mt c_0\in A_{j_0,a_1}\text{ to }r_1(j_0)\mt c_0\in A_{j_0,r_1(j_0)},\]
for each $0<u<k$, $q_u$ connects 
\[r_u(j_u)\mt c_u\in A_{j_u,r_u(j_u)}\text{ to }r_{u+1}(j_u)\mt c_u\in A_{j_u,r_{u+1}(j_u)},\]
while $q_k$ connects
\[r_k(j_k)\mt c_k\in A_{j_k,r_k(j_k)}\text{ to }b_2=c_k=a_2\mt c_k\in A_{j_k,a_2}.\]
But, these paths $q_0,q_1,\dots,q_k$ are precisely what is guaranteed by Claim 2. Thus, the microstructure graph of $P_1'$ is connected, so we proved Claim 3.

Now we finish the proof of the lemma. Assume that the subinstance $P_1'$ of $t(P)$ has a connected scope graph, but disconnected microstructure graph. The corresponding subinstance $P_1$ of $P$ also has a connected scope graph by Claim 1, and by the contrapositive of Claim 3 we obtain that $P_1$ also has a disconnected microstructure graph.

Let $A_{i,a}$ be any domain of a variable of $P_1'$ and let $b\in A_{i,a}$ be any point in $A_{i,a}$. We know that $b=b\mt a$ and that $b\in A_i$. Since $P$ is $Z$-irreducible, there exists a solution $f$ of $P_1$ such that $f(i)=b$. Just like in the proof of Lemma~\ref{l:tA-solution}, for any $(j,c)\in V_1'$, we define $g(j,c):=c\mt f(j)$. (Note that from the definition of $P_1$ and $(j,c)\in V_1'$ follows that $j\in V_1$.) Analogously as in the proof of Lemma~\ref{l:tA-solution}, we obtain that $g$ is a solution of $P_1'$. Moreover, $g(i,a)=a\mt f(i)=a\mt b=b$. So, $t(P)$ is Z-irreducible and Lemma~\ref{Z-irrtodecomp} is proved.
\end{proof}

Now we complete the proof of Theorem~\ref{M-irrdesired}. As $P$ is either solved or reduced to a (2,3)-minimal and Z-irreducible instance, it follows that $P$ can be assumed to be both (1,1)-minimal and cycle consistent, as both are weaker notions than (2,3)-minimality. So Lemma~\ref{Z-irrtodecomp} implies that the decomposition $t(P)$ is (1,1)-minimal, cycle consistent, Z-irreducible and each domain of a variable in $t(P)$ is a unital SMB algebra (though maybe not regular any longer). Nevertheless, by Theorem~\ref{Aljareduce}, $t(P)$ has a solution $f$ such that for each $(i,a)\in V'$, $f(i,a)$ is in the least $\sim$-class of $\m a_{i,a}$. 

Now we can mimic the proof of Lemma~\ref{l:withC} to either solve $P$ or reduce it to an equivalent smaller instance. The difference is that in the second paragraph of the proof of Lemma~\ref{l:withC} we recursively invoked $t(P)$ as an instance of a smaller template. Now we are able to, instead, apply the just-proved fact that $t(P)$ has a solution. This allows us to reduce the instance $P$ until it becomes M-irreducible avoiding the potential vicious circle.\qed

\section{A better fix for the gap}

We have plugged the gap in A. Bulatov's paper \cite{BuSMB}, but at a heavy cost. To prove the tractability of SMB algebras, a special case of Taylor algebras, we used Theorem~\ref{Zhlemma}, a major part of D. Zhuk's proof of tractability of Taylor algebras. Even worse, the structure of Zhuk's proof makes each piece inseparable from the rest of the proof, so if we were to write out all proofs of the facts we used in the previous section, we would be forced to include the full proof of the tractability of Taylor algebras by Zhuk, i.e. of the Dichotomy Conjecture.

Fortunately, a better fix for the same gap can be made using just the ideas in A. Bulatov's original paper \cite{BuSMB} and tweaking them a little. In order to demonstrate this fix, we proceed to expose definitions and statements of results from \cite{BuSMB}. The issue comes up when we want to solve the CSP over restrictions of the instance $(V,D,\vr c)$ to ``coherent sets" $W\subseteq V$. We are able to solve the instance $P|_W$ using just Theorem \ref{M-irrweaker}, instead of Theorem~\ref{M-irrdesired}.

In the section that follows, we have changed some of the notions and results from \cite{BuSMB} beyond what is needed to solve the reduced instances $P|_W$. For example, besides the key notion of block-minimality, we also define a weaker one, block-2-consistency. In effect, after solving $P|_W$, we are able to finish the tractability proof even if we enforce less consistency on the projection of $P$ to coherent sets than A. Bulatov did. We also include a few easy observations which A. Bulatov uses, though he doesn't explicitly state them and correct a minor, fairly obvious, error. In all, we feel that our variant of A. Bulatov's proof is more ``user-friendly" than the original. However, we wish to state that what we present here is still only a variant of the proof by A. Bulatov, since the gap in his proof which we found and fixed turned out to be solvable using the ideas in his original paper.

\subsection{Rees congruences}

\begin{df}
Let $\m a$ be a finite regular SMB algebra over $\sim$. If the least block of $\m a$ is $B$, we call the {\em Rees congruence} of $\m a$ the relation $B^2\cup\Delta_A$. The Rees congruence is denoted by $\theta_{\m A}$. Moreover, we will denote the subuniverse $B$ as $\mathrm{min}(\m a)$.
\end{df}

The Rees congruence is a congruence of the regular SMB algebra $\m a$ since $\mathrm{min}(\m a)$ is a strongly absorbing subuniverse (for any term operation, if an element of $\mathrm{min}(\m a)$ is in an essential position, then the result is in $\mathrm{min}(\m a)$). The term Rees congruence comes from Semigroup Theory.

\begin{lm}\label{minR}
If $\m R\leq_{sd}\m a_1\times\dots\times\m a_n$, where all $\m a_i$ are regular SMB algebras with $\theta_i$ the Rees congruence of $\m a_i$, $1\leq i\leq n$, then the Rees congruence of $\m r$ is the restriction of the product congruence $\theta_1\times\dots\times\theta_n$ to $R$, while $\mathrm{min}(\m r)$ is $(\mathrm{min}(\m a_1)\times\dots\times\mathrm{min}(\m a_n))\cap R$.
\end{lm}

\begin{proof}
Since $R$ is subdirect, we can select $\mathbf{a}_1,\dots,\mathbf{a}_n\in R$ so that $\mathbf{a}_i(i)\in \mathrm{min}(\m a_i)$, and $$\mathbf{a}:=(\dots(\mathbf{a}_1\mt\mathbf{a}_2)\mt\dots)\mt\mathbf{a}_n$$ must be in $(\mathrm{min}(\m a_1)\times\dots\times\mathrm{min}(\m a_n))\cap R$. As $(\mathrm{min}(\m a_1)\times\dots\times\mathrm{min}(\m a_n))\cap R$ is nonempty, it must be the least $\sim$-class of $\m R$.
\end{proof} 

\begin{lm}[Lemma 11 of \cite{BuSMB}]\label{minsetinmin}
Let $\m a$ be a finite regular SMB algebra with the Rees congruence $\theta_A$, $0_A\leq \alpha\prec\beta\leq \theta_{\m A}$. Then for any $U\in M_{\m a}(\alpha,\beta)$ we have $U\subseteq\mathrm{min}(\m a)$.
\end{lm}

\begin{proof}
Let $f\in\poln{1}{a}$ be such that $f(A)=U\in M_{\m a}(\alpha,\beta)$ and for all $x\in A$, $f(f(x))=f(x)$. Then $f(\beta)\nsubseteq\alpha$. Select some $(a,b)\in f(\beta)\setminus\alpha$. Hence $a\neq b$. As $f(\beta)\subseteq\beta\subseteq\theta_{\m A}=0_A\cup(\mathrm{min}(\m a)\times \mathrm{min}(\m a))$, this means that $(a,b)\in (\mathrm{min}(\m a)\times \mathrm{min}(\m a))$. Define $g(x)=f(x)\mt a$.

We know that for all $x\in U\cap \mathrm{min}(\m a)$, $f(x)=x\in \mathrm{min}(\m a)$ and therefore $$f(g(x))=f(f(x)\mt a)=f(x\mt a)=f(x)=x.$$ On the other hand, since $\m a$ is a regular SMB algebra, any nonconstant unary polynomial maps $\mathrm{min}(\m a)$ into $\mathrm{min}(\m a)$. Therefore, for any $x\in A$, $g(x)=f(x)\mt a\in\mathrm{min}(\m a)$, and therefore $$f(g(x))\in U\cap\mathrm{min}(\m a).$$
So we have that $f(g(A))=U\cap \mathrm{min}(\m a)\subseteq U$ and $(f(g(a)),f(g(b)))=(a,b)\in \beta\setminus\alpha$, and thus $f(g(\beta))\nsubseteq\alpha$. Since $U\in M_{\m a}(\alpha,\beta)$, it follows that $U=U\cap\mathrm{min}(\m a)$.
\end{proof}

\subsection{Separation}

We start the subsection by reminding the reader of the following fundamental early result of Tame Congruence Theory.

\begin{thm}{\rm (Theorem 2.8 of \cite{tct}, statements (1), (3), (4) and (6))}\label{TCTfund}
Let $\m a$ be a finite algebra and $\alpha\prec\beta$ in $\cn a$. 
\begin{enumerate}
\item If $U,V\in M_{\m a}(\alpha,\beta)$, then there exists $f\in \poln{1}{a}$ such that $f|_U$ bijectively maps $U$ to $V$.
\item Let $U\in M_{\m a}(\alpha,\beta)$ and $f\in\poln{1}{a}$ satisfy $f(\beta|_U)\nsubseteq\alpha$. Then $f(U)\in M_{\m a}(\alpha,\beta)$
\item If $(a,b)\in\beta\setminus\alpha$ and $U\in M_{\m a}(\alpha,\beta)$, then there exists $f\in\poln{1}{a}$ such that $f(A)=U$ and $(f(a),f(b))\in\beta\setminus\alpha$.
\item If $g\in\poln{1}{a}$ satisfies $g(\beta)\nsubseteq\alpha$, then there exists a $U\in M_{\m a}(\alpha,\beta)$ such that $g(U)\in M_{\m a}(\alpha,\beta)$.
\end{enumerate}
\end{thm}

Now we define separation.

\begin{df}\label{sep1}
Let $\m a$ be an algebra and let $\alpha\prec\beta$ and $\gamma\prec\delta$ in $\cn a$. We say that $(\alpha,\beta)$ {\em can be separated from} $(\gamma,\delta)$ if there exists a polynomial $f\in \mathrm{Pol}_1\m a$ such that $f(\beta)\nsubseteq\alpha$, but $f(\delta)\subseteq\gamma$.
\end{df}

\begin{prp}\label{minsetsarethesame}
Let $\alpha\prec\beta$ and $\gamma\prec\delta$ in $\cn A$ be such that $(\alpha,\beta)$ can't be separated from $(\gamma,\delta)$ and $(\gamma,\delta)$ can't be separated from $(\alpha,\beta)$. Then for any subset $U\subseteq A$, $U\in M_{\m a}(\alpha,\beta)$ iff $U\in M_{\m a}(\gamma,\delta)$.
\end{prp}

\begin{proof}
Let $U\in M_{\m a}(\alpha,\beta)$ and $U=f(A)$ for some idempotent unary polynomial $f\in\poln{1}{a}$. Since $f(\beta)\nsubseteq \alpha$, then $f(\delta)\nsubseteq\gamma$. It follows that there exists some $U'\in M_{\m a}(\gamma,\delta)$ and an idempotent unary polynomial $g\in\poln{1}{a}$ such that $g(A)=U'\subseteq U$. Since $g$ is idempotent, $U'=g(U')\subseteq g(U)\subseteq g(A)=U'$ and hence $g(f(A))=g(U)=U'$. Since $g(f(x))$ is the identity map on $U'$ and $\delta|_{U'}\nsubseteq\gamma|_{U'}$, hence $g(f(\delta))\nsubseteq\gamma$, thus $g(f(\beta))\nsubseteq\alpha$. As $g(f(A))=U'\subseteq U$ and $U\in M_{\m a}(\alpha,\beta)$, it follows that $U'=U$, so $U\in M_{\m a}(\gamma,\delta)$. The reverse implication is analogous.
\end{proof}

The separation which interests us is a separation with respect to a subdirect product. When $\m R\leq_{sd}\m a_1\times\dots\times\m a_n$ and $f\in\poln{1}{r}$, we denote by $f_i$ the polynomial of $\m a_i$ constructed from the same term as $f$, but such that the parameters of $f$ are replaced with their $i$th components.

\begin{df}\label{sep2}
Let $\m R\leq_{sd}\m a_1\times\dots\times\m a_n$, let $\alpha\prec\beta$ in $\mathrm{Con}\:\m a_i$ and $\gamma\prec\delta$ in $\mathrm{Con}\:\m a_j$. We say that $(\alpha,\beta)$ {\em can be separated from} $(\gamma,\delta)$ {\em with respect to} $\m R$ if there exists a polynomial $f\in \mathrm{Pol}_1\m r$ such that $f_i(\beta)\nsubseteq\alpha$, but $f_j(\delta)\subseteq\gamma$. 

Let $\vr t$ be a template of regular SMB algebras and let $P=(V,D,\vr c)$ be a (2,3)-minimal multisorted instance of $CSP(\vr t)$ and for all $i\in V$, let $\m a_i$ be the domain of variable $i$, with $\theta_i$ the Rees congruence of $\m a_i$. If $i,j\in V$, $\alpha\prec\beta$ in $\mathrm{Con}\:\m a_i$ and $\gamma\prec\delta$ in $\mathrm{Con}\:\m a_j$, we say that $(\alpha,\beta)$ can be separated from $(\gamma,\delta)$ with respect to $P$ if $(\alpha,\beta)$ can be separated from $(\gamma,\delta)$ with respect to $\m R$, where $(S,R)\in\vr c$ is a constraint such that $i,j\in S$.
\end{df}

In the above definition the separation with respect to $P$ does not depend on the choice of $(S,R)$. To see this note that $(\alpha,\beta)$ can be separated from $(\gamma,\delta)$ with respect to $\m R$ iff $(\alpha,\beta)$ can be separated from $(\gamma,\delta)$ with respect to the projection of $\m R$ to the set $\{i,j\}$ (as before, $\alpha\prec\beta$ in $\mathrm{Con}\:\m a_i$ and $\gamma\prec\delta$ in $\mathrm{Con}\:\m a_j$). However, by the (2,3)-minimality of $P$, given any $(S,R)\in\vr c$ such that $i,j\in S$, the projection $\mathrm{pr}_{\{i,j\}}R$ is always the same relation.

\begin{df}
Let $\vr t$ be a template of regular SMB algebras and let $P=(V,D,\vr c)$ be a (2,3)-minimal multisorted instance of $CSP(\vr t)$ and let $\m a_i$ be the domain of variable $i$, with $\theta_i$ the Rees congruence of $\m a_i$. By $\vr i_{P}$ we denote the set \[\{(i,\alpha,\beta):i\in V,\: 0_i\leq \alpha\prec\beta\leq \theta_i)\}.\]
Similarly, if $\m R\leq_{sd}\m a_1\times\dots\times\m a_n$, where for each $i\leq n$, $\m a_i$ is a finite regular SMB algebra and $\theta_i$ is the Rees congruence, then by $\vr i_{\m R}$ we denote the set \[\{(i,\alpha,\beta): i\leq n,\: 0_i\leq \alpha\prec\beta\leq \theta_i)\}.\]
\end{df}

\begin{lm}\label{sep1insep2}
Let $\m R\leq_{sd}\m a_1\times\dots\times\m a_n$ and let $\alpha\prec \beta$ and $\gamma\prec\delta$ both be covering pairs in the congruence lattice $\cn a_i$ (for the same $i$). Then $(\alpha,\beta)$ can be separated from $(\gamma,\delta)$ in the sense of Definition~\ref{sep1} iff $(\alpha,\beta)$ can be separated from $(\gamma,\delta)$ with respect to $\m R$ in the sense of Definition~\ref{sep2}.
\end{lm}

\begin{proof}
The proof follows from the subdirectness of $\m R$. If $f\in \poln{1}{a}_i$ is selected such that $f(\beta)\nsubseteq\alpha$, but $f(\delta)\subseteq\gamma$, then let $f(x)=t(x,c_1,\dots,c_n)$ for some $(n+1)$-ary term $t$ of $\m a_i$. As $\m R$ is subdirect, there exist tuples $\mathbf{d}_1,\dots,\mathbf{d}_n\in R$ such that $\mathbf{d}_1(i)=c_1,\dots,\mathbf{d}_n(i)=c_n$. Let the polynomial $g(x)\in\poln{1}{R}$ be given by $g(x)=t(x,\mathbf{d}_1,\dots,\mathbf{d}_n)$. Then $g_i=f$, and hence $g_i(\beta)\nsubseteq\alpha$, but $g_i(\delta)\subseteq\gamma$. By Definition~\ref{sep2}, $(\alpha,\beta)$ can be separated from $(\gamma,\delta)$ with respect to $R$. The reverse direction is analogous.
\end{proof}

\begin{lm}[Lemmas 15 and 16 of \cite{BuSMB}]\label{septominimal}
Let $\m R\leq_{sd}\m a_1\times\dots\times\m a_n$, where for each $i\leq n$, $\m a_i$ is a finite regular SMB algebra and let $(i,\alpha,\beta),(j,\gamma,\delta)\in\vr i_{\m r}$. If $U\in M_{\m a}(\alpha,\beta)$ and $(\alpha,\beta)$ can be separated from $(\gamma,\delta)$ with respect to $\m r$, then there exists $f\in\poln{1}{r}$ such that $f_i(A_i)=U$, $f_j(\gamma)\subseteq\delta$, $f$ is an idempotent polynomial of $\m r$ and for each $k\leq n$, $f_k(A_k)\subseteq \mathrm{min}(\m a_k)$.
\end{lm}

\begin{proof}
Let $g\in\poln{1}{r}$ be such that $g_i(\beta)\nsubseteq\alpha$, but $g_j(\delta)\subseteq\gamma$. By Theorem~\ref{TCTfund} (4), there must be at least one $U'\in M_{\m a}(\alpha,\beta)$ such that $g_i(U')=U''$ and $U''\in M_{\m a}(\alpha,\beta)$. Let $h'\in\poln{1}{a_i}$ be an idempotent polynomial such that $h'(A_i)=U$. By replacing the parameters of $h'$ with elements of $R$ whose $i$th components are those parameters we get a polynomial $g'\in\poln{1}{r}$ such that $g'_i=h'$ (note that $g'$ might not be an idempotent polynomial). By Theorem~\ref{TCTfund} (1), there must exist polynomials $h'',h'''\in\poln{1}{a_i}$ such that $h''(U)=U'$ and $h'''(U'')=U$. Let $g'',g'''\in\poln{1}{r}$ be such that $g''_i=h''$ and $g'''_i=h'''$. Finally, let $\mathbf{a}\in\mathrm{min}(\m R)=R\cap\prod\limits_{k=1}^n\mathrm{min}(\m a_k)$ and $f'\in\poln{1}{r}$ be defined by
\[f'(x)=[(g'''\circ g\circ g''\circ g')(x)]\mt\mathbf{a}.\]

For any elements $c,d\in A_j$ such that $(c,d)\in\delta$, $g''_j\circ g'_j$ maps $(c,d)$ to $(c',d')\in \delta$, then $g_j$ maps $(c',d')$ to a pair $(c'',d'')\in\gamma$, and any further application of unary polynomials of $\m a_j$ maps $(c'',d'')$ into $\gamma$. Thus $f'_j(\delta)\subseteq\gamma$. Clearly, $x\mt \mathbf{a}$ maps $R$ into $\mathrm{min}(\m R)$, so for all $k\leq n$, $f'_k(A_k)\subseteq \mathrm{min}(\m a_k)$.

Now we consider $f'_i$. First of all, for any $a\in A_i$, $g'_i(a)=b\in U$, and then $g'''_i\circ g_i\circ g''_i$ maps $b$ first to $U'$, then to $U''$ and finaly back to $U$, so $(g'''_i\circ g_i\circ g''_i\circ g'_i)(a)=b'\in U$. By Lemma \ref{minsetinmin}, we know that $U\subseteq\mathrm{min}(\m a_i)$, and hence $b'\mt \mathbf{a}(i)=b'$. We conclude that $f'_i(A_i)\subseteq U$.

On the other hand, for elements of $U$, we get that $g'_i$ restricts as the identity map on $U$, then that $g'''_i\circ g_i\circ g''_i=h'''\circ g_i\circ h''$ bijectively maps $U$ onto $U$, and finally that $x\mt  \mathbf{a}(i)=x$ for all $x\in U$. Therefore, $f'_i$ restricts to $U$ as a permutation.

Let $f$ be an idempotent power of $f'$ (for example, $f=f'^{(|R|!)}$ would work). It is easy to check that $f_j(\delta)\subseteq \gamma$ and  for all $k\leq n$, $f'_k(A_k)\subseteq \mathrm{min}(\m a_k)$. Moreover, $f_i$ restricts to $U$ as a permutation (actually, it is the identity map), so $f_i(\beta)\nsubseteq \alpha$.
\end{proof}

\begin{lm}[Lemma 17 of \cite{BuSMB}]\label{sepissymmetric}
Let $\m R\leq_{sd}\m a_1\times\dots\times\m a_n$, where for each $i\leq n$, $\m a_i$ is a finite regular SMB algebra and let $(i,\alpha,\beta),(j,\gamma,\delta)\in\vr i_{\m r}$. If $(\alpha,\beta)$ can be separated from $(\gamma,\delta)$ with respect to $\m r$, then $(\gamma,\delta)$ can be separated from $(\alpha,\beta)$ with respect to $\m r$.

Also, let $\vr t$ be a template of regular SMB algebras, $P=(V,D,\vr c)$ a (2,3)-minimal multisorted instance of $CSP(\vr t)$ and $(i,\alpha,\beta),(j,\gamma,\delta)\in\vr i_P$. If $(\alpha,\beta)$ can be separated from $(\gamma,\delta)$ with respect to $P$, then $(\gamma,\delta)$ can be separated from $(\alpha,\beta)$ with respect to $P$.
\end{lm}

\begin{proof}
Let $B$ be some $\delta$-class which contains more than one $\gamma$-class and let $V\subseteq A_i$ be minimal-sized subset so that there exists a polynomial $h\in\poln{1}{R}$ such that
\begin{enumerate}
\item $h(R)\subseteq\mathrm{min}(\m r)$,
\item For all $x\in B$, $h_j(x)\equiv_\gamma x$ and
\item $V=h_i(A_i)$.
\end{enumerate}
Obviously, at least some $V'$ which satisfies $(1)$, $(2)$ and $(3)$ exists, as we can take $h(x)=x\mt \mathbf{b}$ for some $\mathbf{b}\in \mathrm{min}(\m r)$ and we would get $V'=\mathrm{min}(\m a_i)$. So it makes sense to pick the minimal-sized $V$.

Suppose that $h\in\poln{1}{r}$ is as above, so $h_i(A_i)=V$ and suppose that $h_i(\beta)\nsubseteq\alpha$. Then there exists some $U\in M_{\m a}(\alpha,\beta)$ such that $U\subseteq V$. According to Lemma~\ref{septominimal}, there exists an idempotent polynomial $g\in\poln{1}{r}$ such that $g_i(A_i)=U$, $g_j(\delta)\subseteq\gamma$ and for each $k\leq n$, $g_k(A_k)\subseteq \mathrm{min}(\m a_k)$. Select $\mathbf{a}\in R$ such that $\mathbf{a}(j)\in B$ and let 
$$f(x):=d(x,g(x),g(\mathbf{a})).$$
Since $g$ is idempotent, then $g_i$ is also an idempotent polynomial of $\m a_i$ and this means that $g_i|_U$ is the identity map on the set $U$. Now, for any $x\in U$, 
$$f_i(x)=d(x,x,g_i(\mathbf{a}(i)))=g_i(\mathbf{a}(i)).$$
Therefore, $|f_i(U)|=1$ and since $|U|>1$ and $U\subseteq V$, for $W:=f_i(h_i(A_i))$ we have
$$|W|=|f_i(h_i(A_i))|=|f_i(V)|<|V|.$$ 
On the other hand, for every $x\in B$, $g_j(x)\equiv_\gamma g_j(\mathbf{a}(j))$ and therefore, for all $x\in B$, 
$$f_j(h_j(x))\equiv_\gamma f_j(x)=d(x,g_j(x),g_j(\mathbf{a}(j)))\equiv_\gamma d(x,g_j(\mathbf{a}(j)),g_j(\mathbf{a}(j)))=x.$$
Finally, $f(h(R))\subseteq \mathrm{min}(\m r)$ follows from $h(R)\subseteq \mathrm{min}(\m r)$, $g(R)\subseteq \mathrm{min}(\m r)$ and from Definition~\ref{regSMBdef} (1).

Thus from $|W|<|V|$ we get a contradiction with the minimality of $V$, and thus our assumption that $h_i(\beta)\nsubseteq\alpha$ is false. It follows that $(\gamma,\delta)$ can be separated from $(\alpha,\beta)$ with respect to $R$, as desired.

For the second statement, assume that $(S,R)\in \vr c$ are such that $i,j\in S$. Then the second statement follows from applying the first one to $\m r$.
\end{proof}

\begin{cor}\label{cantsepisequiv}
Let $\m R\leq_{sd}\m a_1\times\dots\times\m a_n$, where for each $i\leq n$, $\m a_i$ is a finite regular SMB algebra and $\theta_i$ is the Rees congruence. Then ``cannot be separated" is an equivalence relation on $\vr i_{\m R}$. Also, if $\vr t$ is a template of regular SMB algebras and $P=(V,D,\vr c)$ is a (2,3)-minimal multisorted instance of $CSP(\vr t)$, then ``cannot be separated" is an equivalence relation on $\vr i_P$.
\end{cor}

\begin{proof}
Reflexivity and transitivity of the ``cannot be separated" relation was claimed to be trivial before Lemma 17 of \cite{BuSMB}, and indeed, they follow from Definitions \ref{sep1} and \ref{sep2}, even when the congruence covers involved are not below the Rees congruence. The symmetry of the same relation follows from Lemma~\ref{sepissymmetric}.
\end{proof}

\subsection{Collapsing polynomials}

\begin{df}\label{collapsingpolydef}
Let $\m R\leq_{sd}\m a_1\times\dots\times\m a_n$, where for each $i\leq n$, $\m a_i$ is a finite regular SMB algebra and $\theta_i$ is the Rees congruence. Also, assume that $(i,\alpha,\beta)\in \vr i_{\m R}$. $f\in \poln{1}{R}$ is an $(\alpha,\beta)$-collapsing polynomial if
\begin{enumerate}
\item $f$ is idempotent,
\item For all $j\leq n$, $f_j(A_j)\subseteq\mathrm{min}(\m a_j)$,
\item For any $(j,\gamma,\delta)\in\vr i_{\m R}$, if $(\alpha,\beta)$ and $(\gamma,\delta)$ can be separated, then $f_j(\delta)\subseteq \gamma$ and 
\item For any $(j,\gamma,\delta)\in\vr i_{\m R}$, if $(\alpha,\beta)$ and $(\gamma,\delta)$ cannot be separated, then $f_j(A_j)\in M_{\m a}(\gamma,\delta)$.
\end{enumerate}
\end{df}

For the next statement, recall that $(\alpha,\beta)$-trace is the restriction of a $\beta$-class to an $(\alpha,\beta)$-minimal set which intersects more than one $\alpha$-class. 

Lemma~\ref{onesubtrace} that follows was just referred to Exercise 8.8 (1) of \cite{tct} by A. Bulatov in Proposition 6 of \cite{BuSMB}. That exercise speaks about Mal'cev algebras instead of SMB algebras, and is given without proof in \cite{tct}. However, this citation is not quite correct, as congruence covers below the Rees congruence in the congruence lattice of a regular SMB algebra $\m a$ might not remain covers when everything is restricted to the Mal'cev algebra $\mathrm{min}(\m a)$. Regardless, the same idea solves that exercise and proves our lemma. Also, we note that our statement is slightly stronger than Proposition 6 of \cite{BuSMB}, and we will use that extra detail to simplify a proof later on.

\begin{lm}{\rm (based on Exercise 8.8 (1) of \cite{tct})}\label{onesubtrace}
Let $\m a$ be a regular SMB algebra, $\theta$ its Rees congruence and $\alpha\prec\beta\leq \theta$ in $\cn a$. If $(a,b)\in \beta\setminus\alpha$, then there exist an $(\alpha,\beta)$-trace $N$ and $c\in A$ such that $\{a,c\}\subseteq N$ and $(c,b)\in\alpha$. 
\end{lm}

\begin{proof}
Let $(a,b)\in\beta\setminus\alpha$ and let $V\in M_{\m a}(\alpha,\beta)$. By Theorem~\ref{TCTfund} (3), there exists $f\in\poln{1}{a}$ such that $f(A)=V$ and $(f(a),f(b))\in (\beta|_V)\setminus(\alpha|_V)$. Since $\alpha\prec\beta$, we know that $\beta=\Cg(\alpha\cup\{(a,b)\})=\Cg(\alpha\cup\{(f(a),f(b))\})$. We define the relation
\[
\begin{gathered}
\beta':=\{(p(a_1,a_2,\dots,a_n),p(b_1,b_2,\dots,b_n)):n\in\mathbb{N},\,p\in\poln{n}{a} \text{ and}\\
\text{for all }i\leq n,\,(a_i,b_i)\in\alpha\text{ or }(a_i,b_i)=(f(a),f(b))\}.
\end{gathered}
\]
Obviously $\alpha\subsetneq\beta'\subseteq\beta$. We claim that $\beta'=\beta$. If we prove that $\beta'$ is a congruence then $\beta'=\beta$ would follow from $\alpha\prec\beta$ in $\cn a$, so we will prove $\beta'\in\cn a$. 

Reflexivity of $\beta'$ follows from $\alpha\subseteq\beta'$, while its compatibility with the fundamental operations follows from the definition of a polynomial. It remains to show symmetry and transitivity.

{\bf (S):} Let $(s,t)\in\beta'$, where $s=p(a_1,\dots,a_n)$ and $t=p(b_1,b_2,\dots,b_n)$ as in the definition of $\beta'$. The polynomial $q(x_1,\dots,x_n):=d(s,p(x_1,\dots,x_n),t)$ witnesses that $(t,s)\in\beta'$, since
$$(t,s)=(d(s,s,t),d(s,t,t))=(q(a_1,a_2,\dots,a_n),q(b_1,b_2,\dots,b_n)).$$

{\bf (T):} Let $(s,t),(t,u)\in\beta'$ and let $p,q\in\pol{a}$, like in the definition of $\beta'$, be such that
\[
\begin{gathered}
s=p(a_1,\dots,a_m)\\
t=p(b_1,\dots,b_m)=q(c_1,\dots,c_n)\\
u=q(d_1,\dots,d_n).
\end{gathered}
\]
The polynomial $r(x_1,\dots,x_m,y_1,\dots,y_m)$ given by 
\[d(p(x_1,\dots,x_m),t,q(y_1,\dots,y_m))\]
satisfies
\[
\begin{gathered}
s= d(s,t,t) = r(a_1,\dots,a_m,c_1,\dots,c_n)\\
u= d(t,t,u) = r(b_1,\dots,b_m,d_1,\dots,d_n).
\end{gathered}
\]
Hence $(s,u)\in\beta'$ and $\beta'$ is transitive.

Now we have that $\beta'=\beta$, and since $(a,b)\in\beta$, it follows that there exists a polynomial $g\in\poln{n+1}{a}$ and $(c_1,d_1),\dots,(c_n,d_n)\in\alpha$ such that 
\[(a,b)=(g(f(a),c_1,\dots,c_n),g(f(b),d_1,\dots,d_n)).\]
Define $p\in\poln{1}{a}$ by $p(x)=g(x,c_1,\dots,c_n)$. Since 
$$p(f(a))=a\not\equiv_\alpha b=g(f(b),d_1,\dots,d_n)\equiv_\alpha g(f(b),c_1,\dots,c_n)= p(f(b)),$$
according to Theorem~\ref{TCTfund} (2), $U:=f(V)\in M_{\m a}(\alpha,\beta)$. Therefore $U$ contains a trace $N$ such that $\{p(f(a)),p(f(b))\}\subseteq N$. Since $a=p(f(a))$ and $p(f(b))\equiv_\alpha b$, we can select $c:=p(f(b))$ to prove the lemma.
\end{proof}

\begin{lm}[Lemma 18 of \cite{BuSMB}]\label{collapsingexist}
Let $\m R\leq_{sd}\m a_1\times\dots\times\m a_n$, where for each $i\leq n$, $\m a_i$ is a finite regular SMB algebra and $\theta_i$ is the Rees congruence. Also, assume that $(i,\alpha,\beta)\in \vr i_{\m R}$ and $\mathbf{a}\in\mathrm{min}(\m r)$ is such that that the $\beta$-class of $\mathbf{a}(i)=a$ contains more than one $\alpha$-class and let $b\in A_i$ be such that $(a,b)\in(\beta-\alpha)$. Then there exists an $(\alpha,\beta)$-collapsing polynomial $f\in\poln{1}{r}$ such that $f(\mathbf{a})=\mathbf{a}$ and $(f_i(b),b)\in \alpha$.
\end{lm}

\begin{proof}
According to Lemma~\ref{onesubtrace}, there exists some $U\in M_{\m a_i}(\alpha,\beta)$ and $b'\in A_i$ such that $a,b'\in U$ and $(b',b)\in\alpha$. Let $(j,\gamma,\delta)\in\vr i_{\m r}$ be such that $(\alpha,\beta)$ can be separated from $(\gamma,\delta)$ with respect to $\m r$. According to Lemma~\ref{septominimal}, we can select idempotent polynomials $g^{j\gamma \delta}\in\poln{1}{\m r}$ such that $g^{j\gamma \delta}_i(A_i)=U$,  $g^{j\gamma \delta}_j(\delta)\subseteq\gamma$ and $g^{j\gamma \delta}(R)\subseteq\mathrm{min}(\m r)$. By idempotence, for all $x\in U$, $g^{j\gamma \delta}_i(x)=x$. Let $g'$ be the composition of all selected polynomials $g^{j\gamma \delta}$ and $g''(x)=d(g'(x),g'(\mathbf{a}),\mathbf{a})$. We still have that $g''_i(A_i)=U$ and for all $x\in U$, $g''_i(x)=x$. Also, for any $(j,\gamma,\delta)\in\vr i_{\m r}$ such that $(\alpha,\beta)$ can be separated from $(\gamma,\delta)$ with respect to $\m r$, $g''_j(\delta)\subseteq \gamma$. Finally, $g''(\mathbf{a})=\mathbf{a}$. Let $g$ be an idempotent power of $g''$. $g$ satisfies the following properties:
\begin{enumerate}
\item $g\in\poln{1}{r}$ is idempotent,
\item For all $j\leq n$, $g_j(A_j)\subseteq \mathrm{min}(\m a_j)$,
\item For any $(j,\gamma,\delta)\in\vr i_{\m r}$ such that $(\alpha,\beta)$ can be separated from $(\gamma,\delta)$ with respect to $\m r$, $g_j(\delta)\subseteq \gamma$,
\item $g_i(A_i)=U$ and
\item $g(\mathbf{a})=\mathbf{a}$.
\end{enumerate}
Let $f\in\poln{1}{r}$ be a polynomial such that $f$ also satisfies the above properties $(1)$-$(5)$ and moreover such that $\sum\limits_{j=1}^n|f_j(A_j)|$ is minimal among such polynomials. We claim that $f$ has all desired properties. Properties $(1)$-$(3)$ are three of the four defining properties of an $(\alpha,\beta)$-collapsing polynomial. From $(1)$ follows that $f_i$ is an idempotent polynomial of $\m a_i$, and since $b'\in U=f_i(A_i)$, therefore $f_i(b')=b'$. From $(b,b')\in\alpha$ follows that $(f_i(b),f_i(b'))=(f_i(b),b')\in\alpha$, so $(f_i(b),b)\in\alpha$. Finally, $f(\mathbf{a})=\mathbf{a}$ by $(5)$.

We need to prove that for every $(j,\gamma,\delta)\in\vr i_{\m r}$ such that $(\alpha,\beta)$ can not be separated from $(\gamma,\delta)$ with respect to $\m r$, $f_j(A_j)\in M_{\m a_j}(\gamma,\delta)$. Suppose that this property fails for some $(j,\gamma,\delta)$, so there exists a $V\in M_{\m a_j}(\gamma,\delta)$ such that $V\subsetneq f_j(A_j)$. Let $p'\in\poln{1}{a_j}$ be an idempotent polynomial such that $p'(A_j)=V$ and let $h'\in\poln{1}{r}$ be such that $h'_j=p'$. We know that $(h'\circ f)_j(A_j)=V$ and since $f_j$ is the identity on $f_j(A_j)$, while $h'_j$ is the identity map on $V$, we obtain that $(h'\circ f)_j(\delta)\nsubseteq\gamma$. By Lemma~\ref{sepissymmetric}, $(\gamma,\delta)$ can not be separated from $(\alpha,\beta)$ with respect to $\m r$, so $(h'\circ f)_i(\beta)\nsubseteq \alpha$. As $f_i(A_i)=U$, this means that $h'_i(\beta|_U)\nsubseteq \alpha$. By Theorem~\ref{TCTfund} (2) we obtain $h'_i(U)=U'\in M_{\m a_i}(\alpha,\beta)$. Let $p''\in\poln{1}{\m a_i}$ satisfy $p''(U')=U$. Such $p''$ exists by Theorem~\ref{TCTfund} (1). Let $h''\in\poln{1}{r}$ satisfy $h''_i=p''$. 

Let $f':=(h''\circ h'\circ f)^{(|R|!)}$ be the idempotent power of $h''\circ h'\circ f$. We have that $f'_i(A_i)=U$ and for all $x\in U$, $f_i'(x)=x$. Moreover, for all $j\leq n$, $f'_j(A_j)\subseteq\mathrm{min}(\m a_j)$. Also, for all $(j,\gamma,\delta)\in \vr i_{\m r}$ such that $(\alpha,\beta)$ can be separated from $(\gamma,\delta)$ with respect to $\m r$, $f'_j(\delta)\subseteq \gamma$. Thus $f'$ satisfies $(1)$-$(4)$, i.e. the same properties that $g'$ satisfied. Then like in the first paragraph, define $f'':=d(f'(x),f'(\mathbf{a}),\mathbf{a})$ and let $f'''=f''^{(|R|!)}$ be its idempotent power. Analogously as for $f$, we conclude that $f'''$ satisfies $(1)$-$(5)$. 

Finally, for all $k\leq n$, $f'''_k(A_k)=p_k(f_k(A_k))$ for some $p\in \poln{1}{r}$, so 
\begin{equation}\tag{*}
f'''_k(A_k)|\leq |f_k(A_k)|.
\end{equation}
On the other hand, $f'''_j(A_j)$ can also be written as $q\circ h'\circ f$ for some $q\in\poln{1}{r}$ so $f'''_j(A_j)=q_j( (h'\circ f)_j(A_j))=q_j(V)$, and therefore
\begin{equation}\tag{**}
|f_j'''(A_j)|=|q_j(V)|\leq |V|<|f_j(A_j)|.
\end{equation}
By $(*)$ and $(**)$, $\sum\limits_{k=1}^n|f'''_k(A_k)|<\sum\limits_{k=1}^n|f_k(A_k)|$, contradicting the choice of $f$. Therefore, the assumption that $f_j(\gamma,\delta)\notin M_{\m a_j}(\gamma,\delta)$ must have been false, and so we proved that $f$ is an $(\alpha,\beta)$-collapsing polynomial, with $f(\mathbf{a})=\mathbf{a}$ and $(f_i(b),b)\in \alpha$.
\end{proof}

The next corollary plays the same role for us as Lemma 21 of \cite{BuSMB}.

\begin{cor}\label{colltominsingleton}
Let $\m R\leq_{sd}\m a_1\times\dots\times\m a_n$, where for each $j\leq n$, $\m a_j$ is a finite regular SMB algebra and $\theta_j$ is its Rees congruence. Moreover, let $(i,\alpha,\beta)\in \vr i_{\m r}$ and let $f\in \poln{1}{R}$ be an $(\alpha,\beta)$-collapsing polynomial. Then
\begin{enumerate}
\item If $j\leq n$ is such that there exists $(j,\gamma,\delta)\in\vr i_{\m r}$ so that $(\alpha,\beta)$ can not be separated from $(\gamma,\delta)$ with respect to $\m r$, then $f_j(A_j)\in M_{\m a}(\gamma',\delta')$ for every $(j,\gamma',\delta')\in\vr i_{\m r}$ such that $(\alpha,\beta)$ can not be separated from $(\gamma',\delta')$ with respect to $\m r$.

\item If $j\leq n$ is such that for every $(j,\gamma,\delta)\in\vr i_{\m r}$, $(\alpha,\beta)$ can be separated from $(\gamma,\delta)$ with respect to $\m r$, then $|f_j(A_j)|=1$, and the only element of $f_j(A_j)$ is in $\mathrm{min}(\m a_j)$.
\end{enumerate}

\begin{proof}
(1) is true by Definition~\ref{collapsingpolydef}. The reader may be confused how can $f_j(A_j)$ simultaneously be a minimal set for possibly many covering pairs of congruences, but recall that, according to Corollary~\ref{cantsepisequiv}, $(j,\gamma,\delta)$ and $(j,\gamma',\delta')$ can not be separated with respect to $\m r$, and thus, according to Proposition~\ref{minsetsarethesame}, $U\subseteq A_j$ is in $M_{\m a}(\gamma,\delta)$ iff it is in $M_{\m a}(\gamma',\delta')$.

(2) Let $0_j=\gamma_0\prec\gamma_1\prec\dots\prec\gamma_k=\theta_j$ and let $a,b\in A_j$ be arbitrary. Denoting by $f_j^k:=f_j\circ f_j\circ\dots\circ f_j$, where there are $k-1$ compositions, we proceed to inductively prove that for all $0\leq \ell\leq k$, $(f_j^{\ell+1}(a),f_j^{\ell+1}(b))\in\gamma_{k-\ell}$. 

The base case follows from $f_j(a),f_j(b)\in\mathrm{min}(\m a_j)$, i.e. $(f_j(a),f_j(b))\in\theta_j=\gamma_k$. Assume that $(f_j^{t}(a),f_j^{t}(b))\in\gamma_{k-t+1}$ and denote $a':=f_j^{t}(a)$ and $b':=f_j^{t}(b)$. The assumption of (2) gives that $(\alpha,\beta)$ can be separated from $(\gamma_{k-t},\gamma_{k-t+1})$ with respect to $\m r$. Therefore, by Definition~\ref{collapsingpolydef}, $f_j(\gamma_{k-t+1}\subseteq\gamma_{k-t}$. Hence 
\[(f_j^{t+1}(a),f_j^{t+1}(b))=(f_j(a'),f_j(b'))\in f_j(\gamma_{k-t+1}\subseteq\gamma_{k-t}.\]
The inductive step is proved, and hence we know that $$(f_j^{k+1}(a),f_j^{k+1}(b))\in \gamma_{k-k}=\gamma_0=0_j.$$
In other words, $f_j^{k+1}(a)=f_j^{k+1}(b)$. But, from idempotence of $f$, it follows that $f_j$ is also an idempotent polynomial of $\m a_j$, i.e. $f_j(f_j(x))=f_j(x)$. Applying this several times, we get
\[f_j(a)=f_j^{k+1}(a)=f_j^{k+1}(b)=f_j(b).\]
As $a$ and $b$ were arbitrarily chosen elements of $A_j$, we get $|f_j(A_j)|=1$, as desired. The only element of $f_j(A_j)$ lies in $\mathrm{min}(\m a_j)$ by Definition \ref{collapsingpolydef} (2).
\end{proof}

\end{cor}

\subsection{Split elements, alignment and link partitions}

\begin{df}\label{splitdf}
Let $\m a$ be a finite regular SMB algebra, $\theta$ its Rees congruence and $0_{\m a}\leq \alpha\prec\beta\leq\theta$ in $\cn a$. We say that $a\in A$ is an $\alpha\beta$-split element if there exist elements $b,c\in A$ such that $(b,c)\in\beta$ and $(a\mt b,a\mt c)\notin\alpha$.
\end{df}

\begin{exmpl}\label{topsplitbotnot}
Let $\m a$ be a finite regular SMB algebra, $\theta$ its Rees congruence and $0_{\m a}\leq \alpha\prec\beta\leq\theta$ in $\cn a$. Then $a\in\mathrm{min}(\m a)$ is not an $\alpha\beta$-split element, regardless of the choice of $\alpha$ and $\beta$, since $a\mt b=a=a\mt c$. On the other hand, if $\m a$ is unital with the neutral element $1$, then $1$ is always a split element, again regardless of the choice of $\alpha$ and $\beta$. To see this, just take any $(b,c)\in(\beta-\alpha)$. Then $(1\mt b,1\mt c)=(b,c)\notin\alpha$.  
\end{exmpl}

\begin{df}\label{aligndef}
Let $\m R\leq_{sd}\m a_1\times\dots\times\m a_n$, where for each $i\leq n$, $\m a_i$ is a finite regular SMB algebra and $\theta_i$ is the Rees congruence. Let $W\subseteq\{1,2,\dots,n\}$ and for each $i\in W$ let $(i,\alpha_i,\beta_i)\in\vr i_{\m R}$. By $\overline{\alpha}$ and $\overline{\beta}$ we mean $\lb\alpha_i:i\in W\rb$ and $\lb\beta_i:i\in W\rb$, respectively. We say that $R$ is $\overline{\alpha\beta}$-aligned if, whenever $\mathbf{a}\in R$ and $i,j\in W$, $\mathbf{a}(i)$ is an $(\alpha_i,\beta_i)$-split element iff $\mathbf{a}(j)$ is an $(\alpha_j,\beta_j)$-split element.
\end{df}

The following lemma is based on Lemma 19 of \cite{BuSMB}.

\begin{lm}\label{alignorseparate}
Let $\m R\leq_{sd}\m a_1\times\dots\times\m a_n$, where for each $i\leq n$, $\m a_i$ is a finite regular SMB algebra and $\theta_i$ is the Rees congruence. Let $W\subseteq\{1,2,\dots,n\}$, for each $i\in W$ let $(i,\alpha_i,\beta_i)\in\vr i_{\m R}$ and for each $i,j\in W$ assume that $(\alpha_i,\beta_i)$ and $(\alpha_j,\beta_j)$ can not be separated with respect to $\m r$. Then $R$ is $\overline{\alpha\beta}$-aligned.
\end{lm}

\begin{proof}
Assume that $R$ is not $\overline{\alpha\beta}$-aligned. Then there exists $\mathbf{a}\in R$ and $i,j\in W$ is such that $\mathbf{a}(i)$ is an $(\alpha_i,\beta_i)$-split element and $\mathbf{a}(j)$ is not an $(\alpha_j,\beta_j)$-split element. Hence, there exists a pair $(b_i,c_i)\in\beta_i$ such that $(\mathbf{a}(i)\mt b_i,\mathbf{a}(i)\mt c_i)\notin\alpha_i$, while $\{(\mathbf{a}(j)\mt b,\mathbf{a}(j)\mt c):(b,c)\in\beta_j\}\subseteq \alpha_j$. Hence, for the polynomial $f\in\poln{1}{r}$ given by $f(x)=\mathbf{a}\mt x$ we know that $f_i(\beta_i)\nsubseteq \alpha_i$, while $f_j(\beta_j)\subseteq \alpha_j$. Therefore, $(\alpha_i,\beta_i)$ can be separated from $(\alpha_j,\beta_j)$, a contradiction.
\end{proof}

\subsection{Link partitions, coherent sets and block-2-consistency}

The next lemma plays the role of Lemma 20 (2) from \cite{BuSMB}.

\begin{lm}\label{unitallink}
Let $\vr t$ be a template of regular SMB algebras, $P=(V,D,\vr c)$ a (2,3)-minimal multisorted instance of $CSP(\vr t)$ and for each $i\in V$ let $\m a_i$ be unital. Assume that $(i,\alpha_i,\beta_i)\in \vr i_P$ for each $i\in V$ and for each $(S,R)\in \vr c$ let $R$ be $\overline{\alpha\beta}$-aligned, where $\overline{\alpha}=\lb\alpha_i:i\in S\rb$ and $\overline{\beta}=\lb\beta_i:i\in S\rb$. Then $P$ has a link partition.
\end{lm}

\begin{proof}
For each $i\in V$ let $\varepsilon_i$ consist of two classes: 
\[
\begin{gathered}
A_{i,1}=\{a\in A_i:a\text{ is an }\alpha_i\beta_i-\text{split element}\}\text{ and}\\
A_{i,2}=\{a\in A_i:a\text{ is not an }\alpha_i\beta_i-\text{split element}\}.
\end{gathered}
\] 
Let $(S,R)\in \vr c$. We know that $R$ is $\overline{\alpha\beta}$-aligned. By Example~\ref{topsplitbotnot} and since all $\m a_i$ are unital, each of $A_{i,1}$ and $A_{i,2}$ is nonempty. Then by Definition~\ref{aligndef}, $P$ has a link partition.
\end{proof}

\begin{df}\label{sizedef}
Let $\vr t$ be a template of regular SMB algebras, $P=(V,D,\vr c)$ a (2,3)-minimal multisorted instance of $CSP(\vr t)$ and for all $i\in V$, $\theta_i$ is the Rees congruence. The {\em size} of $P$ is\[Size(P):=\max\{|A_i|:i\in V\text{ and }\theta_i\neq 1_{\m a_i}\}.\]
In other words, the size of $P$ is the maximal size of a non-Mal'cev domain of a variable of $P$.
\end{df}

\begin{prp}\label{linkmax}
Let $\vr t$ be a template of regular SMB algebras, $P=(D,V,\vr c)$ a (2,3)-minimal multisorted instance of $CSP(\vr t)$, for all $i\in V$ let $\theta_i$ is the Rees congruence, and let $Size(P)=M$. Let $W\subseteq V$ be the set $\{i\in V:|A_i|=M$ and $\theta_i\neq 1_{\m A_i}\}$. Also, assume that the restricted instance $P|_W$ has a link partition. Then, for any connected component $Q$ of $P|_{W}$, the instance 
$$P_Q=(V,\{B_i:i\in V\},\{(S,R\cap\prod\limits_{i\in S}B_i):(S,R)\in \vr c\}),$$ 
where for all $i\in W$, $B_i:=A_i\cap Q$ and for all $i\in V\setminus W$, $B_i:=A_i$, is an instance of $CSP(\vr t)$ such that for each $Q$, $Size(P_Q)<M$.
\end{prp}

\begin{proof}
Since $P|_W$ has a link partition, the microstructure graph $\Gamma_{P|_W}$ is disconnected, while from (2,3)-minimality of $P$ follows the (2,3)-minimality of $P|_W$, which in turn implies that the scope graph $\Gamma_W$ is connected, and that $\Gamma_{P|_W}$ is (1,1)-minimal and cycle consistent. Thus, by Proposition~\ref{linksubuniv}, for any connected component $Q$ of $\Gamma_{P|_W}$ and any $i\in W$, $A_i\cap Q$ is the domain of a subuniverse of $\m a_i$. Hence, the tightened instance $P_Q$ is an instance of $CSP(\vr t)$. Moreover, for any $i\in W$, $B_i\subsetneq A_i$, so $|B_i|<|A_i|=M$, while for each $i\in V\setminus W$, either $\m B_i=\m A_i$ is a Mal'cev algebra, or $|B_i|=|A_i|<M$. Thus $Size(P_Q)<M$.
\end{proof}

\begin{df}\label{coherentdef}
Let $\vr t$ be a template of regular SMB algebras, $P=(V,D,\vr c)$ a (2,3)-minimal multisorted instance of $CSP(\vr t)$ and $(i,\alpha,\beta)\in\vr i_P$. By $W_{i,\alpha,\beta}$ we denote the set
\[\{j\in V:(\exists(j,\gamma,\delta)\in \vr i_P)(\alpha,\beta)\text{ can not be separated from }(\gamma,
\delta)\text{ wrt. }P\}.\]
$W\subseteq V$ is a {\em coherent set} iff there exists $(i,\alpha,\beta)\in\vr i_P$ so that $W=W_{i,\alpha,\beta}$.
\end{df}

\begin{df}\label{blockmindef}
Let $\vr t$ be a template of regular SMB algebras and $P=(V,D,\vr c)$ a (2,3)-minimal multisorted instance of $CSP(\vr t)$. We say that $P$ is {\em block-2-consistent} if for every coherent set $W$, every $i,j\in W$, if $(\{i,j\},R_{i,j})$ is a constraint relation of $P$ and $(a,b)\in R_{i,j}$, then the restricted instance $P|_W$ has a solution $f$ such that $f(i)=a$ and $f(j)=b$. We say that $P$ is {\em block-minimal} if, for any coherent set $W$, where $P|_W=(W,D',\vr c')$ is the restricted instance, any $(S,R)\in \vr c'$ and any $\mathbf{a}\in R$, there exists a solution $f$ of $P|_W$ such that for all $i\in S$, $f(i)=\mathbf{a}(i)$.
\end{df}

Clearly, in a $(2,3)$-minimal instance, for any coherent set $W$ and any pair of variables $i,j\in W$, there exists a constraint $(\{i,j\},R_{i,j})$ which is a constraint both of $P$ and of $P|_W$. However, block-minimality can be quite a lot stronger, when there are large coherent sets in $P$. We will now prove that any (2,3)-minimal and weakly M-irreducible multisorted instance over regular SMB algebras can be reduced to an equivalent one which is block-minimal. However, for the purpose of tractability, we will prove that even block-2-consistency is sufficient.

The following is our variant of Proposition 22 of \cite{BuSMB}. We have to amend its proof since $P|_W$ might not have a link partition.

\begin{thm}
Let $\vr t$ be a template of regular SMB algebras, let $P=(V,D,\vr c)$ be a (2,3)-minimal and weakly M-irreducible multisorted instance of $CSP(\vr t)$ and let $W$ be a coherent set for $P$. Then there exists an algorithm which inputs $P$ and $W$ and outputs new constraint relations $\vr c'=\{(S,R_S'):(S,R_S)\in \vr c\}$ such that, for all $(S,R_S)\in\vr c$, $\m r_S'\leq \m r_S$ and $R_S'$ consists precisely of all tuples $\mathbf{a}\in R_S$ such that there exists a solution $f$ of the restricted instance $P|_W$ so that $f(i)=\mathbf{a}(i)$ for all $i\in S\cap W$.
\end{thm}

\begin{proof}
The algorithm in question we will call $CHKCOHSET((V,D,\vr c),W)$. The procedure $SOLVE(V',D',\vr c')$ is the main algorithm which solves CSP over regular SMB algebras. We will define this main algorithm in the next subsection, but here we will use it only on instances $(V',D',\vr c')$ of $CSP(\vr t)$ such that $Size(V',D',\vr c')<Size(V,D,\vr c)$. 
\begin{enumerate}
\item[{\bf Step 1.}] Set up the working domains $T_S:=R_S$, for all $(S,R_S)\in \vr c$.
\item[{\bf Step 2.}] If there exists no $i\in W$ such that $\theta_i<1_{\m a_i}$ and $|A_i|=Size(P)$ then
\begin{enumerate}
\item[{\bf Step 2.1.}] For all $(S,R_S)\in \vr c$ and any $\mathbf{a}\in R_S$,\\ if $SOLVE(W,D|_W,\vr c|_W\cup\{(\{i\},\{\mathbf{a}(i)\}):i\in S\cap W\})\})=NO$,\\ set $T_S:=T_S\setminus\{\mathbf{a}\}$.
\item[{\bf Step 2.2.}] Output $\{(S,T_S):(S,R_S)\in\vr c\}$ and stop.
\end{enumerate}
\item[{\bf Step 3.}] Set $W':=\{i\in W:\theta_i<1_{\m a_i}$ and $|A_i|=Size(P)\}$.
\item[{\bf Step 4.}] For any $(S,R_S)\in\vr c$ such that $S\cap W'\neq\emptyset$ and any $\mathbf{a}\in R_S$ do
\begin{enumerate}
\item[{\bf Step 4.1.}] Let $i\in S\cap W'$ and $a:=\mathbf{a}(i)$.
\item[{\bf Step 4.2.}] Compute the connected component of $a$ in $\Gamma_{P|_{W'}}$ and call it $Q$.
\item[{\bf Step 4.3.}] For each $j\in W'$ set $B_j:=A_j\cap Q$,
\item[{\bf Step 4.4}] For each $j\in W\setminus W'$ set $B_i:=A_i$,
\item[{\bf Step 4.5.}] For each $(T,R_T)\in\vr c|_W$\\
\mbox{ }\hspace{-0.7cm}set $R'_T$ as $R_T\cap(\prod_{j\in T}B_j)$ and\\
\mbox{ }\hspace{-0.7cm}set $\vr c':=\{(T,R'_T):(T,R_T)\in\vr c|_W\}\cup\{(\{j\},\{\mathbf{a}(j)\}):j\in S\cap W\}$.
\item[{\bf Step 4.6.}] If $SOLVE(W,\{B_i:i\in W\},\vr c')=NO$,\\ set $T_S:=T_S\setminus\{\mathbf{a}\}$.
\end{enumerate}
\item[{\bf Step 5.}] Fix $t\in W'$.
\item[{\bf Step 6.}] For any $(S,R_S)\in\vr c$ such that $S\cap W'=\emptyset$ and any $\mathbf{a}\in R_S$ do
\begin{enumerate}
\item[{\bf Step 6.1.}] Set $q:=0$
\item[{\bf Step 6.2.}] For any $c\in A_t$ do
\begin{enumerate}
\item[{\bf Step 6.2.1.}] Compute the connected component of $c$ in $\Gamma_{P|_{W'}}$ and call it $Q$.
\item[{\bf Step 6.2.2.}] For each $i\in W'$ set $B_i:=A_i\cap Q$ 
\item[{\bf Step 6.2.3.}] For each $i\in W\setminus W'$ set $B_i:=A_i$.
\item[{\bf Step 6.2.4.}] For each $(T,R_T)\in\vr c|_W$ \\
\mbox{ }\hspace{-1.2cm}set $R_T'$ as $R_T\cap(\prod_{i\in S}B_i)$,\\ 
\mbox{ }\hspace{-1.2cm}set $\vr c':=\{(T,R_T'):(T,R_T)\in\vr c|_W\}\cup\{(\{j\},\{\mathbf{a}(j)\}):j\in S\cap W\}$.
\item[{\bf Step 6.2.5.}] If $SOLVE(W,\{B_i:i\in W\},\vr c')=YES$, set $q:=1$.
\end{enumerate}
\item[{\bf Step 6.3.}] If $q=0$ set $T_S:=T_S\setminus\{\mathbf{a}\}$
\end{enumerate}
\item[{\bf Step 7.}] Output $\{(S,T_S):(S,R_S)\in \vr c\}$ and stop.
\end{enumerate}

We note in passing that the calls to $SOLVE$ in Steps 2.1, 4.6 and 6.2.5 are formally not correct, as the instances which are being solved are most likely not multisorted instances in the sense of Definition~\ref{multisortedCSP}, namely they are not subdirect. What we mean, of course, is that we tighten the instances so that for each $i\in S\cap W$ the domain is tightened to the singleton $\{\mathbf{a}(i)\}$, then run (1,1)-minimality and only then the algorithm SOLVE.

The algorithm first checks whether $Size(P|_W)<Size(P)$. If this is the case it simply checks each pair in $R_{\{i,j\}}$, $i,j\in W$, making a linear (at most $\sum\limits_{i,j\in W}|R_{\{i,j\}}|$ many) number of calls to $SOLVE$ with smaller instance size and stops at the end of Step 2. 

If the procedure reached Step 3, this means that $Size(P|_W)=Size(P)$ and thus there exist some $i\in W$ such that $\m A_i$ is not Mal'cev and which satisfies $|A_i|=Size(P)$. The set of all $i$ for which this is true is denoted as $W'$. Since $P$ is weakly M-irreducible, all $\m a_i$ such that $i\in W'$ are unital. On the other hand, $W'\subseteq W$ and $W$ is a coherent set. Hence, we can select $\overline{\alpha}=\lb\alpha_i:i\in W'\rb$ and $\overline{\beta}=\lb\beta_i:i\in W'\rb$ so that for each $i\in W'$, $(i,\alpha_i,\beta_i)\in\vr i_P$ and for each $i,j\in W'$, $(\alpha_i,\beta_i)$ can not be separated from $(\alpha_j,\beta_j)$ with respect to $P|_{W'}$.

Let $(S,R)\in\vr c|_{W'}$ be any constraint. From Definition~\ref{sep2} and remarks immediately following it, we have that for each $i,j\in S$, $(\alpha_i,\beta_i)$ can not be separated from $(\alpha_j,\beta_j)$ with respect to $R$. Thus, according to Lemma~\ref{alignorseparate}, the restriction of $R$ to $W'$ is $\overline{\alpha\beta}$-aligned. Since all $\m a_i$, $i\in W'$ are unital, according to Lemma~\ref{unitallink}, $P|_{W'}$ has a link partition. Then Proposition~\ref{linkmax} proves that replacing each domain of a variable in $W'$ with its intersection with the connected component $Q$, and tightening the instance $P$ accordingly, constitutes a valid instance of $CSP(\vr t)$, whose size is strictly smaller than $Size(P)$.

Thus Step 4 checks for each $(S,R_S)$ such that $S\cap W'\neq \emptyset$, and each $\mathbf{a}\in R_S$, whether it is possible that a solution $f$ to $P|_W$ satisfies $f(i)=\mathbf{a}(i)$ for all $i\in S\cap W$, while Steps 5-6 perform the same for each $(S,R_S)$ such that $S\cap W'=\emptyset$, and each $\mathbf{a}\in R_S$. Of course, in Steps 5-6 we are forced to check each connected component of the hypergraph $\Gamma_{P|_{W'}}$, since by choosing $\mathbf{a}$ we have not determined which component should the solution of $P|_W$ go through, but there are only at most a constant number of such components (bounded from above by the domain size), so this is not an issue. The running time of computing connected component of an element in the hypergraph $\Gamma_{P|_{W'}}$ is linear, and thus the running time of Step 4, and also of Steps 5-6, is dominated by a linear number of calls to $SOLVE$, always applied to instances of $CSP(\vr t)$ of size smaller than $Size(P)$.
\end{proof}


\subsection{Ensembles and the proof of tractability}

\begin{df}\label{ensembledef}
Let $\vr t$ be a template of regular SMB algebras, $P=(V,D,\vr c)$ a (2,3)-minimal multisorted instance of $CSP(\vr t)$ and for all $i\in V$, let $\theta_i$ be the Rees congruence. Let $\overline{\beta}=\lb \beta_i:i\in V\rb$ be a family of congruences such that for all $i\in V$, $\beta_i\in\cn A_i$ and $\beta_i\leq \theta_i$. We say that $\lb f_{j,\gamma,\delta}:(j,\gamma,\delta)\in \vr i_P\rb$ is a {\em $\overline{\beta}$-ensemble} for $P$ if 
\begin{enumerate}
\item For every $(j,\gamma,\delta)\in \vr i_P$, where $W=W_{i,\gamma,\delta}$ is the coherent set, $f_{i,\gamma,\delta}$ is a solution to the subinstance $P|_W$ such that for all $i\in W$, $f_{j,\gamma,\delta}(i)\in\mathrm{min}(\m a_i)$.
\item For all $(j,\gamma,\delta),(j',\gamma',\delta')\in \vr i_P$ and all $i\in W_{j,\gamma,\delta}\cap W_{j',\gamma',\delta'}$, it holds that $(f_{j,\gamma,\delta}(i),f_{j',\gamma',\delta'}(i))\in\beta_i$.
\item For any $(S,R)\in \vr c$, there exists $\mathbf{a}\in R$ such that for any $i\in S$ and any $(j,\gamma,\delta)\in \vr i_P$ such that $i\in W_{j,\gamma,\delta}$, $(f_{j,\gamma,\delta}(i),\mathbf{a}(i))\in\beta_i$.
\end{enumerate}
\end{df}

\begin{lm}[Lemma 24 of \cite{BuSMB}]\label{tightenensemble}
Let $\vr t$ be a template of regular SMB algebras, $P=(V,D,\vr c)$ a (2,3)-minimal and block-minimal multisorted instance of $CSP(\vr t)$ and for all $i\in V$, let $\theta_i$ be the Rees congruence. Let $\overline{\beta}=\lb \beta_i:i\in V\rb$ be a family of congruences such that for all $i\in V$, $\beta_i\in\cn A_i$ and $\beta_i\leq \theta_i$. Also, let $\overline{\alpha}=\lb \alpha_i:i\in V\rb$ be defined so that for one $i\in V$, let $\alpha_i\prec \beta_i$ in $\cn a_i$, while for all $j\in V\setminus\{i\}$, let $\alpha_j=\beta_j$. If there exists a $\overline{\beta}$-ensemble for $P$, then there exists an $\overline{\alpha}$-ensemble for $P$.
\end{lm}

\begin{proof}
Let $\vr m=\lb \varphi_{j,\gamma,\delta}:(j,\gamma,\delta)\in \vr i_P\rb$ be a $\overline{\beta}$-ensemble for $P$. Let $\xi:V\rightarrow \prod \m a_j/\beta_j$ be given by: for any $k\in V$, and any $W=W_{j,\gamma,\delta}$ such that $k\in W$, $\xi(k)=[\varphi_{j,\gamma,\delta}(k)]_{\beta_k}$. By the definition of a $\overline{\beta}$-ensemble for $P$, $\xi$ does not depend on the choices of $(j,\gamma,\delta)\in\vr i_P$. If $\xi(i)$ consists of only one $\alpha_i$-class, then $\vr m$ is an $\overline{\alpha}$-ensemble for $P$, and we are done.

So suppose that $\xi(i)$ contains more than one $\alpha_i$-class and select the $\alpha_i$-class $[\varphi_{i,\alpha_i,\beta_i}]_{\alpha_i}$, call it $B$. We want to construct a $\overline{\beta}$-ensemble $\vr m'=\lb \varphi'_{j,\gamma,\delta}:(j,\gamma,\delta)\in \vr i_P\rb$ such that for all $(j,\gamma,\delta)\in \vr i_P$ for which $i\in W_{j,\gamma,\delta}$ holds, $\varphi'_{j,\gamma,\delta}(i)\in B$. This would say that $\vr m'$ is an $\overline{\alpha}$-ensemble and finish our proof. 

We define $\varphi'_{i,\alpha_i,\beta_i}:=\varphi_{i,\alpha_i,\beta_i}$ and if $(j,\gamma,\delta)\in \vr i_P$ is such that $i\notin W_{j,\gamma,\delta}$, we define $\varphi'_{j,\gamma,\delta}:=\varphi_{j,\gamma,\delta}$. To define $\varphi'_{j,\gamma,\delta}$ such that $i\in W_{j,\gamma,\delta}$, we need to work a little. 

Let us fix some notation. Let $W:=W_{i,\alpha_i,\beta_i}$ and $U:=W_{j,\gamma,\delta}$, where $i\in U\cap W$. Denote by $S_U$, $S_W$ and $S_{U\cap W}$ the solution sets to the restricted instances $P|_U$, $P|_W$ and $P|_{U\cap W}$, respectively. According to block-2-consistency, for each of $X=U,W,U\cap W$, $S_X$ is a subdirect subuniverse of $\prod\limits_{k\in X}\m a_k$ such that for all $k,\ell\in X$, $pr_{k,\ell}S_X=R_{\{k,\ell\}}$. Moreover, by the definition of the restriction of an instance, the projections $\pi_{U\cap W}(S_U)$ and $\pi_{U\cap W}(S_W)$ are both contained in $S_{U\cap W}$. Denote also $\varphi:=\varphi_{i,\alpha_i,\beta_i}$ and $\psi:=\varphi_{j,\gamma,\delta}$. Our immediate goal is to define $\varphi'_{j,\gamma,\delta}\in S_U$.

According to Lemma~\ref{collapsingexist} we can select $f\in\poln{1}{S_U}$ which is an $(\alpha_i,\beta_i)$-collapsing polynomial for $S_U$ such that $f(\psi)=\psi$ and $(f_i(\varphi(i)),\varphi(i))\in\alpha_i$, i.e. $f_i(\varphi(i))\in B$. Define 
$$
\begin{gathered}
\varphi'_{j,\gamma,\delta}(k)=f_k(\varphi(k))\text{ for }k\in W\cap U\text{ and}\\
\varphi'_{j,\gamma,\delta}(k)=\psi(k)\text{ for }k\in U\setminus W.
\end{gathered}
$$
From above, $\varphi'_{j,\gamma,\delta}(i)=f_i(\varphi(i))\in B$. We need to prove that $\varphi'_{j,\gamma,\delta}\in S_U$.

Let $P|_U=(U,\{A_i:i\in U\},\vr c|_U)$ be the restriction of $P$ to $U$ and select any $(S,R)\in \vr c|_U$. Since $pr_{U\cap W}(\varphi)\in S_{U\cap W}$, there exists some $\varphi_1\in R$ such that for all $j\in S\cap V$, $\varphi_1(j)=\varphi(j)$. Since $f\in\poln{1}{S_U}$, then each parameter used in the construction of $f$ is a solution to $P|_U$. Hence the restrictions of those parameters to $S$ must be tuples in $R$. Therefore, the restriction of $f$ to $S$ is in $\poln{1}{r}$. We claim that $f|_S(\varphi_1)=\varphi'_{j,\gamma,\delta}|_S$. This breaks down into the folowing two cases:
\begin{itemize}
\item If $k\in S\cap W$, then $\varphi'_{j,\gamma,\delta}(k)=f_k(\varphi(k))=f_k(\varphi_1(k))$, while
\item If $k\in S\setminus W\subseteq U\setminus V$, then first note that $pr_{i,k}(R)=R_{\{i,k\}}=pr_{i,k}(S_U)$ by (2,3)-minimality and block-2-consistency. We know that for every $(k,\eta,\zeta)\in\vr i_{\m r}$, $(\alpha_i,\beta_i)$ can be separated from $(\eta,\zeta)$ with respect to $\m R$, and therefore for every $(k,\eta,\zeta)\in\vr i_{\m S_U}$, $(\alpha_i,\beta_i)$ can be separated from $(\eta,\zeta)$ with respect to $\m S_U$ (here we consider $\m S_U$ as a subalgebra of $\prod\limits_{k\in U}\m a_k$). Since $f$ is an $(\alpha_i,\beta_i)$-collapsing polynomial of $\m s_U$, according to Corollary~\ref{colltominsingleton} (2), $|f_k(A_k)|=1$, and since $f(\psi)=\psi$, it follows that $f_k(A_k)=\{\psi(k)\}$. Therefore, $\varphi'_{j,\gamma,\delta}(k)=\psi(k)=f_k(\varphi_1(k))$.
\end{itemize}
We have proved that $\varphi'_{j,\gamma,\delta}|_S=f|_S(\varphi_1)$ and since $f|_S\in\poln{1}{r}$ and $\varphi_1\in R$, it follows that $\varphi'_{j,\gamma,\delta}|_S\in R$, so $\varphi'_{j,\gamma,\delta}|_S$ is a solution to $P|_U$, i.e. $\varphi'_{j,\gamma,\delta}\in S_U$. Hence the family $\{\varphi'_{j,\gamma,\delta}:(j,\gamma,\delta)\in \vr i_P\}$ satisfies property $(1)$ of $\overline{\alpha}$-ensembles.

For property (2), we have proved above that, whenever $i\in W_{j,\gamma,\delta}$, then $$\varphi'_{j,\gamma,\delta}(i)\in B=[\varphi_{i,\alpha,\beta}]_{\alpha_i}=[\varphi'_{i,\alpha,\beta}]_{\alpha_i}.$$ 
On the other hand, if $k\neq i$, $k\in W_{j,\gamma,\delta}$ and $i\notin W_{j,\gamma,\delta}$, then $\alpha_k=\beta_k$. Moreover, in this case, 
$$\varphi'_{j,\gamma,\delta}(k)=\varphi_{j,\gamma,\delta}(k)\in\xi(k).$$
If $k\neq i$ and $i,k\in W_{j,\gamma,\delta}$, again we have $\alpha_k=\beta_k$ and this breaks down into the following two cases:
\begin{itemize}
\item If $k\in W_{i,\alpha_i,\beta_i}$, then $\varphi'_{j,\gamma,\delta}(k)=f_k(\varphi(k))\equiv_{\beta_k} f_k(\psi(k))=\psi(k)\in \xi(k)$, while
\item If $k\notin W_{i,\alpha_i,\beta_i}$, then $\varphi'_{j,\gamma,\delta}(k)=\psi(k)\in \xi(k)$.
\end{itemize}
Hence we have proved that $\varphi'_{j,\gamma,\delta}(i)\in B$ and whenever $k\neq i$, $\varphi'_{j,\gamma,\delta}(k)\in \xi(k)$, so family $\{\varphi'_{j,\gamma,\delta}:(j,\gamma,\delta)\in \vr i_P\}$ satisfies property $(2)$ of $\overline{\alpha}$-ensembles. Define $\xi'(i)=B$ and for $k\neq i$, $\xi'(k)=\xi(k)$.

Now suppose that $(S,R)\in\vr c$. We need to find $r\in R$ such that for all $k\in S$, $r(k)\in \xi'(k)$. If $i\notin S$, such an $r$ exists by condition $(3)$ for $\overline{\beta}$-ensembles, since $\xi'|_S=\xi|_S$. So suppose that $i\in S$. We essentially repeat the proof of (1), but with $R$ playing the role of $S_U$. Since $\vr m$ is a $\overline{\beta}$-ensemble, we know that there exists some $\mathbf{a}\in R$ such that $\mathbf{a}(j)\in\xi(j)$ for all $j\in S$. We still denote $W_{i,\alpha_i,\beta_i}$ by $W$. Since $\varphi=\varphi'_{i,\alpha_i,\beta_i}\in 
S_W$, there must exist some $\mathbf{b}\in R|_{S\cap W}$ such that $\mathbf{b}|_{S\cap W}=\varphi|_{S\cap W}$. In particular, $\mathbf{b}(i)\in B$ and for all $j\in S\cap W$, $j\neq i$, we have $\mathbf{b}(j)\in \xi(j)=\xi'(j)$.

Let $f\in\poln{1}{r}$ be the $(\alpha_i,\beta_i)$-collapsing polynomial such that $f(\mathbf{a})=\mathbf{a}$ and $f_i(\mathbf{b}(i))\in B$, where $f$ exists by Lemma~\ref{collapsingexist}. We define $\mathbf{c}=f(\mathbf{b})$. Clearly, $\mathbf{c}\in R$. Moreover,
\begin{itemize}
\item $\mathbf{c}(i)=f_i(\mathbf{b}(i))\in B=\xi'(i)$.
\item For $j\in W\cap S$ and $j\neq i$, we note that $\mathbf{b}(j),\mathbf{a}(j)\in \xi(j)=\xi'(j)$, and hence $$\mathbf{c}(j)=f_j(\mathbf{b}(j))\equiv_{\alpha_j}f_j(\mathbf{a}(j))=\mathbf{a}(j)\in\xi'(j).$$
\item Finally, for $j\in S\setminus W$, and every $(j,\gamma,\delta)\in\vr i_{\m r}$, $(\alpha_i,\beta_i)$ can be separated from $(\gamma,\delta)$ with respect to $\m r$. Since $f$ is an $(\alpha_i,\beta_i)$-collapsing polynomial of $\m r$, according to Corollary~\ref{colltominsingleton} (2), $|f_j(A_j)|=1$, and since $f(\mathbf{a})=\mathbf{a}$, it follows that $f_k(A_j)=\{\mathbf{a}(k)\}$. Therefore, $$\mathbf{c}(j)=f_j(\mathbf{b}(j))=\mathbf{a}(j)\in \xi'(j).$$
\end{itemize}
\end{proof}

\begin{thm}\label{blockminsoln}
Let $\vr t$ be a template of regular SMB algebras, $P=(V,D,\vr c)$ a (2,3)-minimal and block-2-consistent multisorted instance of $CSP(\vr t)$ and assume that all constraint relations of $P$ are nonempty. Then $\vr t$ has a solution.
\end{thm}

\begin{proof}
For all $i\in V$, let $\theta_i$ be the Rees congruence of $\m a_i$ and let $\overline{\theta}=\lb \theta_i:i\in V\rb$. First we prove that there exists a $\overline{\theta}$-ensemble.

Select any $(i,\alpha,\beta)\in\vr i_P$ and denote by $W:=W_{i,\alpha,\beta}$. Using block-2-consistency of $P$, select $f_{i,\alpha,\beta}$ which is a solution of $P|_W$ and such that for all $j\in W_{i,\alpha,\beta}$, $f_{i,\alpha,\beta}(j)\in \mathrm{min}(\m a_j)$. We will prove that $\lb f_{i,\alpha,\beta}:(i,\alpha,\beta)\in \vr i_P\rb$ is a $\overline{\theta}$-ensembie.

Clearly, $f_{i,\alpha,\beta}$ is a solution of the subinstance $P|_{W}$ such that for all $j\in W$, $f_{i,\alpha,\beta}(j)\in \mathrm{min}(\m a_j)$. Moreover, for any $(i,\alpha,\beta),(i',\alpha',\beta')\in\vr i_P$ and any $j\in W_{i,\alpha,\beta}\cap W_{i',\alpha',\beta'}$, we have that $f_{i,\alpha,\beta}(j)$ and $f_{i',\alpha',\beta'}(j)$ are both in $\mathrm{min}(\m a_j)$, so $(f_{i,\alpha,\beta}(j),f_{i',\alpha',\beta'}(j)\in \theta_j$ by the definition of the Rees congruence. Finally, if $(S,R)\in \vr c$, select some $\mathbf{a}\in\mathrm{min}(R)$. According to Lemma~\ref{minR}, for all $j\in S$, $\mathbf{a}(j)\in\mathrm{min}(\m a_j)$. Since for any $(i,\alpha,\beta)\in\vr i_P$ such that $j\in W_{i,\alpha,\beta}$, $f_{i,\alpha,\beta}(j)\in \mathrm{min}(\m a_j)$, we conclude that $(f_{i,\alpha,\beta}(j),\mathbf{a}(j))\in\gamma_j$, by the definition of the Rees congruence. As all three conditions of Definition~\ref{ensembledef} are fulfilled, we conclude that $\lb f_{i,\alpha,\beta}:(i,\alpha,\beta)\in \vr i_P\rb$ is a $\overline{\theta}$-ensemble for $P$.

Now we can successively apply Lemma~\ref{tightenensemble} to obtain $\overline{\gamma}$-ensembles for $P$ with ever smaller congruences $\gamma_i$, until inductively we get that there exists a $\overline{0}$-ensemble for $P$, where $\overline{0}:=\lb 0_{\m a_i}:i\in V\rb$. Let $\lb g_{i,\alpha,\beta}:(i,\alpha,\beta)\in \vr i_P\rb$ be a $\overline{0}$-ensemble for $P$.

But then $P$ has a solution $f$ defined by: for each $i\in V$, select $(i,\alpha,\beta)\in \vr i_P$ arbitrarily and define $f(i):=f_{i,\alpha,\beta}(i)$. According to Definition~\ref{ensembledef} (2), $f(i)$ is well-defined, while according to Definition~\ref{ensembledef} (3), $f$ is a solution of the instance $P$.
\end{proof}

\begin{cor}
Let $\vr t$ be a template of regular SMB algebras. Then $CSP(\vr t)$ is tractable.
\end{cor}

\begin{proof}
The main algorithm, $SOLVE(P)$, inputs the instance $P=(V,D,\vr c)$ of $\vr t$ and performs the following steps:
\begin{enumerate}
\item[{\bf Step 1.}] Replace $P$ with an equivalent $(2,3)$-minimal and weakly M-irredu\-cible instance, then set $M:=Size(P)$.
\item[{\bf Step 2.}] Flag all $(i,\alpha,\beta)\in \vr i_P$ as 0.
\item[{\bf Step 3.}] For any $(i,\alpha,\beta)\in \vr i_P$ with flag 0 do
\begin{enumerate}
\item[{\bf Step 3.1.}] Set $W:=\{i\}$ and flag $(i,\alpha,\beta)$ as 1.
\item[{\bf Step 3.2.}] For any $(j,\gamma,\delta)\in \vr i_P$ do
\begin{enumerate}
\item[{\bf Step 3.2.1.}] If $(\alpha,\beta)$ can't be separated from $(\gamma,\delta)$ wrt. $P$, then set $W:=W\cup\{j\}$ and flag $(j,\gamma,\delta)$ as 1.
\end{enumerate}
\item[{\bf Step 3.3.}] $CHKCOHSET(P,W)$.
\item[{\bf Step 3.4.}] If for any $i,j\in W$, $T_{\{i,j\}}=\emptyset$, output $NO$ and stop.
\item[{\bf Step 3.5.}] If for any $i,j\in W$, $|T_{\{i,j\}}|<|R_{\{i,j\}}|$ then
\begin{enumerate}
\item[{\bf Step 3.5.1.}] Update $P$ with $R_{\{i,j\}}:=T_{\{i,j\}}$ for all $i,j\in W$.
\item[{\bf Step 3.5.2.}] Perform $(2,3)$-minimality on $P$.
\item[{\bf Step 3.5.3.}] If $Size(P)<M$, then $SOLVE(P)$ and stop.
\item[{\bf Step 3.5.4.}] Go to Step 2.
\end{enumerate}
\item[{\bf Step 3.6.}] Go to the next element of $\vr i_P$ in Step 3.
\end{enumerate}
\item[{\bf Step 4.}] Output '$YES$' and stop.
\end{enumerate}
After initially tightening the instance to a $(2,3)$-minimal and weakly M-irreducible one (which can be done according to Theorem~\ref{M-irrweaker}), the algorithm $SOLVE$ next computes all coherent sets. After applying the procedure $CHKCOHSET$ to $P$ and a coherent set $W$, the algorithm either concludes there is no solution since even $P|_W$ has no solution, or the instance $P$ is tightened to an instance of smaller size which can recursively be solved, or the instance $P$ is tightened to a same-sized one (which thus must still be weakly M-irreducible, while $(2,3)$-minimality has been established in the previous step) and the whole checking has to start over, or the instance is unchanged which means that $P|_W$ has a solution through any edge. If the procedure reached Step 4, this means that for every coherent set $W$, $P|_W$ has a solution through any edge, i.e. that $P$ is block-2-consistent. Then by Theorem~\ref{blockminsoln}, $P$ has a solution.

There is a linear number of times an instance can be tightened, and each coherent set is checked exactly once if there are no tightenings, so there is at most a quadratic number of applications of $CHKCOHSET$. Moreover, checking for separation between covering pairs in $\vr i_P$ takes a constant amount of time, so the running time of the algorithm $SOLVE$ is polynomial.
\end{proof}

\section{Concluding remarks}

We have (re)proved tractability of the Constraint Satisfaction Problem over SMB algebras in two ways. Along the way, we uncovered several similarities between the Dichotomy Theorem proofs by Bulatov and by Zhuk. Firstly, Zhuk's irreducibility and Bulatov's block minimality are related properties, they serve the same purpose in the two proofs, and it is easy to find a common generalization. Secondly, there is the benefit one gets from going to a term reduct, like we did by going from an SMB algebra to its reduct which is a regular SMB algebra. This idea was taken to its logical extreme in the paper \cite{dreamteam} by L. Barto, Z. Brady, A.Bulatov, M. Kozik and D. Zhuk, which introduced the idea of minimal Taylor algebras. In minimal Taylor algebras the approaches of Bulatov and Zhuk are connected and there we may find a fertile ground to generalize out idea from Theorem~\ref{Aljareduce}. The main drawback is that we still lean on Theorem~\ref{Zhlemma}, the only proof of which uses the full proof of the Dichotomy Theorem by Zhuk. So our first open problem is

\begin{prb}\label{prb1}
Prove Theorem~\ref{Zhlemma} without resorting to the full power of Zhuk's dichotomy proof.
\end{prb}

In minimal Taylor algebras, the subset $\mathsf{umax}(\m a)$, which is the sole sink strong component of Bulatov's directed graph of the algebra, is actually a binary absorbing subuniverse. Using Theorem \ref{Zhlemma}, we may assume that in a multisorted instance, $\m a_i=\mathsf{umax}(\m a_i)$ for each domain of a variable $\m a_i$. To move further in using our proof as a template, it would be beneficial that the subuniverses $R\leq \prod\limits_{i=1}^k\m a_i$ were somehow regular. In our case of regular SMB algebras we have $\mathsf{umax}(\m a)=\mathrm{min}(\m a)$, and therefore the relations $R$ we are interested in are subuniverses of products of Mal'cev algebras, which are very well-behaved. Bulatov proves something similar for $R\leq \prod\limits_{i=1}^k\m a_i$ in the case $\m a_i=\mathsf{amax}(\m a_i)$, where $\mathsf{amax}(\m a_i)\subseteq \mathsf{umax}(\m a_i)$ is a well-behaved subset. We thus formulate our second problem, solving which (together with Problem~\ref{prb1}) would probably simplify the proof of the Dichotomy Theorem:

\begin{prb}\label{prb2}
Either prove that the multisorted instances of CSP over minimal Taylor algebras can be reduced to the $\mathsf{amax}(\m a_i)$, or develop a theory of subuniverses of products of minimal Taylor algebras all of which satisfy the condition $\m a_i=\mathsf{umax}(\m a_i)$.
\end{prb}

\section*{Acknowledgement}

We owe great thanks to Professor Andrei Bulatov for pointing out one mistake in our proofs and several other helpful suggestions and hints.

\section*{Declarations}

\subsection*{Author contributions}

This is original research, and all authors have contributed to it. The results included in this paper have not been submitted to, or published in, other journals.

\subsection*{Ethical statement}

Ethical approval is not needed for this research since it is purely theoretical.

\subsection*{Conflict of interest statement}

The authors declare that they have no conflicts of interest. 

\subsection*{Data availability statement}

Data sharing is not applicable to this article as datasets were neither generated nor analyzed.

\subsection*{Funding statement}

We have listed all sources of financial support in the first page.


\begin{thebibliography}{99}
\bibitem{barto-collapsewidth}
	Barto, L.:
	The collapse of the bounded width hierarchy.
	J. Logic Comput.
	{\bf 26},
	923--943
	(2016)

\bibitem{Barto-Valerioteconj} 
	Barto, L.:
	Finitely related algebras in congruence modular varieties have few subpowers.
	J. Eur. Math. Soc.
	{\bf 20},
	1439--1471
	(2018)
	
\bibitem{dreamteam} 	
	Barto, L., Brady, Z, Bulatov, A., Kozik, M., Zhuk, D.: 
	Unifying the Three Algebraic Approaches to the CSP via Minimal Taylor Algebras, manuscript.
	\url{https://arxiv.org/abs/2104.11808}

\bibitem{bk} 	
	Barto, L., Kozik, M.: 
	Constraint satisfaction problems solvable by local consistency methods.
	Journal of the ACM 
	{\bf 61} 1:03,
	19 pp., 
	(2014)

\bibitem{BIMMVW} 
	Berman, J., Idziak, P., Markovi\'c, P., McKenzie, R., Valeriote, M., Willard, R.:
	Varieties with few subalgebras of powers.
	Trans. Amer. Math. Soc.
	{\bf 362}
	1445--1473
	(2010)

\bibitem{3el} Bulatov, A.:
	A dichotomy theorem for constraints on a three-element set.
	J. ACM 
	{\bf 53},
	66--120
	(2006)

\bibitem{BulatovMalcev}
	Bulatov, A.:
	Complexity of Maltsev Constraints. 
	[Russian]
	Algebra i Logika, 
	{\bf 45}
	655--686
	(2006) 

\bibitem{Bulatovcons} 
	Bulatov, A.:
	Complexity of Conservative Constraint Satisfaction Problems.
	ACM Trans. Comput. Logic 
	{\bf 12} 
	Article no. 24
	66pp.
	(2011)

\bibitem{BuSMB} 
	Bulatov, A.:
	Constraint Satisfaction Problems over semilattice block Mal'tsev algebras,
	Information and Computation {\bf 268}
	Article no. 104437,
	14 pp., 
	(2019)
	
\bibitem{BulatovGraph1} Bulatov, A.:
	Local structure of idempotent algebras I, manuscript.
	\url{https://arxiv.org/abs/2006.09599}

\bibitem{BulatovGraph2}
	Bulatov, A.:
	Local structure of idempotent algebras II, manuscript.
	\url{https://arxiv.org/abs/2006.10239}

\bibitem{BulatovGraph3}
	Bulatov, A.:
	Graphs of relational structures: restricted types.
	In: Proc. 31th IEEE Ann. Symp. on Logic in Comput. Sci. (LICS)
	New York, USA, 2016,
	pp. 642--651.
	doi: 10.1145/2933575.2933604
	IEEE Computer Society,
	ISBN:978-1-4503-4391-6

\bibitem{Buconf} Bulatov, A.: 
	A dichotomy theorem for nonuniform CSPs
	In: Proc. 2017 IEEE 58th Annual Symp. on Foundations of Comput. Sci. (FOCS),
	Berkeley, CA (USA, October 2017),
	pp. 319-330, doi: 10.1109/FOCS.2017.37,
	IEEE Computer Society Conference Publishing Service
	Los Alamitos, CA.

\bibitem{Budich}
	Bulatov, A.:
	A dichotomy theorem for nonuniform CSPs, manuscript.
	\url{https://arxiv.org/abs/1703.03021}

\bibitem{Bulect}
	Bulatov, A.:
	Constraint Satisfaction Problems: Complexity and Algorithms, 
	In: Klein, S., Martín-Vide, C., Shapira, D. (eds) 
	Language and Automata Theory and Applications. LATA 2018. 
	Ramat Gan, (Israel, April 2018),	
	pp. 1--25, doi: 10.1007/978-3-319-77313-1
	Lecture Notes in Comput. Sci., 
	vol 10792, (2018)
	Springer, New York.

\bibitem{BulatovDalmau}
	Bulatov A., Dalmau, V.:
	Mal'tsev constraints are tractable.
	SIAM J. Comput.
	{\bf 36}
	16--27
	(2006)

\bibitem{BJK}
	Bulatov, A., Jeavons P., Krokhin, A.,
	Classifying the complexity of constraints using finite algebras.
	SIAM J. Comp.
	{\bf 34}
	720--742
	(2005)

\bibitem{BDJN}
	Bulin, J., Deli\'c, D., Jackson M., Niven, T.:
	A finer reduction of constraint problems to digraphs.
	Log. Meth. Comput. Sci.
	{\bf 11}, 
	Article no. 18,
	33 pp.
	(2015)

\bibitem{burris-sank} 
	Burris, S., Sankappanavar, H.P.:
	A course in universal algebra.
	Graduate Texts in Mathematics, vol. 78.
	Springer,
	New York
	(1981)

\bibitem{SMB1} 
	\Dj api\'{c}, P., Markovi\'{c}, P., McKenzie, R., Proki\'{c} A.:
	SMB Algebras I: On the variety of SMB algebras.
	manuscript.

\bibitem{FV}
	Feder, T., Vardi, M. Y.:
	The computational structure of monotone monadic SNP and constraint satisfaction: A study through Datalog and group theory. 
	SIAM J. Comput.
	{\bf 28}
	57--104 (1999)

\bibitem{tct} 
	Hobby, D., McKenzie, R.:
	The structure of finite algebras.
	Contemporary Mathematics, vol. 76.
	American Mathematical Society, Providence (1988)

\bibitem{IMMVW}
 	Idziak, P., Markovi\'{c}, P., McKenzie, R., Valeriote, M., Willard, R.:
	Tractability and learnability arising from algebra with few subpowers. 
	SIAM J. Comput.
	{\bf 39}, 
	3023--3037 (2010)

\bibitem{J}
	Jeavons, P. G.:
	On the algebraic structure of combinatorial problems.
	Theor. Comp. Sci. {\bf 200}
	185--204
	(1998)

\bibitem{K}
	Kun, G.:
	Constraints, {MMSNP}, and Expander Relational Structures
	Combinatorica
	{\bf 33}
	335--347
	(2013)

\bibitem{RP}
	Markovi\'{c}, P., McKenzie, R.:
	On the Constraint Satisfaction Problem over semilattices of Mal'cev blocks.
	(early version untitled),
	manuscript.

\bibitem{MMontop}
	M. Mar\'{o}ti,
	Maltsev on top.
	manuscript,
	\url{http://www.math.u-szeged.hu/~mmaroti/pdf/200x%20Maltsev%20on%20top.pdf}.

\bibitem{M2}
	Mar\'{o}ti, M.:
	Tree on top of Maltsev, manuscript.
	\url{http://www.math.u-szeged.hu/~mmaroti/pdf/200x%20Tree%20on%20top%20of%20Maltsev.pdf}

\bibitem{MM}
	Mar\'{o}ti, M., McKenzie, R. N.:
	Existence theorems for weakly symmetric operations.
	Algebra Universalis
	{\bf 59}
	463--489
	(2008)

\bibitem{alvin} 
	McKenzie, R., McNulty, G., Taylor, W.:
	Algebras, lattices, varieties. Vol. I.
	The Wadsworth \& Brooks/Cole Mathematics Series,
	Wadsworth \& Brooks/Cole Advanced Books \& Software, Monterey (1987)

\bibitem{Zhconf} Zhuk, D.: 
	A proof of CSP dichotomy conjecture, 
	In: Proc. 2017 IEEE 58th Annual Symp. on Foundations of Comput. Sci. (FOCS),
	Berkeley, CA (USA, October 2017),
	pp. 331-342, doi: 10.1109/FOCS.2017.38,
	IEEE Computer Society Conference Publishing Service
	Los Alamitos, CA.

\bibitem{Zhdich} Zhuk, D.:
	A Proof of the CSP Dichotomy Conjecture,
	Journal of the ACM 
	{\bf 67} 5:30,
	78 pp., 
	(2020)
\end{thebibliography}
\end{document}